\documentclass[11 pt, a4paper]{article}
\usepackage[hmargin=1.25in,vmargin=1.25in]{geometry}
\usepackage{amsmath,amssymb,amsthm,mathrsfs,amsfonts,mathtools}
\usepackage{tabularx,siunitx,booktabs,threeparttable,multirow}
\usepackage[dvipsnames,svgnames, x11names]{xcolor}% Must load before tikz
\usepackage[normalem]{ulem}
\usepackage[titletoc]{appendix}                   % For appendices
\usepackage{enumitem,setspace,titlesec,natbib}
\usepackage[hyperindex,breaklinks,hypertexnames=false]{hyperref}      % hyperindex,breaklinks needed for arXiv
\usepackage[open,openlevel=2,numbered]{bookmark}  % Must load after hyperref
\usepackage{algorithm}
\usepackage{algorithmic}
\usepackage{bm}
\usepackage{microtype}
\usepackage{tikz,pgfplots} % Graph packages
\pgfplotsset{compat=1.11,compat/show suggested version=false}
\pgfplotsset{table/search path={.,Empiricaloutput}}
\usetikzlibrary{arrows,plotmarks}
\usepgfplotslibrary{groupplots}
\usetikzlibrary{intersections, backgrounds}
\usepgfplotslibrary{fillbetween}
\makeatletter
\renewcommand\paragraph{\@startsection{paragraph}{4}{\z@}%
            {-2.5ex\@plus -1ex \@minus -.25ex}% Space before sectioning title
            {1.25ex \@plus .25ex}             % Space after sectioning title
            {\normalfont\normalsize\bfseries}}
\makeatother
\hypersetup{colorlinks, linkcolor = {NavyBlue}, citecolor = {NavyBlue}, urlcolor ={NavyBlue}}

\let\AIorigcitet\citet
\let\AIorigcitep\citep
\let\AIorigcite\cite
\let\AIorigref\ref
\let\AIorigeqref\eqref
\NewDocumentCommand{\AIcitet}{o m}{\allowbreak\mbox{\IfNoValueTF{#1}{\AIorigcitet{#2}}{\AIorigcitet[#1]{#2}}}}
\NewDocumentCommand{\AIcitep}{o m}{\allowbreak\mbox{\IfNoValueTF{#1}{\AIorigcitep{#2}}{\AIorigcitep[#1]{#2}}}}
\NewDocumentCommand{\AIcite}{o m}{\allowbreak\mbox{\IfNoValueTF{#1}{\AIorigcite{#2}}{\AIorigcite[#1]{#2}}}}
\DeclareRobustCommand{\AIref}[1]{\mbox{\AIorigref{#1}}}
\DeclareRobustCommand{\AIeqref}[1]{\mbox{\AIorigeqref{#1}}}
\DeclareRobustCommand{\AIold}[1]{{\color{BrickRed}\let\citet\AIcitet\let\citep\AIcitep\let\cite\AIcite\let\ref\AIref\let\eqref\AIeqref\sout{#1}}}

\newtheorem{thm}{Theorem}[section]
\newtheorem{cor}{Corollary}[section]
\newtheorem{lem}{Lemma}[section]
\newtheorem{pro}{Proposition}[section]
\newtheorem{ass}{Assumption}[section]
\theoremstyle{definition}
\newtheorem{defn}{Definition}[section]

\newtheorem{rem}{Remark}[section]

\begin{document}
\pdfbookmark[1]{Title}{title}
\title{Debiased Inference via Posterior Corrections for High-Dimensional Regression\thanks{We thank an editor and anonymous referees for comments that helped greatly improve this paper.}}
\author{Qihui Chen\thanks{School of Management and Economics and Shenzhen Finance Institute, The Chinese University of Hong Kong, Shenzhen (CUHK-Shenzhen); qihuichen@cuhk.edu.cn}\\CUHK-Shenzhen
\and
Zheng Fang\thanks{Department of Economics, Emory University; zheng.fang@emory.edu} \\Emory University
\and
Ruixuan Liu\thanks{CUHK Business School, Chinese University of Hong Kong; ruixuanliu@cuhk.edu.hk} \\CUHK}

%\author{
%	Qihui Chen\\ School of Management and Economics \\CUHK-Shenzhen\\ qihuichen@cuhk.edu.cn
%	\and
%	Zheng Fang \\ Department of Economics \\ Emory University\\ zheng.fang@emory.edu
%	\and
%	Ruixuan Liu\\ CUHK Business School\\  CUHK\\ ruixuanliu@cuhk.edu.hk}
\date{\today}
\maketitle
\begin{abstract} 
	There has been significant progress in Bayesian inference based on sparsity-inducing priors, such as spike-and-slab and horseshoe priors, for high-dimensional regression. The resulting posteriors, however, in general do not possess desirable frequentist properties, and the credible sets thus cannot serve as valid confidence sets even asymptotically. We introduce a novel debiasing approach that corrects the bias for the entire (possibly approximate) Bayesian posterior distribution. We establish new Bernstein-von Mises theorems that guarantee the frequentist validity of the debiased law. We demonstrate the practical performance of our proposal through Monte Carlo simulations and two empirical applications in economics.
\end{abstract}

\begin{center}
\end{center}

\newpage

\section{Introduction}
% Applied researchers now routinely work with regression models that feature a large number of covariates. 

A primary inferential goal in econometrics is to estimate the ceteris paribus effect of a specific variable while controlling for other variables \citep{BelloniChernoHansen2013HD, BelloniChernozhukovCherverikovWei2018Many}. The prevailing practice interprets the coefficient on a particular regressor as a causal effect conditional on the included controls. As the plausibility of unconfoundedness is often argued using a large set of covariates, practitioners have increasingly adopted high-dimensional regression models. This setting has been extensively studied, predominantly using frequentist methods. Bayesian inference, on the other hand, has long been valued for its coherent framework for handling uncertainty in statistical analysis. As highlighted by \citet{Rubin1984Applied}, Bayesian methods provide direct answers to many empirical questions by quantifying uncertainty about unknown parameters conditional on the data.\footnote{In discussing the commonly-used rule-of-thumb confidence interval, \cite{Rubin1984Applied} states that ``the interval is-at least in my experience-nearly always interpreted Bayesianly, that is, as providing a fixed observed interval in which the unknown (parameter of interest) $\mu$ lies with 95\% probability.'' As advocated by \cite{Imbens2021Pvalue}, it may even be preferable to adopt Bayesian inference in cases where Bayesian and frequentist procedures lead to different conclusions.} Bayesian posterior distributions unify estimation and inference in a way that aligns well with the needs of empirical work.

%\textcolor{red}{[Do not fully understand Rubin's argument here...]}

In addressing the challenges posed by high dimensionality, recent methodological advances have substantially expanded the Bayesian toolkit for regression models with many covariates. Notable progress includes the development of spike-and-slab priors \citep{MitchellBeauchamp_BayesianVariable_1988,GeorgeMcCulloch_VariableSelection_1993} and horseshoe priors \citep{CarvalhoPolsonScott2010Horseshoe}, as well as scalable approximate Bayesian inference techniques \citep{RaySzabo2022Variational}. These innovations have attracted a growing interest in economics, as illustrated by \citet{GiannoneLenzaPrimiceri2021Sparsity}. Although much of this literature focuses on model selection and estimation, an increasingly active strand studies Bayesian inference for a low-dimensional parameter of interest in the presence of high-dimensional controls \citep{HahnCarvalhoPuelzHe2018Regularization,Yang2019HDL,Castilloetal2024VB,DiTragliaLiu2025BML,Walker2026BvMPLM}. These studies address regularization bias through reparametrizations of priors tailored to the parameter of interest. We further discuss their relation to our approach with details below. Absent such adjustments, standard regularized Bayesian procedures may produce severely biased estimates and undercovered credible sets\footnote{For instance, our Monte Carlo experiments reveal that the finite-sample frequentist coverage of the nominal $95\%$ Bayesian credible set for a regression coefficient drops to just $2.8\%$ when the sample size is $100$ and the covariate dimension is $200$, under simulation design S2 in Section \ref{Sec: Monte Carlo}.}.

In this paper, we develop a novel inferential procedure for high-dimensional linear regression that introduces a debiasing step for the Bayesian posterior distribution of the parameter of interest. Building on the concept of debiased point estimators, which is well established in the frequentist literature, we instead debias the entire posterior distribution, yielding a new debiased inferential approach. The core theoretical contribution of our work is the establishment of new Bernstein-von Mises (BvM) results for the debiased procedure. These results formally justify the asymptotic normality of the debiased law and ensure that intervals constructed from its quantiles attain asymptotically correct frequentist coverage. Importantly, the resulting confidence intervals can match the inferential performance of debiased frequentist methods \citep{VandeGeer_OnAsymptotically_2014} in high-dimensional regimes. Our framework is general. It is agnostic to the choice of prior for the regression coefficients and permits flexible selection of pilot estimators for the precision matrix. This flexibility allows the procedure to be tailored to a wide range of empirical applications without sacrificing computational scalability.

Another key advantage of our approach is its compatibility with computationally efficient approximate Bayesian inference, such as variational Bayes \citep{RaySzabo2022Variational}. The specific implementations we present in Section~\ref{Sec:42} illustrate how the method can be practically adapted to high-dimensional settings. Traditional spike-and-slab priors, while theoretically appealing, often entail substantial computational cost due to the reliance on Markov chain Monte Carlo (MCMC) sampling. In contrast, variational Bayes delivers approximate posteriors in seconds while still attaining optimal contraction rates in high-dimensional settings. However, without debiasing, credible sets constructed from such approximate posteriors are not guaranteed to achieve correct frequentist coverage. By incorporating our debiasing procedure, the variational Bayes posterior becomes a building block for valid inference.

Our work is related to and motivated by a common inferential problem across frequentist and Bayesian approaches to high dimensional regression, as regularization introduces bias that invalidates conventional inference. In the frequentist literature, the LASSO is a leading method for high dimensional regression, whereas the shrinkage induced by its $\ell_1$ penalty complicates inference. This has led to methods that add a correction term to a LASSO-type point estimator \citep{JavanmardMontanari2014CI,VandeGeer_OnAsymptotically_2014,ZhangZhang2014CI}. Our approach applies the debiasing step at the distributional level (exact or approximate posterior distributions) and therefore requires new technical arguments to establish the large-sample properties. 

Simulation studies show our procedure is competitive with related debiased frequentist procedures in terms of estimation and coverage accuracy. In particular, when covariates are correlated and coefficients are relatively large (corresponding to the coefficients equal to one or two in our simulation), our method offers pronounced advantages with better coverage properties, smaller bias,  and lower mean squared error. % We note that our focus on bias instead of efficiency is consistent with the frequentist debiased machine learning 

In the Bayesian literature, early contributions by \citet{Ghosal1999HDL} and \citet{Bontemps2011BvM} established posterior asymptotic normality in Gaussian linear regression when the number of covariates grows slowly with and remains smaller than the sample size. Sparsity-inducing priors such as the spike-and-slab prior, which is often regarded as the benchmark in sparse Bayesian linear regression \citep{PiironenVehtari2017Sparsity,RaySzabo2022Variational}, permit estimation of much high dimensional models but may induce biased posterior summaries. \citet{Castilloetal_BayesianSparse_2015} analyzed Bayesian posteriors under spike-and-slab priors, establishing asymptotic mixed normality under strong signal conditions. \citet{Yang2019HDL} introduced a debiasing approach based on reparameterization of the prior specification, assigning a data-dependent Gaussian prior to the parameter of interest. See also \citet{Castilloetal2024VB} that extends this type of reparameterization to spike-and-slab prior. Our approach instead performs post-processing, which is inspired by the one-step correction from \citet{YiuFongHolmesRousseau2025Semiparametric}. Such post-processing requires no modification of existing posterior computation and can therefore be combined with a wide range of exact or approximate methods such as variational Bayes. This flexibility comes with a qualification. The reparameterization approach embeds debiasing in a coherent Bayesian model which enjoys decision-theoretic optimality. Such optimality, however, does not automatically extend to decisions derived from a post-processed distribution. We therefore view the two approaches as complementary rather than competing.

%\textcolor{blue}{Liu: I suggest that we shorten the next sentences and put it into the next paragraph for comparing with DiTraglia and Liu. What do you think?}This flexibility, however, comes with a decision-theoretic qualification: under a coherent Bayesian model, minimizing posterior expected loss yields a Bayes-optimal action for general loss functions, whereas this optimality guarantee does not automatically extend to actions derived from a post-processed distribution. We therefore view the two approaches as complementary rather than competing.

More broadly, our paper connects with developments in semiparametric Bayesian inference and correction methods for posterior or prior distributions \citep{RayVaart2020Semi,BreunigLiuYu2022DR,BreunigLiuYu2024DiD,BreunigLiuYu2025Bart,YiuFongHolmesRousseau2025Semiparametric,Walker2026BvMPLM}. These works depart from naive plug-in principles to achieve robust large-sample properties. We contribute to this literature by explicitly addressing high-dimensional regression with sparsity-inducing priors and by showing that debiasing can be performed on approximate posteriors. Theoretically, we also establish BvM results for a growing subset of coefficients, which is a novel contribution in this context. Building on the reparametrization/orthogonalization approach in \citet{Walker2026BvMPLM}, \citet{DiTragliaLiu2025BML} work with a system of reduced-form regressions and recover the parameter of interest from the posterior of the corresponding reduced-form covariance matrix. \cite{DiTragliaLiu2025BML} emphasize the regularization-induced confounding in high-dimensional regression and provide the full Bayesian interpretation to the linear regression under conditional homoskedasticity. Note that the number of covariates is required $p = o(n)$ for consistency and $p=o(\sqrt n)$ for establishing the BvM therein. By contrast, we allow for $p$ to be (potentially much) larger than $n$ and meanwhile also permit inference on a subvector whose dimension may grow with $n$. We reiterate that our procedure is of post-processing nature that is suitable for alleviating bias but does not necessarily inherit the Bayesian decision-theoretic optimality. We further elaborate how we achieve the orthogonalization and the conceptual overlap of our methodology with \cite{Walker2026BvMPLM} and \cite{DiTragliaLiu2025BML} in Section \ref{Sec:43}.

%Their approach relies on priors for reduced-form covariance matrices to remain consistent with the likelihood principle. However, when $p$ is large relative to $n$, practical computation of posteriors for high-dimensional covariance matrices remains challenging. \textcolor{red}{\citet{DiTragliaLiu2025BML} only involve $2\times2$ reduced-form error covariance matrix. Delete it?}

% \textcolor{red}{She assumes that the correlation between the regressor of interest and controls is extremely small; too restrictive in applications. In such cases, priors on controls do not transit to the posterior of the parameter of interest. In other words, regularization on other parameters do not introduce the bias for the estimator of interest.}

%The second one is the debiased LASSO, which is based on debiasing a Bayesian point estimator and resort to the normal critical value in forming the confidence interval. In comparison, we utilize the information from the entire posterior distribution in our debiasing procedure. Indeed, it provides additional notable gain in terms of the point estimation.

The remainder of the paper is organized as follows. Section~\ref{Sec: Debiased Bayesian inference} introduces our inferential procedure and discusses implementation details, while Section~\ref{Sec: Applications} presents two empirical applications. In Section~\ref{Sec: 4}, we develop our main theoretical results under high-level conditions, which are then verified under more primitive assumptions. Section~\ref{Sec: Monte Carlo} examines the finite-sample performance of our procedure through Monte Carlo simulations. Section~\ref{Sec: Conclusion} concludes. All proofs and supporting results are relegated to the appendices, which also contain additional simulation evidence.

\section{The Posterior Correction Procedure}\label{Sec: Debiased Bayesian inference}
Consider the prototypical linear regression model:
\begin{align}\label{Eqn:Model}
	Y_i=X_i^\intercal\beta_0+\varepsilon_i,\quad \mathbb{E}_0[X_{i}\varepsilon_i]=0, \quad i = 1,\ldots, n,
\end{align}
%\varepsilon_i\mid X_i\sim N(0,\sigma^2),
where $\{Y_i,X_i\}_{i=1}^n$ are independent and identically distributed (\textit{i.i.d.}) and $\mathbb{E}_0[\cdot]$ is the population expectation.
%\footnote{If the Gaussian assumption on $\varepsilon_i$ is omitted, the framework can be interpreted as quasi-Bayesian; see related discussions in the literature. Our asymptotic analysis does not rely on this assumption.}
The covariate $X_i$ is a $p$-dimensional vector, with $p$ potentially exceeding the sample size $n$.  The true regression coefficient vector $\beta_0$ is assumed to be sparse so that only a small subset of regressors have nonzero effects\footnote{There are limitations of sparsity-based modelling as pointed out by \cite{GiannoneLenzaPrimiceri2021Sparsity}. In addition, sparsity is generally not preserved under linear transformations  \citep{KolesarMullerRoelsgaard2023Fragility}. Our starting point from the sparse model reflects the availability of a relatively extensive theory advance rather than an intrinsic requirement of our debiasing method.}. We denote the population covariance matrix of $X_i$ by $\Omega_0:=\mathbb{E}_0[X_iX_i^\intercal]$, and its inverse, the precision matrix, by $\Theta_0:=\Omega_0^{-1}$. 

Since the seminal work of \citet{Tibshirani_Lasso_1996}, the LASSO has become the default method for estimating the high-dimensional linear regression model, owing to its ability to perform model selection and coefficient estimation simultaneously. It is well known that the LASSO estimator, as a penalized least-squares estimator, can be interpreted as a Bayesian point estimator---specifically, the posterior mode under independent Laplace priors on the regression coefficients (see Remark \ref{Rem:DebiasPoint}). However, when one considers the entire posterior distribution rather than only its mode, the Laplace prior alone fails to induce sparsity. As shown in Theorem 7 of \citet{Castilloetal_BayesianSparse_2015}, the full Bayesian posterior corresponding to the Laplace prior lacks desirable asymptotic properties.\footnote{\citet{Castilloetal_BayesianSparse_2015} remark that “the LASSO is essentially non-Bayesian, in the sense that the corresponding full posterior distribution is a useless object” (p.1988). In practice, sampled posterior distributions of regression coefficients under the Laplace prior are indeed non-sparse \citep{Baietal_ReviewSpikeSlab_2021}.}

Our method is inspired by the frequentist literature on debiased inference \citep{ZhangZhang2014CI}, which applies a debiasing step to the LASSO estimator. Let $\hat{\beta}^{\text{Pilot}}_n$ denote any pilot estimator of $\beta_0$, and let $\hat{\Theta}_n$ be a pilot estimator of the precision matrix $\Theta_0$. The celebrated debiased inference procedure is based on the following construction:
\begin{align}
\hat{\beta}_n^{\text{Debias}}=\hat{\beta}_n^{\text{Pilot}}+ \hat{\Theta}_n\left[\frac{1}{n}\sum_{i=1}^{n}X_i(Y_i-X_i^\intercal\hat{\beta}_n^{\text{Pilot}})\right].
\end{align}
It is well established that the individual coordinates of the debiased estimator $\hat{\beta}_n^{\text{Debias}}$ are asymptotically normal, enabling valid statistical inference.

Our debiased inference procedure departs from the aforementioned frequentist approach in several key respects. The Bayesian modeling perspective treats the unknown regression coefficient as the random vector $\beta$. We begin by assigning a prior that may be sparsity-inducing, such as a spike-and-slab prior or a horseshoe-type prior, and computing the (approximate) initial posterior distribution of $\beta$, denoted by $\Pi_{\beta}(\,\cdot\,|\{Y_i,X_i\}_{i=1}^n)$. Under sparsity, the resulting posterior achieves (near-)optimal contraction rates \citep{Castilloetal_BayesianSparse_2015,GaoVaartZhou2020General,SongLiang_NearlyOptimal_2023}. We then debias the entire posterior distribution using a correction term analogous to the frequentist construction. In this step, however, we replace the uniform weights $1/n$ with Bayesian bootstrap weights. Specifically, the Bayesian bootstrap weights are defined by normalizing independent standard exponential random variables (independent of $\{Y_i,X_i\}_{i=1}^n$ and the prior of $\beta$):
\begin{align}\label{Eqn: Bayesian weights}
W_{ni}=\frac{\omega_i}{\sum_{j=1}^n \omega_j},\quad \{\omega_j\}_{j=1}^{n} \stackrel{i.i.d.}{\sim} \textup{Exp}(1)\quad, i=1,\ldots, n.
\end{align}
In summary, we study the following corrected draw:
\begin{align}\label{Eqn:DebiasBayes}
\tilde{\beta}
= \beta + \hat{\Theta}_n \left[ \sum_{i=1}^n W_{ni} X_i( Y_i - X_i^\intercal \beta ) \right],
\end{align}
where $\beta\sim \Pi_{\beta}(\,\cdot\,|\{Y_i,X_i\}_{i=1}^n)$ independently of the weights \(\{W_{ni}\}\), and $\hat{\Theta}_n$ is a proper estimator of the precision matrix. We call the conditional distribution of $\tilde{\beta}$ given $\{Y_i,X_i\}_{i=1}^n$ the debiased law. Algorithm~\ref{Algorithm} describes how to generate draws from it.
\begin{algorithm}[H]
	\caption{Posterior-Correction Procedure}\label{Algorithm}
	\begin{algorithmic}
		\STATE \textbf{Input:} Data $\{Y_i,X_i\}_{i=1}^n$; number of draws $B$.
        
		\STATE \textbf{Initial Posterior:} $\Pi_{\beta}(\,\cdot\,|\{Y_i,X_i\}_{i=1}^n)$. 
        
        \STATE \textbf{Pilot Estimation:}  Choose a pilot estimator $\hat{\Theta}_n$ of the precision matrix $\Theta_0$.
        
		\STATE \textbf{Posterior Computation and Correction:}
		\FOR{$b=1,\ldots, B$}
		\STATE  (a) Draw $\beta^b$ from the initial posterior $\Pi_{\beta}(\,\cdot\,|\{Y_i,X_i\}_{i=1}^n)$.
		\STATE (b) Generate Bayesian bootstrap weights:
        \[W^b_{ni}=\frac{\omega^b_i}{\sum_{j=1}^n \omega^b_j},\quad \{\omega_j^b\}_{j=1}^{n} \stackrel{i.i.d.}{\sim} \textup{Exp}(1),  \quad i=1,\ldots,n.\]
		\STATE (c) Compute the debiased draw:
		\[\tilde{\beta}^b=\beta^b+ \hat{\Theta}_n\left[\sum_{i=1}^{n}W_{ni}^bX_i(Y_i-X_i^\intercal\beta^b)\right].\]
		\ENDFOR
		
		\STATE \textbf{Output: $\{\tilde{\beta}^{b}:b=1,\ldots,B\}$}
	\end{algorithmic}
\end{algorithm}
 
Several remarks are in order. First, replacing the uniform weights $1/n$ with the Bayesian bootstrap weights $W_{ni}^b$ is essential. To see this, consider the case $p<n$, where one may take $\hat{\Theta}_n =(\sum_{i=1}^{n} X_{i}X_{i}^\intercal/n)^{-1}$. If instead $W_{ni}^b=1/n$ for every $i$, $\tilde{\beta}^b$ simplifies to $(\sum_{i=1}^{n}X_{i}X_{i}^{\intercal})^{-1}\sum_{i=1}^{n}X_{i}Y_{i}$, which exhibits no posterior uncertainty. Second, the debiasing step introduces no additional computational burden. Since the Bayesian bootstrap weights are independent of the posterior draws of $\beta$, steps (a)--(b) of Algorithm \ref{Algorithm} can be parallelized efficiently. Moreover, the one-time computation of $\hat{\Theta}_n$ can be carried out simultaneously, which remains computationally efficient as known in the frequentist literature (see Section \ref{Sec:43}). Third, step (c) of Algorithm~\ref{Algorithm} admits a natural interpretation as a Bayesian version of the residual bootstrap. Specifically, consider the case $p<n$ with $\hat{\Theta}_n =(\sum_{i=1}^{n} X_{i}X_{i}^\intercal/n)^{-1}$. For a given posterior draw $\beta^b$, define bootstrap residuals $\varepsilon_i^{b\ast}$ as $\varepsilon_i^{b}=Y_i-X_i^\intercal \beta^b$ weighted by the Bayesian bootstrap weights $nW_{ni}^b$. Construct the pseudo responses $Y_{i}^{b\ast} = X_i^{\intercal}\beta^b + \varepsilon_i^{b\ast}$ and regress $Y_{i}^{b\ast}$ on $X_i$. The resulting estimator coincides with $\tilde{\beta}^{b}$. Thus, each debiased draw $\tilde\beta^b$ can be viewed as the regression coefficient obtained from a ``Bayesian residual bootstrap,'' where residuals are reweighted via Bayesian bootstrap rather than resampled by the empirical bootstrap. While other bootstrap weights can be employed, the Bayesian bootstrap weights provide a natural Bayesian interpretation (see Section \ref{Sec:43}).

 The posterior-correction procedure provides point estimation and uncertainty quantification simultaneously. For $j=1,\ldots,p$, let $\beta_{0,j}$ and $\tilde{\beta}^{b}_j$ be the $j$th coordinate of $\beta_{0}$ and $\tilde{\beta}^{b}$ respectively. For $0<\alpha<1$, the central $100(1-\alpha)\%$ quantile interval for the regression coefficient $\beta_{0,j}$ is defined as:
\begin{align}\label{Eqn:CS}
	\mathcal{C}_{n,j}(\alpha,B) 
	= \big[\,\hat{c}_{n,j}(\alpha/2,B),~~\hat{c}_{n,j}(1-\alpha/2,B)\,\big],
\end{align}
where $\hat{c}_{n,j}(\alpha,B)$ denotes the $\alpha$th quantile of the corrected draws $\{\tilde{\beta}_j^b:b=1,\ldots,B\}$. The corresponding point estimator is their average, $\bar{\tilde \beta}_j = \sum_{b=1}^B \tilde{\beta}_j^b/B$. Section~\ref{Sec: 4} establishes that these quantile intervals have asymptotically correct frequentist coverage.

For simultaneous inference on, say, the first $J$ coefficients, where $1\le J\le p$, we use a frequentist centering estimator $\hat\beta_n$. For $0<\alpha<1$, a Monte Carlo approximation to the simultaneous confidence band consists of the intervals
\begin{align}\label{Eqn:SimultaneousBand}
    \mathcal C_{n,j}(\alpha,B)
    = \left[\hat\beta_{n,j}-\frac{\hat c_n(1-\alpha,B)}{\sqrt n},\,
    \hat\beta_{n,j}+\frac{\hat c_n(1-\alpha,B)}{\sqrt n}\right],
    \quad j=1,\ldots,J,
\end{align}
where $\hat c_n(1-\alpha,B)$ denotes the empirical $(1-\alpha)$ quantile of the simulated maxima
\begin{align*}
    \max_{1\le j\le J}\left|\sqrt n\{\tilde\beta_j^b-\hat\beta_{n,j}\}\right|,
    \qquad b=1,\ldots,B.
\end{align*}
Whereas the marginal intervals in \eqref{Eqn:CS} are computed directly from quantiles of the corrected marginal draws, the simultaneous band uses the center $\hat\beta_n$ to form a scalar maximum so as to compute the critical values. Thus, a centering estimator is required for simultaneous inference, but is not needed to compute the marginal intervals. Corollary~\ref{Cor:BvM2} establishes its asymptotic coverage under suitable anti-concentration conditions.

\section{Insights from Empirical Applications}\label{Sec: Applications}
 Our paper aims not only to present a comprehensive theoretical development, but also to demonstrate the practical benefits through concrete empirical applications. Before presenting the theoretical properties of our debiased inference method, we first illustrate via two empirical examples adapted from \citet{BelloniChernoHansen2013HD,BelloniChernozhukovHansen2014HighATE,GiannoneLenzaPrimiceri2021Sparsity}. In both cases, our goal is not to determine the overall sparsity of the regressors, but rather to conduct inference on a particular regression coefficient---an objective well suited to our inference procedure.

For each application, we begin by presenting Bayesian inference results for the coefficient of interest using both a spike-and-slab prior and a horseshoe-type prior. The prior specifications and hyperparameter settings follow those used in our Monte Carlo simulations. For the spike-and-slab prior, we obtain the posterior distribution via the variational Bayesian approximation described in \citet{RaySzabo2022Variational}, drawing 8,000 samples from the resulting approximate posterior. The posterior under the horseshoe prior is generated using the MCMC algorithm of \citet{KimLeeGupta2020BayesianSC}, with a total of 16,000 draws and the first 8,000 discarded as burn-in. In both applications, the spike-and-slab posterior collapses to a point mass at zero. Under the horseshoe prior, the posterior exhibits greater dispersion, though in the second application the density remains tightly concentrated near zero. In both applications, the 95\% credible intervals under both priors contain zero.

%\textcolor{red}{Also include R packages above?}

To implement the debiasing adjustment, we estimate the precision matrix using the nodewise LASSO regression method of \citet{VandeGeer_OnAsymptotically_2014}, adopting the default tuning parameters from the R package \texttt{hdi} \citep{DBMM2015hdi}. For the Bayesian bootstrap, we generate 8,000 additional sets of bootstrap weights in parallel with the posterior draws of the regression coefficients. After applying the debiasing step, the debiased laws become approximately Gaussian, providing empirical support for our theory. Overall, the debiased inference reveals a significant effect in the first application but no significant effect in the second.

\subsection{Determinants of Economic Growth}
The influential work of \citet{Barro1991Growth} initiated a long-standing debate over what drives long-term economic growth across countries. Researchers have identified numerous potential predictors, many of which are included in the original data set compiled by \citet{Barro1991Growth}. In line with \citet{BelloniChernozhukovHansen2013InferHigh}, we use this data set to analyze average GDP growth between 1960 and 1985 for 90 countries. The data set contains 60 candidate predictors, covering pre-1960 measures of a wide range of socio-economic, institutional, and geographical factors. Under the initial posteriors, the 95\% credible intervals for the coefficient on initial GDP contain zero, whereas our 95\% confidence intervals exclude zero. Thus, the corrected procedure rejects the null hypothesis of a zero coefficient at the 5\% level and indicates a negative association between initial GDP and subsequent growth. This suggests the presence of catch-up effect, everything else equal.

%\textcolor{red}{Zheng: may want to change the legends using different shapes/markers for readability of printed papers.}

\begin{figure}[htbp]
\centering
\begin{tikzpicture}
\begin{axis}[
    width=0.95\textwidth,
    height=7cm,
    xlabel={Coefficient value},
                    ylabel={Density},
    xlabel style={font=\small},
    ylabel style={font=\small},
    ticklabel style={font=\footnotesize},
    axis lines=left,
    xmin=-1.8, xmax=0.6,
    xtick={-1.5, -1, -0.5, 0, 0.5,1},
    ymin=0, ymax=2.8,
    ytick={0,1,2},
    legend style={
        at={(0.1,0.99)},      % 图例位置 (x,y)
        anchor=north west,   % 图例锚点
        draw=none,           % 可选：去掉边框
        fill=none,           % 可选：去掉背景
        font=\small,
        align=left
    },
    legend cell align=left,
    tick style={black},
    legend columns=1,
    every axis plot/.append style={thick},
]

\pgfplotstableread[col sep=comma]{DH1.csv}\DHone
\pgfplotstableread[col sep=comma]{DV1.csv}\DVone
\pgfplotstableread[col sep=comma]{HS1.csv}\HSone

% Variational Bayes
\addplot[
    dashed,
    thick,
    black
] coordinates {(0,0) (0,2.7)};
\addlegendentry{Spike-and-Slab}

% Horseshoe
\addplot[
    color=green,
    thick
]
table[
    x=x,
    y=y,
] {\HSone};
\addlegendentry{Horseshoe}

% Debias Variational Bayes
\addplot[
    color=blue,
    thick
]
table[
    x=x,
    y=y,
] {\DVone};
                        \addlegendentry{Debiased-Spike-and-Slab}

% Debias Horseshoe
\addplot[
    color=red,
    thick
]
table[
    x=x,
    y=y,
] {\DHone};
                        \addlegendentry{Debiased-Horseshoe}

% Horseshoe
\addplot [
    name path=E,
    draw=none,
    forget plot
]
table[
    x=x,
    y=y,
] {\HSone};

\addplot [
    name path=F,
    domain=-3:3,
    draw=none,
    forget plot
] {0};

\addplot [
    green,
    opacity=0.3
] fill between [of=E and F, soft clip={domain=-0.846580927:0.126281795}];

% Debias Variational Bayes
\addplot [
    name path=A,
    draw=none,
    forget plot
]
table[
    x=x,
    y=y,
] {\DVone};

\addplot [
    name path=B,
    domain=-3:3,
    draw=none,
    forget plot
] {0};

\addplot [
    blue,
    opacity=0.3
] fill between [of=A and B, soft clip={domain=-1.252297236:-0.133349493}];

% Debias Horseshoe
\addplot [
    name path=C,
    draw=none,
    forget plot
]
table[
    x=x,
    y=y,
] {\DHone};

\addplot [
    name path=D,
    domain=-3:3,
    draw=none,
    forget plot
] {0};

\addplot [
    red,
    opacity=0.3
] fill between [of=C and D, soft clip={domain=-1.201603532:-0.1728713}];
\end{axis}
\end{tikzpicture}
\caption{Initial posterior distributions and debiased laws for the coefficient on initial GDP.}
\end{figure}
	
\subsection{Decline in Crime Rates}
Using U.S. state-level data, \citet{DonohueLevitt2001Abortion} identified a strong relationship between the legalization of abortion following the 1973 Roe v.\ Wade decision and the subsequent decline in crime rates. In their analysis, the dependent variable is the change in log per capita murder rates between 1986 and 1997 across states. This outcome is regressed on a measure of the effective abortion rate, which is always included as a predictor alongside twelve year dummy variables, and a collection of controls capturing alternative determinants of crime, such as the number of police officers and prisoners per 1,000 residents, among others. \citet{BelloniChernozhukovHansen2014HighATE} extended the set of controls by incorporating these variables in various forms, including levels, differences, squared differences, cross-products, initial conditions, and interactions with linear and quadratic time trends, resulting in a database of 284 variables and 576 observations. For this expanded model, the 95\% credible intervals under both initial posteriors and our debiased 95\% confidence intervals all contain zero. Note that the horseshoe induced posterior places almost all of its posterior mass at zero, which is close to the degenerate (at zero) posterior of the spike-and-slab. Thus, the two debiased distributions based on the spike-and-slab and horseshoe priors are largely indistinguishable. In this example, standard Bayesian approaches place (nearly) all posterior mass at zero, while our procedure yields proper distributions that may better capture estimation uncertainty. %\textcolor{red}{ [This sentence sounds off...]}
 
 \begin{figure}[htbp]
\centering
\begin{tikzpicture}
\begin{axis}[
    width=0.95\textwidth,
    height=7cm,
    xlabel={Coefficient value},
                    ylabel={Density},
    xlabel style={font=\small},
    ylabel style={font=\small},
    ticklabel style={font=\footnotesize},
    axis lines=left,
    xmin=-0.7, xmax=0.7,
    xtick={-0.5, 0, 0.5},
    ymin=0, ymax=10.5,
    ytick={0,5,10},
    legend style={
        at={(0.1,0.99)},      % 图例位置 (x,y)
        anchor=north west,   % 图例锚点
        draw=none,           % 可选：去掉边框
        fill=none,           % 可选：去掉背景
        font=\small,
        align=left
    },
    legend cell align=left,
    tick style={black},
    legend columns=1,
    every axis plot/.append style={thick},
]

\pgfplotstableread[col sep=comma]{DH2.csv}\DHone
\pgfplotstableread[col sep=comma]{DV2.csv}\DVone
\pgfplotstableread[col sep=comma]{HS2.csv}\HSone

% Variational Bayes
\addplot[
    dashed,
    thick,
    black
] coordinates {(0,0) (0,10.5)};
\addlegendentry{Spike-and-Slab}

% Horseshoe
\addplot[
    color=green,
    thick
]
table[
    x=x,
    y=y,
] {\HSone};
\addlegendentry{Horseshoe}

% Debias Variational Bayes
\addplot[
    color=blue,
    thick
]
table[
    x=x,
    y=y,
] {\DVone};
                        \addlegendentry{Debiased-Spike-and-Slab}

% Debias Horseshoe
\addplot[
    color=red,
    thick
]
table[
    x=x,
    y=y,
] {\DHone};
                        \addlegendentry{Debiased-Horseshoe}

% Horseshoe
\addplot [
    name path=E,
    draw=none,
    forget plot
]
table[
    x=x,
    y=y,
] {\HSone};

\addplot [
    name path=F,
    domain=-3:3,
    draw=none,
    forget plot
] {0};

\addplot [
    green,
    opacity=0.3
] fill between [of=E and F, soft clip={domain=-0.018099333:0.00054724}];

% Debias Variational Bayes
\addplot [
    name path=A,
    draw=none,
    forget plot
]
table[
    x=x,
    y=y,
] {\DVone};

\addplot [
    name path=B,
    domain=-3:3,
    draw=none,
    forget plot
] {0};

\addplot [
    blue,
    opacity=0.3
] fill between [of=A and B, soft clip={domain=-0.43715268:0.384011518}];

% Debias Horseshoe
\addplot [
    name path=C,
    draw=none,
    forget plot
]
table[
    x=x,
    y=y,
] {\DHone};

\addplot [
    name path=D,
    domain=-3:3,
    draw=none,
    forget plot
] {0};

\addplot [
    red,
    opacity=0.3
] fill between [of=C and D, soft clip={domain=-0.434090556:0.382386431}];
\end{axis}
\end{tikzpicture}
\caption{Initial posterior distributions and debiased laws for the coefficient on the effective abortion rate (blue and red overlap)}
\end{figure}

\section{Main Theoretical Results}\label{Sec: 4}
%\AIold{Various features of posterior distributions are used by empirical researchers for frequentist type inference. In particular, regions with high posterior probabilities, known as credible sets, are often interpreted as confidence sets. The frequentist validity of the debiased posterior distributions serves as a theoretical justification for the proposed Bayesian procedure. In this section, we present the main theoretical results, beginning with a generic BvM theorem formulated under a set of high-level conditions. The result applies broadly and is not restricted to specific prior choices.} 
In this section, we develop the large sample theory of the debiased posterior through its repeated-sampling properties. In particular, we establish BvM-type results showing that the debiased law is asymptotically Gaussian and that intervals formed from its quantiles have asymptotically correct frequentist coverage. We begin with a generic result under high-level conditions that is neither tied to a particular prior nor specifies a computation method to obtain the posterior.

For the reader’s convenience, we introduce notation that will be used throughout the paper. Let $I_p$ denote the $p\times p$ identity matrix, $e_j$ its $j$th column, and $E_J$ the matrix formed by any $J$ columns of $I_p$. For $1 \leq q < \infty$, we write $\|v\|_{q}$ for the standard $\ell_q$-norm of a vector $v$, that is, $\|v\|_{q}: = (\sum_j |v_j|^{q})^{1/q}$. We use $\|v\|_{0}$ to denote the number of nonzero entries of $v$, and $\|v\|_{\infty}$ for its sup-norm. For a matrix $A$, $\|A\|_{q}$ denotes the induced (operator) $\ell_q$-norm for $1\leq q\leq \infty$, and $\|A\|_{\max}$ the elementwise sup-norm. For two positive sequences $a_n$ and $b_n$, we write $a_n\lesssim b_n$ if $a_n\leq C b_n$ for some constant $C$ and all $n$, and $a_n\asymp b_n$ if $a_n\lesssim b_n$ and $b_n \lesssim a_n $. Let $Z^{(n)} := \{Y_i,X_i\}_{i=1}^n$ denote the data and $W^{(n)} := \{W_{ni}\}_{i=1}^n$ the Bayesian bootstrap weights. As previously, $\mathbb{E}_0[\cdot]$ indicates that the expectation is evaluated with respect to the distribution of $Z^{(n)}$ under $\beta=\beta_0$, while $\Pi_W(\,\cdot\,|Z^{(n)})$ denotes the conditional distribution of $W^{(n)}$ given $Z^{(n)}$. Since $W^{(n)}$ is independent of $Z^{(n)}$, this conditional distribution does not depend on $Z^{(n)}$. For a positive sequence $r_n$, the asymptotic symbols $o_{P_0}(r_n)$ and $O_{P_0}(r_n)$ are defined with respect to the same underlying probability measure $P_0$. The sub-Gaussian norm of a random variable $Z$, denoted by $\|Z\|_{\psi_2}$, is defined as $\|Z\|_{\psi_2}:=\inf\{t>0: \mathbb E_{0}[\exp(Z^2/t^2)]\leq 2\}$. For a random vector $Z$, its sub-Gaussian norm is defined as $\|Z\|_{\psi_2}:=\sup _{\|x\|_2=1}\| Z^{\intercal}x\|_{\psi_2}$.

\subsection{Bernstein-von Mises Theorem}
We begin by establishing the result under high-level assumptions that specify the posterior contraction rate of $\Pi_{\beta}(\,\cdot\,| Z^{(n)})$, the convergence rate of the precision matrix estimator $\hat{\Theta}_n$, a frequentist point estimator that centers the debiased law, and suitable moment and regularity conditions.

\begin{ass}[Contraction rate of $\Pi_{\beta}(\,\cdot\,| Z^{(n)})$]\label{Assump:Beta}
There exist some constant $C>0$ and a positive sequence $\epsilon_n\to 0$ such that 
\[\mathbb{E}_0\big[\Pi_{\beta}\big(\|\beta-\beta_0\|_1\geq C\epsilon_n \mid Z^{(n)}\big)\big]\to 0.\]
\end{ass}

\begin{ass}[Convergence rate of $\hat{\Theta}_n$]\label{Assump:Precision}
Let $\hat{\Omega}_n:= \sum_{i=1}^{n}X_iX_i^{\intercal}/n$. There exist positive sequences $\gamma_n\to 0$ and $\delta_n\to 0$ such that 
\[\|\hat{\Theta}_n\hat{\Omega}_n - I_{p}\|_{\max}=O_{P_0}(\gamma_n) \text{ and } \|\hat{\Theta}_n-\Theta_0\|_{\infty}=O_{P_0}(\delta_n).\]
\end{ass}

\begin{ass}[Centering point estimator]\label{Assump:Frequentist}
There exists a point estimator $\hat\beta_n$ (depending only on $Z^{(n)}$) such that the following expansion holds
\[\hat\beta_n = \beta_0  + \Theta_0 \frac{1}{n} \sum_{i=1}^{n}X_i\varepsilon_i + \Delta_n, ~\text{ where } \|\Delta_n\|_{\infty} = o_{P_0}(n^{-1/2}).\]
\end{ass}

\begin{ass}[Moment and regularity conditions]\label{Assump:Moments}
(i) $\{Y_{i},X_{i}\}_{i=1}^{n}$ are i.i.d. and satisfy the model \eqref{Eqn:Model}; 
(ii) $X_{i}$'s are sub-Gaussian random vectors with $\|X_{i}\|_{\psi_2}\in(0,\infty)$;
(iii) $\varepsilon_{i}$'s are sub-Gaussian random variables with $\|\varepsilon_{i}\|_{\psi_2}\in(0,\infty)$; 
(iv) $\mathbb{E}_0[|e_j^{\intercal}\Theta_0 X_{i}\varepsilon_i|^3]<\infty$ for each $j=1,\ldots, p$. 
\end{ass}			

Assumption \ref{Assump:Beta} imposes a contraction rate on the initial posterior distribution of $\beta$, accommodating both settings where $p$ is small relative to $n$ and those where $p$ is much larger than $n$. Importantly, Assumption~\ref{Assump:Beta} neither imposes Gaussianity nor specifies the likelihood used to construct the initial posterior distribution. Thus, a Gaussian likelihood used to construct the initial exact or variational posterior may be interpreted as a working likelihood, provided that the resulting posterior distribution satisfies the required $\ell_1$ contraction rate. We later provide primitive conditions under which this assumption holds for spike-and-slab priors and horseshoe-type priors, although we stress that the use of sparsity-inducing priors is not essential to our general theory. Alternative structures may be exploited to verify Assumption \ref{Assump:Beta}---see, e.g., \citet{WakayamaImaizumi2025Bayesian} who establish posterior contraction using a data-adaptive truncated Gaussian prior supported on leading eigenspaces of the empirical covariate covariance matrix, without assuming sparsity. 

Assumption \ref{Assump:Beta} also allows for either the exact posterior or an approximate posterior for $\beta$. The exact posterior arises directly from the Bayes’ rule, given a prior and the likelihood function, which are often obtained via MCMC algorithms in low-dimensional settings. In contrast, in high dimensions, it may no longer be feasible. For example, the point-mass spike-and-slab prior is often considered theoretically ideal for sparse Bayesian problems. However, exploring the full posterior can be computationally prohibitive, because of the combinatorial complexity of updating the discrete indicators whether to include each variable or not. Recently, the variational Bayesian approximation has become increasingly popular where one relies on an approximate posterior that minimizes the Kullback–Leibler divergence to the true posterior within the mean-field variational family. We provide sufficient conditions to ensure that Assumption \ref{Assump:Beta} holds for mean-field variational approximations under spike-and-slab priors \citep{RaySzabo2022Variational}.

Assumption \ref{Assump:Precision} specifies the convergence rate of the precision matrix estimator. 
When the dimensionality $p$ is fixed, this condition is trivially satisfied by the plug-in estimator $\hat{\Theta}_n = \hat{\Omega}_n^{-1}$ with $\gamma_n = 0$ and $\delta_n = n^{-1/2}$ under mild moment assumptions. 
In high-dimensional settings ($p \to \infty$), we will present primitive conditions under which the assumption holds when $\hat{\Theta}_n$ is obtained from the nodewise LASSO regressions of \citet{VandeGeer_OnAsymptotically_2014} or the CLIME approach of \cite{Caietal_ConstrainedSparse_2011}. Assumption \ref{Assump:Frequentist} requires the existence of a frequentist point estimator that serves as the center of the debiased law. It is satisfied by the OLS estimator when $p$ is fixed and by the debiased LASSO estimator in high-dimensional models \citep{VandeGeer_OnAsymptotically_2014} under the primitive conditions provided in Section \ref{Sec:42}. Besides LASSO, other types of pilot estimators have also been studied in the literature; see Remark \ref{Rem:DebiasPoint}. Finally, Assumption~\ref{Assump:Moments} imposes sub-Gaussianity on both the regressors and the error terms. These regularity conditions are standard in the high-dimensional inference literature.

The debiased distribution is jointly determined by the initial posterior distribution of $\beta$ and the distribution of the Bayesian bootstrap weights. Specifically, the debiased distribution is induced by the product of two conditional distributions:
\begin{align}
\Pi(\,\cdot\,|Z^{(n)}) := \Pi_{\beta}(\,\cdot\,|Z^{(n)}) \times \Pi_W(\,\cdot\,|Z^{(n)}),    
\end{align}
where $\Pi_{\beta}(\,\cdot\,|Z^{(n)})$ and $\Pi_W(\,\cdot\,|Z^{(n)})$ are, respectively, the initial posterior distribution of $\beta$ (or its approximation) and the distribution of the Bayesian bootstrap weights conditional on the data. For each $j=1,\ldots,p$, let $\mathcal{L}_{\Pi}\big(e_{j}^{\intercal}\sqrt{n}(\tilde\beta-\hat{\beta}_{n})\mid Z^{(n)}\big)$ denote the conditional law (under the correction scheme) of $e_{j}^{\intercal}\sqrt{n}(\tilde\beta-\hat{\beta}_{n})$ given $Z^{(n)}$. We study the weak convergence of these conditional laws, which we measure using the bounded Lipschitz distance $d_{BL}$. 
For probability measures $P$ and $Q$ on $\mathbf{R}^k$, $d_{BL}$ is defined as
\begin{align}\label{Eqn:dBL}
d_{BL}(P,Q)
:= \sup\left\{
  \left| \int f\, d(P-Q) \right| 
  : \|f\|_{BL}\le 1
\right\},
\end{align}
where the bounded Lipschitz norm of a measurable function $f:\mathbf{R}^k\to\mathbf{R}$ is given by
\begin{align}
\|f\|_{BL}
:= \sup_{x\in\mathbf{R}^k}|f(x)|
   + \sup_{x\neq y}\frac{|f(x)-f(y)|}{\|x-y\|_\infty}.
\end{align}

\begin{thm}\label{Thm:BvM1}
Suppose Assumptions \ref{Assump:Beta}-\ref{Assump:Moments} hold. If $\sqrt{n}\epsilon_n\gamma_{n}\to0$ and $\sqrt{n}(\epsilon_n\|\Theta_0\|_\infty+ \delta_n)(\sqrt{\log p/n} + \log^2 n\log^{2} p/n)\to 0$, then for each $j=1,\ldots, p$,
\begin{align}\label{Eqn: BvM1-thm}
d_{BL}\left(\mathcal{L}_{\Pi}\big(e_{j}^{\intercal}\sqrt{n}(\tilde\beta-\hat{\beta}_{n})\mid Z^{(n)}\big), N(0,\sigma^2_{0,j})\right)\overset{P_{0}}{\to} 0,
\end{align}
where $\sigma^2_{0,j}: = e_{j}^{\intercal}\Theta_0 \mathbb{E}_0[X_{i}X_{i}^{\intercal}\varepsilon_{i}^{2}]\Theta_0 e_{j}$. 
\end{thm}

Theorem \ref{Thm:BvM1} establishes a coordinatewise BvM-type result for the debiased law. In particular, the asymptotic variance $\sigma^{2}_{0,j}$ for the $j$th coordinate of $\sqrt{n}(\tilde{\beta}-\hat{\beta}_{n})$ coincides with the asymptotic variance for the centering estimator $\hat\beta_n$ under possibly heteroskedastic errors. 
The conditions $\sqrt{n}\epsilon_n\gamma_{n}\to0$ and $\sqrt{n}(\epsilon_n\|\Theta_0\|_\infty+ \delta_n)(\sqrt{\log p/n} + \log^2 n\log^{2} p/n)\to 0$ ensure two key properties:  
(i) the bias stemming from the initial posterior and the error from estimating the precision matrix are asymptotically negligible; and  
(ii) the remaining variation in the debiased law is asymptotically Gaussian. Importantly, the theorem accommodates a high-dimensional regime in which the dimension $p$ may grow exponentially with the sample size $n$. 
As will be shown later, under spike-and-slab or horseshoe priors for $\beta$ and nodewise LASSO or CLIME estimation for $\Theta_0$, the quantities $\epsilon_n$, $\gamma_n$, and $\delta_n$ can all be of order $\sqrt{\log p/n}$ if the sparsity levels of $\beta_0$ and $\Theta_0$ are bounded (implying $\|\Theta_0\|_\infty$ is bounded). Hence, the result remains valid as long as $\log^{5/2} p$ grows at most linearly with $n$, up to logarithmic factors.

We note that Assumption \ref{Assump:Frequentist} is not required for Theorem \ref{Thm:BvM1}. To see this, replace $\hat\beta_n$ in \eqref{Eqn: BvM1-thm} with the infeasible/oracle ``estimator'' $\beta_n^\dag:=\beta_0+\Theta_0 n^{-1}\sum_{i=1}^n X_i\varepsilon_i$. Then \eqref{Eqn: BvM1-thm} implies that the conditional median of $e_j^\intercal\tilde\beta$ differs from $e_j^\intercal\beta_n^\dag$ by $o_p(n^{-1/2})$. Consequently, centering at this oracle preserves the limiting distribution while dispensing with Assumption \ref{Assump:Frequentist}. However, a centering estimator is needed for simultaneous inference over a subvector, which we shall discuss towards the end of this section.

The BvM-type theorem enables the construction of marginal confidence intervals from quantiles of the debiased law. 
For $0<\alpha<1$, let $q_{n,j}(\alpha)$ denote the $\alpha$th conditional quantile of $\tilde{\beta}_{j}$ (the $j$th coordinate of $\tilde{\beta}$) given $Z^{(n)}$. We study the asymptotic frequentist coverage of the $100(1-\alpha)\%$ confidence intervals, defined for $j=1,\ldots,p$ as
\begin{align}\label{Eqn:CSinf}
	\mathcal{C}_{n,j}(\alpha) 
	= \big[\,q_{n,j}(\alpha/2),~~q_{n,j}(1-\alpha/2)\,\big].
\end{align}
In practice, $\mathcal{C}_{n,j}(\alpha)$ is infeasible but can be approximated by its simulated counterpart $\mathcal{C}_{n,j}(\alpha,B)$ defined in \eqref{Eqn:CS} by taking a large $B$.

\begin{cor}\label{Cor:BvM1}
Under the conditions of Theorem~\ref{Thm:BvM1}, for each $j=1,\ldots,p$ and $\alpha\in(0,1)$, if in addition $\sigma_{0,j}^{2}$ is bounded away from zero, then 
\[
\lim_{n\to\infty}P_0\big(\beta_{0,j}\in \mathcal{C}_{n,j}(\alpha)\big) = 1-\alpha.
\]
\end{cor}

Corollary~\ref{Cor:BvM1} implies that intervals formed from conditional quantiles of the debiased law attain asymptotically correct frequentist coverage. At a high level, the initial distribution provides posterior-based regularization, while the correction calibrates its quantiles for valid frequentist inference. 

We next consider simultaneous inference for a growing number of coefficients. Without loss of generality, we focus on the first $J$ coordinates. For this purpose, we strengthen Assumption \ref{Assump:Moments}(iv) to the following condition. 

\begin{ass}[Moment condition for simultaneous inference]\label{Assump:MomentsSimultaneous}
(i) The eigenvalues of $\Sigma_{0,J}: = E_{J}^{\intercal}\Theta_0 \mathbb{E}_0[X_{i}X_{i}^{\intercal}\varepsilon_{i}^{2}]\Theta_0 E_{J}$ are uniformly bounded; 
(ii) there exists a sequence $\nu_J>0$ such that  $\{\mathbb E_0[\|E_{J}^{\intercal}\Theta_0X_{i}\varepsilon_i\|^{3}_{\infty}]\}^{1/3}=O(\nu_{J})$.
\end{ass}

Assumption \ref{Assump:MomentsSimultaneous}(i) is satisfied if, for example, $\mathbb{E}_0[\varepsilon_{i}^{2}|X_i]$ is constant and the eigenvalues of $\Omega_0$ are bounded away from zero. 
Assumption \ref{Assump:MomentsSimultaneous} (ii) holds with $\nu_J = O(\log J)$ when the random vectors $\Theta_0X_{i}\varepsilon_i$ are sub-exponential.

For ease of notation, define 
\begin{align*}
a_n := \sqrt n \gamma_n\epsilon_n + (\delta_n+\epsilon_n\|\Theta_0\|_\infty)\left( \log^{1/2} p+\frac{\log^{2} n\log^{2} p}{n^{1/2}}\right)+\left(\frac{\nu_{J}^{2}J\log^{3}(1+J)}{n^{1/3}}\right)^{1/4}.
\end{align*}

\begin{thm}\label{Thm:BvM2}
Suppose Assumptions \ref{Assump:Beta}-\ref{Assump:Frequentist}, 
\ref{Assump:Moments}(i)-(iii), and \ref{Assump:MomentsSimultaneous} hold. 
If $a_n=o(1)$, then there exists some $\bar{\mathbb G}_J\sim N(0,\Sigma_{0,J})$ independent of \(Z^{(n)}\) such that
\[
\|E_{J}^{\intercal}\sqrt{n}(\tilde{\beta}-\hat{\beta}_{n})-\bar{\mathbb G}_J\|_\infty = O_{P_{Z,W}}(a_n),
\]
It follows that
\[
d_{BL}\!\left(
\mathcal{L}_{\Pi}\big(E_{J}^{\intercal}\sqrt{n}(\tilde{\beta}-\hat{\beta}_{n})\mid Z^{(n)}\big),\,
N(0,\Sigma_{0,J})
\right)
\overset{P_{0}}{\longrightarrow} 0.
\]
\end{thm}

%\textcolor{red}{Zheng: If we assume enough moment restrictions on $\|E_{J}^{\intercal}\Theta_0X_{i}\varepsilon_i\|_{\infty}$, we can obtain Theorem \ref{Thm:BvM2} under $J/n^2\to\infty$ (up to logarithmic factors by Theorem \ref{Thm:strong approx}. In particular, if $\mathbb E_0[\|E_{J}^{\intercal}\Theta_0X_{i}\varepsilon_i\|^{4}_{\infty}]$ is bounded uniformly in $n$ and $p$, then Theorem \ref{Thm:BvM2} holds under $J/n^2\to\infty$ (up to logarithmic factors). Perhaps we could just mention this in passing. }

Theorem \ref{Thm:BvM2} extends the marginal BvM result in Theorem \ref{Thm:BvM1} to simultaneous inference on a growing subset of regression coefficients. Specifically, the result establishes joint asymptotic normality of the debiased law for any subvector of $\tilde{\beta}$ of size $J$, provided that $J$ increases with the sample size $n$ at a suitable rate. This rate is derived from the strong approximation result for the exchangeable bootstrap in \cite{FangSantosShaikhTorgovitsky2023LP}. The condition $\nu_J^{6}J^{3}\log^{9}(1+J)/n \to 0$ captures the complexity of the simultaneous inference problem, balancing the growth of the subset dimension against sample size and moment bounds. If we assume enough moment restrictions on $\|E_{J}^{\intercal}\Theta_0X_{i}\varepsilon_i\|_{\infty}$ (e.g., $\mathbb E_0[\|E_{J}^{\intercal}\Theta_0X_{i}\varepsilon_i\|^{4}_{\infty}]$ being bounded uniformly in $n$ and $p$), we can obtain Theorem \ref{Thm:BvM2} under $J/n^2\to 0$ (up to logarithmic factors) by Theorem \ref{Thm:strong approx}. We note that  Gaussian approximations developed for maxima of sums of high-dimensional random vectors as in \citet{ChernozhukovChetverikovKatoKoike2022Improved} are inapplicable because Theorem \ref{Thm:BvM2} is concerned with approximating the distribution of a vector. If the parameter of interest is the maximum of a subvector with length $J$, then the results in \citet{ChernozhukovChetverikovKatoKoike2022Improved} and the representation of the Bayesian bootstrap weights in \eqref{Eqn: Bayesian weights} suggest that we may allow for much larger $J$ with $\log J=o(n^{1/5})$.

Analogous to Corollary \ref{Cor:BvM1}, we can establish the frequentist validity of the simultaneous band in \eqref{Eqn:SimultaneousBand}. For $0<\alpha<1$, let $\hat c_{n,J}(1-\alpha)$ be the conditional $(1-\alpha)$ quantile of $\left\|E_J^{\intercal}\sqrt n(\tilde\beta-\hat\beta_n)\right\|_{\infty}$ given $Z^{(n)}$ and $c_J(1-\alpha)$ the $(1-\alpha)$ quantile of \(\|\mathbb G_J\|_\infty\). Define 
\begin{align}\label{Eqn:SimultaneousBandInf}
\mathcal C_{n,j}(\alpha)
=\left[\hat\beta_{n,j}-\frac{\hat c_{n,J}(1-\alpha)}{\sqrt n},\,
\hat\beta_{n,j}+\frac{\hat c_{n,J}(1-\alpha)}{\sqrt n}\right],\quad j=1,\ldots,J.
\end{align}

The validity of \eqref{Eqn:SimultaneousBandInf} demands that $\|\mathbb G_J\|_{\infty}$ and the population analog $c_J(1-\alpha)$ of $\hat c_{n,J}(1-\alpha)$ be well-behaved in the sense of a suitable ``anticoncentration'' inequality. To derive such an inequality, we impose the final assumption for our general theory. 

\begin{ass}[Anti-concentration]\label{Ass:anti-concentration}
(i) $\bar\sigma_J^2:=\max_{j=1}^J\sigma_{0,j}^2>0$; (ii) $c_J(1-\alpha')\ge c_J(0.5)+\varpi_J$ for some \(\varpi_J>0\) and $\alpha'\in(\alpha,1/2)$; (iii) \((a_n\vee b_n)[\bar\sigma_J\log(1+J)+\varpi_J] = o(\varpi_J^2)\) with \(b_n:=\nu_JJ^{1/3}\log^{5/6}(1+J)/n^{1/6}+\sqrt n\varrho_n\) for \(\varrho_n>0\) such that $\|\Delta_n\|_\infty=O_{P_0}(\varrho_n)$. 
\end{ass}

Assumption \ref{Ass:anti-concentration}(i) imposes a minimal condition that $\|\mathbb G_J\|_\infty$ be non-degenerate, whereas many anti-concentration inequalities require the minimum variance to be positive (see, e.g., \citet{ChernozhukovChetverikovKato2016Coupling}). Assumption \ref{Ass:anti-concentration}(ii) introduces $\varpi_J$ that may be intuitively interpreted as a concentration parameter. The extent of concentration is then regulated by Assumption \ref{Ass:anti-concentration}(iii). Note that $a_n$ is the coupling rate of the strong approximation for the debiased draw $\tilde\beta$ in Theorem \ref{Thm:BvM2}, whereas $b_n$ is the coupling rate of \(E_J^\intercal \sqrt n (\hat\beta_n-\beta_0)\) to $\mathbb G_J$; see Lemma \ref{Lem:StrongApprox2}. In the special case when $\varpi_J$ is bounded away from zero and $\bar\sigma_J$ is bounded from above, Assumption \ref{Ass:anti-concentration}(iii) simplifies to $(a_n\vee b_n)\log(1+J)=o(1)$, meaning that the coupling rates should tend to zero faster than $1/\log(1+J)$ for the purpose of  simultaneous inference, which is a mild restriction.

\begin{cor}\label{Cor:BvM2}
Let the conditions of Theorem \ref{Thm:BvM2} and Assumption \ref{Ass:anti-concentration} hold with $\alpha\in(0,1/2)$. If in addition \(b_n=o(1)\), then it follows that
\[
\lim_{n\to\infty}P_0\left(\beta_{0,j}\in\mathcal C_{n,j}(\alpha),~~
\forall j=1,\ldots,J\right)= 1-\alpha.
\]
\end{cor}

\subsection{Primitive Conditions}\label{Sec:42}

We verify the high-level conditions in Assumptions \ref{Assump:Beta}-\ref{Assump:Frequentist} by introducing a set of primitive conditions. For concreteness, we focus on a benchmark specification in which the regression coefficients are assigned a spike-and-slab prior, the precision matrix is estimated via the nodewise LASSO procedure, and the centering point estimator is constructed by the debiased LASSO \citep{VandeGeer_OnAsymptotically_2014}. The general theoretical framework, however, extends to global–local shrinkage priors \citep{PolsonScott1996GlobalLocal}, including the horseshoe family \citep{CarvalhoPolsonScott2010Horseshoe}, and other estimation approaches for the precision matrix, including the CLIME approach \citep{Caietal_ConstrainedSparse_2011}. Detailed results for these additional cases are provided in the Appendix C.

% In contrast with the LASSO that performs shrinkage of regression coefficients, the S-S priors are designed to conduct explicit variable selection. 

We follow the construction of \cite{Castilloetal_BayesianSparse_2015}, who specify the spike-and-slab prior through the following hierarchical mixture for $\beta = (\beta_1,\ldots,\beta_p)^{\intercal}$, 
\begin{align}\label{SSprior}
	\pi(\beta|r,\lambda) = \prod_{j=1}^{p}\big[(1-r)\,\delta_0(\beta_j) + r\,\psi(\beta_j|\lambda)\big],	
\end{align}
where $\delta_0$ denotes a point mass at zero and $\psi(\beta_j|\lambda) = \frac{\lambda}{2}\exp(-\lambda|\beta_j|)$ is the Laplace density with hyperparameter $\lambda>0$. Below we call $\lambda$ as the Laplace slab rate for easy reference. The point mass in the spike-and-slab prior is designed for variable selection. The mixing weight $r$ controls the proportion of nonzero coefficients\footnote{If $r$ were deterministic, it would represent the expected number of nonzero coefficients \textit{a priori}.} and thus governs the degree of model sparsity. Following \citet{Castilloetal_BayesianSparse_2015}, we assign a Beta hyper-prior 
\begin{align}\label{SSpriorBeta} % Visualize Beta: https://mathlets.org/mathlets/beta-distribution/
r\sim \mathrm{Beta}(1,p^u),
\end{align}
with hyperparameter $u>1$, large values of which favor sparse models by placing most of the prior mass on small values of $r$.  

The exact posterior under this prior places mass over all $2^p$ candidate models, which is computationally prohibitive in high dimensions. We therefore adopt a variational Bayesian approximation scheme. The variational posterior is defined as the best approximation of the posterior within a given class. The crux is to choose this class both sufficiently rich to approximate the exact posterior well, while at the same time also simple enough so that the minimization problem can be efficiently solved. For this purpose, we consider the following mean-field family from \cite{RaySzabo2022Variational} for $\mu = (\mu_1,\ldots,\mu_p)^{\intercal}$, $\sigma^2 = (\sigma^2_1,\ldots,\sigma^2_p)^{\intercal}$, and $\gamma = (\gamma_1,\ldots,\gamma_p)^{\intercal}$,
 \begin{align}\label{VBFamily}
 \hspace{-0.2cm}\mathcal{P}_{MF}:=\left\{P_{\mu,\sigma^2,\gamma}=\bigotimes_{j=1}^p\big[\gamma_j\,N(\mu_j,\sigma_j^2)+(1-\gamma_j)\,\delta_0\big]:\mu_j\in\mathbf{R},\sigma^2_j>0,\gamma_j\in[0,1] \right\},
 \end{align}
 where $\gamma_j$ plays the role of a variational inclusion probability. 
 The resulting variational posterior is defined as the minimizer of the Kullback-Leibler divergence with respect to the exact posterior:
 \begin{align}\label{VBPost}
 	\tilde{\Pi}_{\beta}(\,\cdot\,|Z^{(n)})= \underset{P_{\mu,\sigma^2,\gamma}\in 	\mathcal{P}_{MF}}{\arg \min}\mathrm{KL}\left(P_{\mu,\sigma^2,\gamma}~\|~\Pi_{\beta}(\,\cdot\,|Z^{(n)})\right),
 \end{align}
 where $\mathrm{KL}(\,P\,\|\,Q\,):=\int \log (dP/dQ) dP$ for two probability distributions $P$ and $Q$, and $\Pi_{\beta}(\,\cdot\,|Z^{(n)})$ is the exact posterior. Note that the above mean-field class enforces substantial independence in the approximated posterior. By doing so, it significantly reduces the model complexity, because there are only $p$ inclusion variables $\gamma_j$ to consider, rather than a total of $2^p$ models which the exact posterior puts mass on.
 
The mean-field class approximates the posterior distribution by a product-form distribution, thereby ignoring dependence among coordinates. As we study the frequentist property of the variational posterior, the distribution of unknown parameters under the true data generating process becomes a point mass. As such, this point mass is factorizable in the terminology of \cite{WangBlei2019VB}, i.e., point masses reside within the mean field family. Therefore, there is no loss incurred by imposing the factorizable variational family for posterior consistency. In fact, it has been shown to achieve the desired rate of contraction to the true parameter \citep{RaySzabo2022Variational}, which is essential for our high-level conditions to ensure the asymptotic normality of the debiased law. 

Without the debiasing step, however, the variational posterior generally fails to deliver asymptotically correct coverage, even in low-dimensional parametric settings \citep{WangBlei2019VB}. This highlights the versatility of our posterior-correction procedure from a complementary perspective. The trade-off for relaxing the exactness of the posterior distribution lies in the substantial computational gains. Component-wise coordinate-ascent variational inference algorithms for $\tilde{\Pi}_{\beta}(\,\cdot\,|Z^{(n)})$ are given in  \cite{RaySzabo2022Variational}. Beyond the mean-field family, our arguments also extend to other closely related variational classes that allow dependence among the nonzero coordinates (see Equation (9) in \citealp{RaySzabo2022Variational}), yielding analogous theoretical guarantees.

Turning to Assumption \ref{Assump:Precision}, we consider the nodewise LASSO regression following \citet{VandeGeer_OnAsymptotically_2014}. Let $X_{i,j}$ denote the $j$th component of $X_{i}$ and let $X_{i,-j}\in\mathbf{R}^{p-1}$ be the subvector of $X_{i}$ excluding its $j$th entry. For each $j=1,\ldots, p$, we define the following LASSO regression:
 \begin{align}\label{Eqn: nodewise1}
 \hat{\theta}_j:=\underset{\theta \in \mathbf{R}^{p-1}}{\arg \min} \frac{1}{n}\sum_{i=1}^{n} (X_{i,j}-X_{i,-j}^{\intercal} \theta)^2 +2 \lambda_j\|\theta\|_1,
 \end{align}
 with the corresponding residual variance given by
  \begin{align}\label{Eqn: nodewise2}
 	\hat{\tau}_j^2:= \frac{1}{n}\sum_{i=1}^{n} (X_{i,j}-X_{i,-j} ^{\intercal}\hat\theta_{j})^2+\lambda_j\|\hat{\theta}_j\|_1,
 \end{align}
 where $\lambda_{j}>0$ is a tuning parameter. Writing $\hat{\theta}_j=(\hat{\theta}_{j, 1},\ldots, \hat{\theta}_{j, j-1}, \hat{\theta}_{j, j+1}, \ldots, \hat{\theta}_{j, p})$, we construct the estimated precision matrix $\hat{\Theta}_n$ as
 \begin{align}\label{Eqn: nodewise3}
 	\hat{\Theta}_n = \left(\begin{array}{cccc}
 		1/\hat{\tau}_1^2 & -\hat{\theta}_{1,2}/\hat{\tau}_1^2 & \cdots & -\hat{\theta}_{1, p}/\hat{\tau}_1^2 \\
 		-\hat{\theta}_{2,1}/\hat{\tau}_2^2 & 1/\hat{\tau}_2^2 & \cdots & -\hat{\theta}_{2, p}/\hat{\tau}_2^2 \\
 		\vdots & \vdots & \ddots & \vdots \\
 		-\hat{\theta}_{p, 1}/\hat{\tau}_p^2 & -\hat{\theta}_{p, 2}/\hat{\tau}_p^2 & \cdots & 1/\hat{\tau}_p^2
 	\end{array}\right).
 \end{align}
We refer readers to Section 2.1.1 in \cite{VandeGeer_OnAsymptotically_2014} for further details.
%The estimator $\hat\Theta_n$ is motivated by noting that the restriction $\Omega_0^{-1}\Omega_0=I_p$ amounts to first order conditions of $p$ linear projections after parametrizing $\Omega_0^{-1}$ in the same structure as $\hat\Theta_n$.

Another important perspective of the nodewise LASSO is about its automatic creation of (approximately) orthogonalized residual \citep{ZhangZhang2014CI} through the one-step update formula \eqref{Eqn:DebiasBayes}. Given the true precision matrix $\Theta_0:=(\mathbb{E}_0[XX^\top])^{-1}$, $\mathbf{e}_j^\top \Theta_0 X^\top=(X C_j^0)^\top/\tau_{j,0}^2$, where $X C_j^0$ is called the anti-projection of $X_j$ and $\tau_{j,0}^2$ is the residual variance. See Section 4.2 of \cite{vdGeer2016Sparsity} for the construction of $C_j^0$. $X C_j^0$ is orthogonal to the linear span of $X_{-j}$ by the well-known Frisch-Waugh-Lovell (FWL) theorem. The nodewise LASSO in \eqref{Eqn: nodewise1} --\eqref{Eqn: nodewise3} produces approximately orthogonal residuals from high-dimensional linear models by regressing $X_j$ on $X_{-j}^\top\theta_j$ with the $\ell_1$-penalty for $1\leq j\leq p$. %Heuristically speaking, our one-step update also aims to achieve the orthogonality or debiasing as in \cite{Walker2023PLM} and \cite{DiTragliaLiu2025BML}, albeit by plugging in certain frequentist estimation of the H-D $\Theta_0$.

We choose the centering point estimator as the debiased LASSO estimator. For $\hat{\Theta}_n$ given in \eqref{Eqn: nodewise1}-\eqref{Eqn: nodewise3}, the debiased LASSO estimator is given by
\begin{align}\label{DebiasedLASSO}
\hat{\beta}_n=\hat{\beta}_n^{\text{LASSO}}+ \hat{\Theta}_n\left[\frac{1}{n}\sum_{i=1}^{n}X_i(Y_i-X_i^\intercal\hat{\beta}_n^{\text{LASSO}})\right],
\end{align}
where $\hat{\beta}_n^{\text{LASSO}}$ is a LASSO estimator given by 
\begin{align}\label{LASSO}
 \hat{\beta}_n^{\text{LASSO}}:=\underset{\beta \in \mathbf{R}^{p}}{\arg \min} \frac{1}{n}\sum_{i=1}^{n} (Y_i-X_{i}^{\intercal}\beta)^2 +\rho\|\beta\|_1,
\end{align}
where $\rho>0$ is a tuning parameter. 

We adopt the following standard notation in analyzing the high-dimensional regression. For a vector $\beta=(\beta_1,\cdots,\beta_p)^{\intercal}\in\mathbf{R}^p$ and a subset $S\subset\{1,\cdots,p \}$ of indices, let $\beta_S$ be the vector $(\beta_j)_{j\in S}\in \mathbf{R}^{|S|}$, where $|S|$ is the cardinality of $S$. Let $S_{\beta}:=\{ j\in \{1,\cdots,p \} :\beta_{j}\neq 0\}$ denote the index set of nonzero components of $\beta$, and $s_{\beta}:= |S_{\beta}|$ denote the sparsity level of $\beta$.
%We write $S_0=S_{\beta_0}$ and $s_0=|S_{\beta_0}|$ for the true coefficient $\beta_0$. 
Recall that $\hat{\Omega}_n= \sum_{i=1}^{n}X_iX_i^\intercal/n$. Define $\nu(\hat{\Omega}_n):=\max_{1\leq j\leq p}\sqrt{e_{j}^{\intercal}\hat{\Omega}_n e_{j}}$ as the square root of the largest diagonal element of $\hat{\Omega}_n$.

\begin{ass}[Model Specification]\label{Assump:Error}
(i) $X_i$ and $\varepsilon_i$ are independent with $\mathbb E_{0}[X_i] =0$ and $\varepsilon_{i}\sim N(0,1)$; (ii) the eigenvalues of $\Omega_0$ are bounded away from zero and the diagonal entries of $\Omega_0$ are bounded, that is, $\max_{1\leq j \leq p}e_{j}^\intercal \Omega_0 e_j<\infty$.
\end{ass}

% $\mathbb E_{0}[X_i] =0$ is also imposed in the frequentist literature. But it can be removed. 

\begin{ass}[Hyper/Tuning Parameters]\label{Assump:Tuning}
(i) $\nu(\hat{\Omega}_n)\sqrt{n} /p\leq\lambda\leq 4\nu(\hat{\Omega}_n)\sqrt{n\log p}$ holds with probability one; (ii) $\max_{1\leq j\leq p} \lambda_j \asymp \sqrt{\log p / n}$; (iii) $\rho \asymp \sqrt{\log p / n}$.
\end{ass}

% \begin{ass}[Design Matrix]\label{Assump:Design}
% The compatibility numbers $\phi(S_{\beta_0})$ and $\bar{\psi}(S_{\beta_0})$ are bounded away from zero with probability approaching one.
% \end{ass}

\begin{ass}[Sparsity and Dimensionality]\label{Assump:Sparsity}
The following conditions hold: 
\[s_{\beta_0}\|\Theta_0\|_{\infty}\frac{\log p}{\sqrt{n}}\left(1 + \frac{\log^{3/2}p\log^2 n}{\sqrt{n}}\right)\to 0\]
 and 
\[\max_{1\leq j\leq p}s_{\Theta_0,j}\frac{\log p}{\sqrt{n}}\left(1 + \frac{\log^{3/2}p\log^2 n}{\sqrt{n}}\right)\to 0,\]
where $s_{\Theta_0,j}:=s_{\Theta_0e_j}$ denotes the sparsity level of the $j$th row/column of $\Theta_0$. 
\end{ass}

Assumption~\ref{Assump:Error} provides a convenient benchmark specification for verifying the high-level conditions. Its Gaussianity and independence restrictions are sufficient rather than necessary and are not imposed in Theorems~\ref{Thm:BvM1} and~\ref{Thm:BvM2}. For the exact and variational spike-and-slab posteriors, these restrictions can be replaced, for purposes of verifying Assumption~\ref{Assump:Beta}, by the high-dimensional score condition stated in Remark~\ref{Rem:NonGaussian}. Following \citet{Castilloetal_BayesianSparse_2015}, we normalize the error variance to one in this benchmark specification.
%All the assumptions are standard in the literature. Following \cite{Castilloetal_BayesianSparse_2015},we assume $\varepsilon_i \sim N(0,1)$ for simplicity. 
In the case of an unknown variance, $\varepsilon_i \sim N(0,\sigma_0^2)$, one may rescale the data using an estimate of $\sigma_{0}^2$, thereby adopting an empirical Bayes approach. Alternatively, a fully Bayesian treatment can be employed by assigning a prior to $\sigma_{0}^2$, for instance, an inverse-Gamma prior. In contrast to \cite{Castilloetal_BayesianSparse_2015} and \cite{RaySzabo2022Variational}, we provide primitive conditions for their high-level conditions on compatibility numbers. Since $\|\Theta_0\|_{\infty
}=O(\max_{1\leq j\leq p}\sqrt{s_{\Theta_0,j}})$, Assumption \ref{Assump:Sparsity} involves the sparsity levels of $\beta_0$ and $\Theta_0$ and the dimension $p$. However, the requirements are slightly different from those in the frequentist inference literature. For example, \citet{VandeGeer_OnAsymptotically_2014} require  $s_{\beta_0}\log p/\sqrt{n}\to 0$ and $\max_{1\leq j\leq p}s_{\Theta_0,j}\log p/\sqrt{n}\to 0$. When the sparsity levels of $\beta_0$ and $\Theta_0$ are bounded, Assumption \ref{Assump:Sparsity} requires $\log^{5/2}p = o(n/\log^2 n)$, which is slightly more retrictive than $\log p = o(\sqrt{n})$ required in \citet{VandeGeer_OnAsymptotically_2014}. Netherless, $p$ is allowed to grow exponentially with $n$ in line with the frequentist literature. 

\begin{pro}\label{Pro:Primitive}
Consider the prior in \eqref{SSprior}-\eqref{SSpriorBeta}, the precision matrix estimator in \eqref{Eqn: nodewise1}-\eqref{Eqn: nodewise3}, and the centering point estimator defined by \eqref{DebiasedLASSO}-\eqref{LASSO}. Suppose Assumptions \ref{Assump:Moments}, \ref{Assump:MomentsSimultaneous}, \ref{Assump:Error}-\ref{Assump:Sparsity} hold and $\nu_{J}^{6}J^{3}\log^{9}(1+J)/n\to 0$. The conclusions of Theorems \ref{Thm:BvM1} and \ref{Thm:BvM2} hold for the exact posterior with
\begin{align*}
a_n =  \left(\max_{1\leq j \leq p}s_{\Theta_0,j}+s_{\beta_0}\|\Theta_0\|_\infty \right) \frac{\log p}{\sqrt{n}}\left(1+\frac{\log^{2} n\log^{3/2} p}{\sqrt{n}}\right)+\left(\frac{\nu_{J}^{2}J\log^{3}(1+J)}{n^{1/3}}\right)^{1/4}.
\end{align*}
If, in addition,
$\lambda = O_{P_0}\!\left(\nu(\hat{\Omega}_n)\sqrt{n\log p}/s_{\beta_0}\right)$,
the conclusions also hold for the variational Bayesian approximate
posterior in \eqref{VBFamily}--\eqref{VBPost}, with $s_{\beta_0}$ in the
displayed expression for $a_n$ replaced by $\rho_n s_{\beta_0}$ for
a sufficiently slowly diverging sequence $\rho_n\to\infty$.
\end{pro}

The additional restriction on $\lambda$ for the variational posterior
comes from \citet[Theorem~1]{RaySzabo2022Variational}, whose maximum
column norm of the design matrix equals $\sqrt{n}\nu(\hat{\Omega}_n)$
in our notation. The slab rate $\lambda$ may be fixed in advance or
chosen as a function of the design matrix. Although
Assumption~\ref{Assump:Tuning}(i) does not automatically imply the
additional restriction, the feasible choice
$\lambda=\nu(\hat{\Omega}_n)\sqrt{n}/p$ satisfies both conditions.
Indeed, since $s_{\beta_0}\leq p$,
\[
\frac{\lambda s_{\beta_0}}
{\nu(\hat{\Omega}_n)\sqrt{n\log p}}
=
\frac{s_{\beta_0}}{p\sqrt{\log p}}
\leq \frac{1}{\sqrt{\log p}}.
\]
Moreover, the fixed choice $\lambda=1$ used in our simulations
satisfies the additional restriction because
$s_{\beta_0}=o(\sqrt{n\log p})$ under our primitive sparsity conditions
and $\nu(\hat{\Omega}_n)$ is bounded away from zero with probability
approaching one. Thus, the additional restriction can be satisfied
without knowledge or estimation of $s_{\beta_0}$.

\begin{rem}
\label{Rem:NonGaussian}
For the spike-and-slab specification, Gaussianity is not necessary for verifying the contraction requirement in Assumption~\ref{Assump:Beta}. As noted by \citet[Section~2.3 and Remark~1 in the supplementary material]{Castilloetal_BayesianSparse_2015}, Gaussianity can be replaced in their exact-posterior contraction result by the score condition
\begin{align}
P_0\left(
\left\|
\frac{1}{n}\sum_{i=1}^{n}X_i\varepsilon_i
\right\|_{\infty}
>
C\nu(\hat{\Omega}_n)
\sqrt{\frac{\log p}{n}}
\right)
\longrightarrow 0
\end{align}
for some constant $C>0$, with the remaining prior and compatibility conditions maintained. For the mean-field variational approximation, \cite{RaySzabo2022Variational} show that the same score condition can replace Gaussianity in establishing the contraction result required here. In our setting, a sufficient set of conditions for the displayed score bound is that the observations are independent across $i$, $\mathbb{E}_0[X_i\varepsilon_i]=0$, and the products $X_{i,j}\varepsilon_i$ have uniformly bounded sub-exponential norms. These conditions require neither Gaussian errors nor independence between $X_i$ and $\varepsilon_i$ and can therefore accommodate conditional heteroskedasticity. Thus, the exact spike-and-slab posterior and its mean-field variational approximation, both constructed from a Gaussian working likelihood, provide a class for which the contraction requirement in Assumption~\ref{Assump:Beta} continues to hold under non-Gaussian errors. Consequently, they provide the required contraction input to our BvM results\footnote{We want to highlight that the continuous-shrinkage contraction result of \citet{SongLiang_NearlyOptimal_2023} invoked in Proposition~\ref{Pro:Horseshoe} is established under Gaussian errors. That proposition should therefore be understood as a primitive verification for the Gaussian subclass. We anticipate this can be extended to non-Gaussian likelihood by tailoring the general theory from \cite{KlejinVaart2006Mis}. This is beyond our current scope and will be pursued elsewhere.}.
\end{rem}
%\begin{rem}[Non-Gaussian Error and Quasi-likelihood]
%All we need from the initial (approximated) posterior distribution of $\beta$ is the posterior contraction as stated in our Assumption \ref{Assump:Beta}. The Gaussianity on the error term is made only for simple illustration. In general, this working likelihood used to obtain the initial (quasi-)posterior need not be correctly specified.  Under $\mathbb{E}_0[X_i\varepsilon_i]=0$, its population criterion is uniquely minimized at $\beta_0$ because
%\[
%\mathbb{E}_0\log\frac{p_\beta^G(Y_i\mid X_i)}{p_{\beta_0}^G(Y_i\mid X_i)}
%=-\frac12(\beta-\beta_0)^\top\Omega_0(\beta-\beta_0).
%\]
%This is a special case of the general theory in Klejin and van der Vaart (2006). Notably, it does not require independence between $X_i$ and $\varepsilon_i$. The contraction proof for the exact and variational spike-and-slab posteriors then only requires the high-dimensional score event $\|\sum_iX_i\varepsilon_i\|_\infty=O_{P_0}(\sqrt{n\log p})$, as noted in Remark B.1 of \cite{RaySzabo2022Variational}. The latter condition can be further deduced from existing maximal inequalities as in \cite{KuchibhotlaChakrabortty_MovingBeyond_2022}. For instance, sufficient conditions include: (i) The regressors $X_i$ are sub-Gaussian. (ii) For some $\vartheta>0$ and $K<\infty$, $\max_{j\le p}\|X_{i,j}\varepsilon_i\|_{\psi_\vartheta}\le K$, and $\max_{j\le p}\mathbb E_0[X_{i,j}^2\varepsilon_i^2]\le K^2$. Hence, the same contraction rates continue to hold under non-Gaussian and heteroskedastic errors.  
%\end{rem}

\begin{rem}
% \cite{Castilloetal_BayesianSparse_2015} conduct a comprehensive theoretical study of the spike-and-slab induced posterior in high-dimensional settings. Their distributional approximation result shows that the posterior can be represented as a mixture of normal distributions, yielding asymptotically valid credible sets—provided that strong regularity conditions are imposed, including stringent assumptions on the hyperparameter $\lambda$ and a beta-min condition ensuring signal separability. However, these assumptions are often unrealistic in practice. More importantly, computing the exact posterior under the spike-and-slab prior is computationally prohibitive in high-dimensional problems. As our Monte Carlo results demonstrate, even when the exact posterior is replaced by its variational Bayes approximation, the method performs poorly for nonzero coefficients in finite samples. This highlights that, despite its appealing theoretical guarantees, the classical spike-and-slab framework faces both theoretical and computational limitations in practical high-dimensional applications.
\cite{Castilloetal_BayesianSparse_2015} provide thorough theoretical analysis of the exact posterior in high-dimensional settings. They establish a distributional approximation result, showing that the posterior can be approximated by a mixture of normal distributions. Under suitable conditions on the hyperparameter $\lambda$ and a standard \emph{beta-min} condition, the posterior credible sets achieve correct asymptotic coverage for nonzero coefficients, while assigning asymptotic point mass at zero for truly zero coefficients. These assumptions are substantially stronger. More importantly, the computational burden of evaluating the exact posterior makes it infeasible in real high-dimensional problems. As our Monte Carlo simulations demonstrate, when using the variational Bayes approximation, the method performs poorly in finite samples related to nonzero coefficients, highlighting the practical limitations of the existing spike-and-slab framework.
% \cite{Castilloetal_BayesianSparse_2015} study comprehensive theoretical properties of the S-S induced posteriors in the high-dimensional setup. They also derived a new distributional approximation, i.e., the posterior distribution is approximated by a mixture of normal. For the non-zero coefficient, the credible set from the posterior has the right asymptotic coverage. For the zero coefficient, it will be assigned with the mass at zero with probability approaching one. This distributional approximation requires additional conditions on the hyperparameter $\lambda$ and the \textit{beta-min} condition. More importantly, the practical restriction is due to the prohibitive computation of obtaining the exact posterior. As we show in the Monte Carlo simulation, this does not work well in finite sample for the non-zero coefficients using the variational Bayes approximation.
\end{rem}

% As in the classical results about the debiased inference, we use an asymptotic scheme, where $p=p_n\geq n\to\infty$. It explicitly allows for the case where the ambient dimension is larger than the sample size, as long as additional requirements in our high-level assumptions are satisfied. Those restrictions about the estimation error of the precision matrix and light tails of covariates and error terms are very mild. In the sequel, we suppress the dependence of the estimation error and stochastic bounds on $p$ and only keep the subscript $n$ for notational simplicity. 

\begin{rem}\label{Rem:DebiasPoint}
It is well known that the LASSO estimator corresponds to the posterior mode under the prior: for $\beta = (\beta_1,\ldots,\beta_p)^{\intercal}$, 
\begin{align}
\pi(\beta|\lambda) = \prod_{j=1}^{p}\frac{\lambda}{2}\exp(-\lambda|\beta_j|),
\end{align}
where $\lambda>0$ corresponds to the tuning parameter in LASSO controlling the amount of shrinkage and sparsity. Consequently, the frequentist debiasing approach can be viewed as applying a correction to a Bayesian point estimator. Beyond using the LASSO as the pilot estimator, the literature has also explored other Bayesian point estimators. For instance, \citet{ZhangPolitis2022Ridge} focus on ridge regression, which corresponds to the posterior mode under Gaussian priors on the regression coefficients, while \citet{Baietal2022Group} study the posterior mode in the additive regression, arising from a continuous spike-and-slab prior \citep{Rockova_ContinuousSpikeSlab_2018}. Our debiased inference procedure allows for these and other classes of priors. 
\end{rem}

%{\color{blue}
%Qihui: Our general debiased framework does not restrict the type of priorrs, although we do focus on sparse-inducing prior from the beginning of Section 2 as well as in Section 4.1. It also applies for dense priors, e.g., Gaussian priors. So we should not say ``Extending our debiased framework.'' This undermines our contribution. }

%{\color{blue}
%Replace \citet{Baietal2022Group} by \citet{RockovaGeorge_SpikeSlab_2018} and move \citet{Baietal2022Group} to the conclusion?}

%\textcolor{red}{[Zheng: Perhaps move discussions in the next two (maybe all three?) paragraphs to the previous section?]}

\subsection{Additional Discussions}\label{Sec:43}
Our proposal can be seen as performing a one-step correction based on the identity $\beta_0 = \chi(\beta_0,\Theta_0,F_0)$, where $F_0$ denotes the distribution of $(Y_i,X_i)$, and $\chi$ is a mapping from the space of $(\beta_0,\Theta_0,F_0)$ to that of $\beta_0$ defined by 
\begin{align}\label{Eqn: one-step}
\chi(\beta,\Theta,F)=\beta+\Theta\int x\left [y-x^\intercal \beta\right]dF(y,x).
\end{align}
We consider the posterior law of $F$, denoted by  $\Pi_{F}(\,\cdot\,|Z^{(n)})$, to be the Bayesian bootstrap law, which can be viewed as the posterior of $F$ in a nonparametric Bayesian model with a particular Dirichlet process prior \citep{Rubin_BayesianBootstrap_1981}.  Then, \eqref{Eqn:DebiasBayes} can be expressed as
\begin{align}
\tilde{\beta} = \chi(\beta,\hat{\Theta}_n,F),
\end{align}
where $(\beta,F)|Z^{(n)} \sim \Pi_{\beta}(\,\cdot\,|Z^{(n)}) \times \Pi_{F}(\,\cdot\,|Z^{(n)})$. Informally speaking, the debiased law is obtained by plugging in the frequentist estimator $\hat{\Theta}_n$ for $\Theta_0$, the initial posterior for $\beta$, and the Bayesian bootstrap for $F$. Note that the evaluation of $\Pi_{F}(\,\cdot\,|Z^{(n)})$ boils down to $\Pi_{W}(\,\cdot\,|Z^{(n)})$, whose randomness involves the weighted exponential random variables. Since the two posteriors are conditionally independent, the debiased law arises from the product of the initial posterior distribution of $\beta$ and the posterior distribution of $F$. Next, we compare our approach with some recent proposals in the literature. 

 %\begin{rem}\label{Rem: one step}
Our debiasing step aligns with the one-step posterior framework of \citet{YiuFongHolmesRousseau2025Semiparametric}, which studies scalar parameters of interest in general semiparametric models. Nonetheless, our proposal is distinct in several important ways. 
First, we consider a high-dimensional regime where the number of covariates may exceed the sample size, and we allow for simultaneous inference for a growing number of coefficients. The one-step correction term in \citet{YiuFongHolmesRousseau2025Semiparametric} involves the influence function perturbed by the Bayesian bootstrap weights. In a linear regression model, the $j$th regression coefficient $\beta_{0,j}$ has the influence function $(y,x^{\intercal})^{\intercal}\mapsto e_{j}^{\intercal}\Theta_{0} x(y-x^\intercal \beta_0)$ at the truth. 
% This is corrrect. Use ChatGPT.
This representation permits the application of their one-step correction in fixed-dimensional settings. Our framework explicitly employs sparsity-inducing priors in high-dimensional regimes. Moreover, our correction applies to the entire vector of regression coefficients, allowing the number of covariates to exceed the sample size while both $n$ and $p$ diverge. The resulting general BvM theorems enable inference on subvectors of regression coefficients whose dimension can grow with $(n,p)$. 

Second, we estimate $\Theta_0$ via a frequentist method rather than a Bayesian one, which substantially improves computational efficiency. This necessitates addressing non-trivial technical challenges and developing new theoretical results. Suppose the covariates $\{X_i\}_{i=1}^n$ are \textit{i.i.d.} draws from ${N}(0,\Theta_0^{-1})$. The likelihood of $\{X_i\}_{i=1}^n$ as a function of the precision matrix $\Theta_0$ is  
\begin{align}
    L(\{X_i\}_{i=1}^n|\Theta)
    = \frac{\det(\Theta)^{n/2}}{(2\pi)^{np/2}}
      \exp\left(-\frac{1}{2}\sum_{i=1}^{n} X_{i}^{\intercal}\Theta X_{i}\right).
\end{align}
A natural Bayesian approach would place a sparsity-inducing prior, such as a spike-and-slab prior or a horseshoe-type prior, on the entries of $\Theta_0$ \citep{Atchade2019Graph}. However, computing the posterior distribution of a high-dimensional precision matrix is prohibitive in practice. Our proposal instead adopts a more pragmatic strategy, as we simply plug in a frequentist estimator. In particular, one may compute $\hat{\Theta}_n$ either by running $p$ nodewise LASSO regressions or by solving $p$ linear programming problems in CLIME, both of which are computationally tractable. This plug-in approach is feasible because our procedure only requires $\hat{\Theta}_n$ to satisfy certain convergence rate conditions. Since our primary inferential target is the regression coefficients $\beta_0$, and $\Theta_0$ only plays a supporting role, this frequentist plug-in approach is both theoretically valid and practically effective. 

Third, our theory accommodates exchangeable bootstrap weights beyond the Bayesian bootstrap. The Bayesian bootstrap provides a convenient posterior-based weighting scheme, though this interpretation applies only to that component and not to the final debiased law. Our proofs rely on properties shared by other common exchangeable weights, which may therefore be substituted. The Bayesian bootstrap law corresponds to the posterior law when the empirical distribution of $\{Y_i,X_i\}_{i=1}^n$ is equipped with a Dirichlet process prior with degenerate (zero-mass) base measure. For more general base measures, however, the resulting posterior law becomes substantially more cumbersome both theoretically and computationally; see Equation (2.1) in \cite{RayVaart2021BvM} for a representation of the posterior in this case.

To facilitate a direct comparison with \cite{Walker2026BvMPLM} and \cite{DiTragliaLiu2025BML}, we consider the model
\begin{equation*}
Y_i=X_{i,1}\beta_1+X_{i,-1}^\top\beta_{-1}+\varepsilon_i,
\end{equation*}
where the parameter of interest is the scalar $\beta_1$ associated with the first regressor. The term $X_{i,-1}^\top\beta_{-1}$ plays the role of the nonlinear function $\eta(W_i)$ in \cite{Walker2026BvMPLM}.\footnote{One may take $W_i=X_{i,-1}$ and impose a linear structure on $\eta(\cdot)$. Alternatively, when $W_i$ is low-dimensional, a series approximation to $\eta(W_i)$ can generate a high-dimensional collection of features $X_{i,-1}$. Equations~(3)--(4) in \cite{Walker2026BvMPLM} provide a fully Bayesian interpretation, while its quasi-Bayesian version allows for non-normal errors.} The approach in \cite{Walker2026BvMPLM} builds on the partialling-out idea of \cite{Robinson1988RootN}, which generalizes the classical FWL theorem. Both \cite{Walker2026BvMPLM} and \cite{DiTragliaLiu2025BML} use Bayesian modeling of the conditional mean $\mathbb{E}[X_1\mid X_{-1}]$ to construct the orthogonal residual. When $X_{-1}$ is high-dimensional, a natural specification is the linear model $\mathbb{E}[X_1\mid X_{-1}]=X_{-1}^\top\gamma_1^0$, as in the second equation of~(13) in \cite{DiTragliaLiu2025BML}. Recall $\Theta_0:=\bigl(\mathbb{E}_0[X_iX_i^\top]\bigr)^{-1}$. Under this specification, the connection between our one-step update and FWL follows from the identity
\begin{equation*}
\mathbf{e}_1^\top\Theta_0 X_i
=\frac{X_i^\top C_1^0}{\tau_{1,0}^2}
=\frac{X_{i,1}-X_{i,-1}^\top\gamma_1^0}{\tau_{1,0}^2},
\end{equation*}
where $X_i^\top C_1^0$ is the anti-projection of the first regressor and $\tau_{1,0}^2$ is the corresponding residual variance. See our discussion in the paragraph following Equation \eqref{Eqn: nodewise3}. Thus, our approach targets the same residual and creates the same type of orthogonalization using a frequentist estimator of the precision matrix. In the high-dimensional setting with $p>n$, our illustration under primitive conditions proceeds with the nodewise LASSO \citep{VandeGeer_OnAsymptotically_2014}. Its computational feasibility has been elaborated in the second paragraph of this subsection. 

\section{Monte Carlo Simulations}\label{Sec: Monte Carlo}
In this section, we evaluate the finite-sample performance of the proposed posterior-correction procedure through a series of Monte Carlo experiments. We examine the frequentist coverage of confidence intervals formed from corrected-posterior laws and assess estimation accuracy in terms of bias and root mean squared error (RMSE). Beyond the stylized linear model with standard Gaussian errors presented in Section \ref{Sec: 4}, we further demonstrate the robustness and effectiveness of our approach across a range of more complex data-generating processes. We highlight that current results are all based on the default choices of hyper-parameters and tuning parameters in the existing R packages, which do not require additional interventions from users.
\subsection{Benchmark Simulation}
\pgfplotstableread{
	beta alpha  Bayes     DebiasBayes  DebiasLasso
	0    0.95   0.9956444  0.9379778     0.9522667
	1    0.95   0.1120000  0.9060000     0.8870000
	2    0.95   0.0750000  0.9050000     0.9100000
	3    0.95   0.0500000  0.8940000     0.8820000
	4    0.95   0.0670000  0.9070000     0.6290000
	5    0.95   0.0820000  0.9080000     0.8400000
}\pfnoheter  %p50 n100 heter

\pgfplotstableread{
	beta  alpha  Bayes      DebiasBayes  DebiasLasso
	0     0.95   0.9947778  0.9392       0.9522222
	1     0.95   0.0620000  0.9210       0.9480000
	2     0.95   0.0530000  0.9090       0.9270000
	3     0.95   0.0570000  0.8880       0.9110000
	4     0.95   0.0600000  0.9020       0.6450000
	5     0.95   0.0730000  0.9250       0.8710000
}\pfnohomochi  %p50 n100 homochi

\pgfplotstableread{
	beta  alpha     Bayes     DebiasBayes  DebiasLasso
	0     0.95      0.9913333 0.9454444    0.9519111
	1     0.95      0.3070000 0.8910000    0.9250000
	2     0.95      0.4590000 0.8810000    0.8470000
	3     0.95      0.4670000 0.8850000    0.7230000
	4     0.95      0.5210000 0.8990000    0.6220000
	5     0.95      0.5810000 0.9010000    0.7560000
}\pfnohomon  %p50 n100 homon

\pgfplotstableread{
	beta  alpha     Bayes     DebiasBayes  DebiasLasso
	0     0.95      0.9998632 0.9392316    0.9528421
	1     0.95      0.0530000 0.9210000    0.9080000
	2     0.95      0.0560000 0.9310000    0.9440000
	3     0.95      0.0780000 0.8900000    0.8970000
	4     0.95      0.0820000 0.9060000    0.6760000
	5     0.95      0.0940000 0.9290000    0.8530000
}\ponoheter  %p100 n100 heter

\pgfplotstableread{
	beta  alpha  Bayes      DebiasBayes  DebiasLasso
	0     0.95   0.9999158  0.9371368    0.9528105
	1     0.95   0.0390000  0.9400000    0.9590000
	2     0.95   0.0460000  0.9130000    0.9290000
	3     0.95   0.0470000  0.9020000    0.9310000
	4     0.95   0.0400000  0.9020000    0.7100000
	5     0.95   0.0530000  0.9250000    0.8700000
}\ponohomochi  %p100 n100 homochi

\pgfplotstableread{
	beta  alpha  Bayes      DebiasBayes  DebiasLasso
	0     0.95   0.9995158  0.9480526    0.9601368
	1     0.95   0.1520000  0.9210000    0.9330000
	2     0.95   0.3130000  0.9000000    0.8700000
	3     0.95   0.3330000  0.8840000    0.7330000
	4     0.95   0.4080000  0.8970000    0.5900000
	5     0.95   0.5460000  0.9140000    0.7280000
}\ponohomon  %p100 n100 homon
	
\pgfplotstableread{
	beta  alpha  Bayes      DebiasBayes  DebiasLasso
	0     0.95    0.9999333  0.9378872     0.953159
	1     0.95    0.0390000  0.9270000     0.921000
	2     0.95    0.0430000  0.9220000     0.944000
	3     0.95    0.0570000  0.8740000     0.918000
	4     0.95    0.0500000  0.9140000     0.728000
	5     0.95    0.0530000  0.9290000     0.885000
}\ptonoheter  %p200 n100 heter

\pgfplotstableread{
	beta  alpha  Bayes      DebiasBayes  DebiasLasso
	0     0.95    0.999959   0.9368205     0.9530615
	1     0.95    0.018000   0.9210000     0.9540000
	2     0.95    0.029000   0.9190000     0.9410000
	3     0.95    0.046000   0.8840000     0.9340000
	4     0.95    0.045000   0.9170000     0.7540000
	5     0.95    0.042000   0.9280000     0.8780000
}\ptonohomochi  %p200 n100 homochi

\pgfplotstableread{
	beta  alpha  Bayes      DebiasBayes  DebiasLasso
	0     0.95    0.9997385  0.9457128     0.9647282
	1     0.95    0.1190000  0.9230000     0.9440000
	2     0.95    0.2240000  0.8780000     0.8640000
	3     0.95    0.2200000  0.8500000     0.7250000
	4     0.95    0.3010000  0.8790000     0.5640000
	5     0.95    0.4730000  0.9240000     0.7400000
}\ptonohomon  %p200 n100 homon

%F2
\pgfplotstableread{
	beta alpha Bayes      DebiasBayes DebiasLasso
	0    0.95  0.9998444  0.9439556   0.9522
	1    0.95  0.1900000  0.9210000   0.8590
	2    0.95  0.2890000  0.9180000   0.9380
	3    0.95  0.3330000  0.9330000   0.9500
	4    0.95  0.3780000  0.9250000   0.9410
	5    0.95  0.5790000  0.9450000   0.9680
}\CovpfnoheterSigmaone  %p50 n100 heter Sigma1

\pgfplotstableread{
	beta alpha Bayes      DebiasBayes DebiasLasso
	0    0.95  0.9998889  0.9355111   0.9498444
	1    0.95  0.0800000  0.9090000   0.9370000
	2    0.95  0.1220000  0.9240000   0.9390000
	3    0.95  0.1300000  0.9270000   0.9550000
	4    0.95  0.1540000  0.9370000   0.9550000
	5    0.95  0.2510000  0.9450000   0.9560000
}\CovpfnohomochiSigmaone  %p50 n100 homochi Sigma1

\pgfplotstableread{
	beta alpha Bayes      DebiasBayes DebiasLasso
	0    0.95  0.9991778  0.9604222   0.9544222
	1    0.95  0.6210000  0.9200000   0.9470000
	2    0.95  0.8080000  0.9120000   0.9240000
	3    0.95  0.8710000  0.9230000   0.9320000
	4    0.95  0.9050000  0.9210000   0.9360000
	5    0.95  0.9280000  0.9450000   0.9510000
}\CovpfnohomonSigmaone  %p50 n100 homon Sigma1

\pgfplotstableread{
	beta alpha Bayes      DebiasBayes DebiasLasso
	0    0.95  0.9998211  0.9430632   0.9519053
	1    0.95  0.1470000  0.9070000   0.8520000
	2    0.95  0.1970000  0.9230000   0.9500000
	3    0.95  0.2510000  0.9220000   0.9500000
	4    0.95  0.3170000  0.9270000   0.9470000
	5    0.95  0.5290000  0.9490000   0.9590000
}\CovponoheterSigmaone  %p100 n100 heter S1

\pgfplotstableread{
	beta alpha Bayes      DebiasBayes DebiasLasso
	0    0.95  0.9999579  0.9342526   0.9499684
	1    0.95  0.0460000  0.9280000   0.9580000
	2    0.95  0.0700000  0.9310000   0.9560000
	3    0.95  0.0790000  0.9140000   0.9370000
	4    0.95  0.1110000  0.9340000   0.9530000
	5    0.95  0.2010000  0.9550000   0.9610000
}\CovponohomochiSigmaone  %p100 n100 homochi S1

\pgfplotstableread{
	beta alpha Bayes      DebiasBayes DebiasLasso
	0    0.95  0.9991053  0.9630316   0.9560211
	1    0.95  0.4760000  0.8800000   0.9440000
	2    0.95  0.6820000  0.9040000   0.9310000
	3    0.95  0.8080000  0.9240000   0.9390000
	4    0.95  0.8660000  0.9300000   0.9380000
	5    0.95  0.9250000  0.9460000   0.9480000
}\CovponohomonSigmaone  %p100 n100 homon S1

\pgfplotstableread{
	beta alpha Bayes      DebiasBayes DebiasLasso
	0    0.95  0.9994667  0.9422359   0.9518769
	1    0.95  0.1050000  0.8900000   0.8770000
	2    0.95  0.1460000  0.9030000   0.9400000
	3    0.95  0.1870000  0.9250000   0.9410000
	4    0.95  0.2320000  0.9280000   0.9520000
	5    0.95  0.4480000  0.9430000   0.9550000
}\CovptnoheterSigmaone  %p200 n100 heter S1

\pgfplotstableread{
	beta alpha Bayes      DebiasBayes DebiasLasso
	0    0.95  0.9998256  0.9352974   0.9508923
	1    0.95  0.0280000  0.9030000   0.9450000
	2    0.95  0.0430000  0.9160000   0.9470000
	3    0.95  0.0590000  0.9200000   0.9530000
	4    0.95  0.0680000  0.9360000   0.9580000
	5    0.95  0.1570000  0.9540000   0.9650000
}\CovptnohomochiSigmaone  %p200 n100 homochi S1

\pgfplotstableread{
	beta alpha Bayes      DebiasBayes DebiasLasso
	0    0.95  0.9987333  0.9634718   0.9564872
	1    0.95  0.3640000  0.8100000   0.9550000
	2    0.95  0.6020000  0.8850000   0.9380000
	3    0.95  0.7430000  0.9130000   0.9350000
	4    0.95  0.8540000  0.9300000   0.9420000
	5    0.95  0.9240000  0.9390000   0.9520000
}\CovptnohomonSigmaone  %p200 n100 homon S1

%F3
\pgfplotstableread{
	beta Bayes       DebiasBayes   DebiasLasso
	0    0.001405958  0.0004518953  0.001585543
	1    0.188266654  0.0198290763  0.006154626
	2    0.379296830  0.0079773669  0.006147921
	3    0.551779900  0.0223169900  0.005791037
	4    0.687083060  0.0496947506  0.029428829
	5    0.798797334  0.0089568229  0.008365539
}\BiaspfnoheterSigmaone  %p50 n100 heter Sigma1

\pgfplotstableread{
	beta Bayes        DebiasBayes  DebiasLasso
	0    0.0001724629  0.021572592  0.022899255
	1    0.2298244455  0.021302251  0.005174676
	2    0.4558012007  0.018473234  0.002541299
	3    0.6727092236  0.017276890  0.005215583
	4    0.8796584906  0.036355861  0.024209296
	5    1.4444430016  0.005467414  0.001221988
}\BiaspfnohomochiSigmaone  %p50 n100 homochi Sigma1

\pgfplotstableread{
	beta Bayes       DebiasBayes  DebiasLasso
	0    0.001383397  0.005604515  0.003177875
	1    0.123256509  0.019453508  0.011714049
	2    0.084265179  0.003992994  0.008970469
	3    0.064312110  0.008131654  0.019866299
	4    0.066967318  0.014106107  0.030450883
	5    0.055819965  0.006161213  0.022271943
}\BiaspfnohomonSigmaone  %p50 n100 homon Sigma1

\pgfplotstableread{
	beta Bayes       DebiasBayes   DebiasLasso
	0    0.002200622  0.0005885165  0.0004309773
	1    0.209249474  0.0309213007  0.0058489507
	2    0.433571940  0.0500124761  0.0268647271
	3    0.606398851  0.0267364013  0.0079332456
	4    0.727287633  0.0412250312  0.0239651555
	5    0.883873170  0.0031275531  0.0058530058
}\BiasponoheterSigmaone  %p100 n100 heter S1

\pgfplotstableread{
	beta Bayes        DebiasBayes  DebiasLasso
	0    0.0004517858  0.017746747  0.014082731
	1    0.2440684313  0.015509021  0.011653418
	2    0.4768083808  0.025150244  0.004387045
	3    0.6944520531  0.036664305  0.026167961
	4    0.8928678258  0.009734311  0.024369349
	5    1.5364042550  0.005458272  0.002653375
}\BiasponohomochiSigmaone  %p100 n100 homochi S1

\pgfplotstableread{
	beta Bayes        DebiasBayes  DebiasLasso
	0    0.0003041703  0.003508030  0.00230409
	1    0.1652696245  0.034634330  0.01369987
	2    0.1777309113  0.028141826  0.02851929
	3    0.1271038110  0.010364185  0.02419351
	4    0.1082425540  0.013614436  0.03484569
	5    0.0426113377  0.009365502  0.01871877
}\BiasponohomonSigmaone  %p100 n100 homon S1

\pgfplotstableread{
	beta Bayes        DebiasBayes  DebiasLasso
	0    0.0001799846  0.009249999  0.007015938
	1    0.2165986597  0.041992712  0.011857832
	2    0.4451585143  0.041864287  0.010924753
	3    0.6271699454  0.043673796  0.018468610
	4    0.7847656468  0.027481907  0.008362704
	5    1.0399637874  0.027770277  0.028437195
}\BiasptnoheterSigmaone  %p200 n100 heter S1

\pgfplotstableread{
	beta Bayes       DebiasBayes  DebiasLasso
	0    0.001088812  0.03254957   0.034389545
	1    0.240153859  0.02675843   0.004916506
	2    0.479685095  0.03712065   0.012825109
	3    0.710580816  0.03511702   0.023680412
	4    0.929041048  0.05068757   0.035217777
	5    1.652809687  0.04848980   0.038397464
}\BiasptnohomochiSigmaone  %p200 n100 homochi S1

\pgfplotstableread{
	beta Bayes       DebiasBayes  DebiasLasso
	0    0.001279871  0.005058094  0.004698813
	1    0.182678905  0.046673645  0.016240946
	2    0.215813474  0.031674831  0.023118812
	3    0.183858651  0.023055634  0.029313589
	4    0.140857600  0.016156247  0.031925811
	5    0.065958831  0.012642770  0.038299084
}\BiasptnohomonSigmaone  %p200 n100 homon S1

%F4
\pgfplotstableread{
	beta Bayes       DebiasBayes  DebiasLasso
	0    0.01925207   0.01589084   0.09247683
	1    0.22005008   0.05499430   0.04028210
	2    0.41834090   0.08614162   0.12418467
	3    0.64245754   0.14470681   0.16520591
	4    0.79990089   0.12796052   0.32591465
	5    0.40817257   0.08722685   0.18174314
}\BiaspfnoheterSigmatwo  %p50 n100 heter Sigma2

\pgfplotstableread{
	beta Bayes       DebiasBayes  DebiasLasso
	0    0.01861526   0.02142069   0.10998463
	1    0.23570318   0.05861051   0.03837186
	2    0.44150229   0.07770014   0.11975520
	3    0.65084910   0.13654482   0.14655078
	4    0.83032609   0.12435710   0.33407180
	5    0.43106131   0.09160899   0.18111704
}\BiaspfnohomochiSigmatwo  %p50 n100 homochi Sigma2

\pgfplotstableread{
	beta Bayes       DebiasBayes  DebiasLasso
	0    0.003968657  0.003084803  0.01835252
	1    0.113172058  0.020778603  0.03020753
	2    0.111124704  0.030065644  0.06907534
	3    0.062742655  0.022823021  0.10995558
	4    0.041286355  0.017870030  0.13628569
	5    0.026227520  0.012266941  0.08989703
}\BiaspfnohomonSigmatwo  %p50 n100 homon Sigma2

\pgfplotstableread{
	beta Bayes       DebiasBayes  DebiasLasso
	0    0.0003854495 0.01096429   0.09335733
	1    0.2331656821 0.06218585   0.03936348
	2    0.4632173206 0.08336527   0.10991463
	3    0.6692672165 0.15808925   0.15691825
	4    0.8253575326 0.13633374   0.31972620
	5    0.4103119204 0.09383610   0.18014970
}\BiasponoheterSigmatwo  %p100 n100 heter S2

\pgfplotstableread{
	beta Bayes       DebiasBayes  DebiasLasso
	0    0.0000417    0.01837294   0.11257658
	1    0.2421614    0.05504538   0.03360933
	2    0.4672417    0.07278867   0.10560262
	3    0.6924365    0.15432506   0.14150559
	4    0.8859219    0.13005213   0.32324648
	5    0.4421301    0.09837324   0.17942505
}\BiasponohomochiSigmatwo  %p100 n100 homochi S2

\pgfplotstableread{
	beta Bayes     DebiasBayes  DebiasLasso
	0 0.003506978 0.000316633 0.01817540
	1 0.184636804 0.037862667 0.03510348
	2 0.248930128 0.055971252 0.07236356
	3 0.212215880 0.047770395 0.12039333
	4 0.123404768 0.030420800 0.15347303
	5 0.041941850 0.016230821 0.10470971
}\BiasponohomonSigmatwo  %p100 n100 homon S2

\pgfplotstableread{
	beta Bayes     DebiasBayes  DebiasLasso
	0 0.0007001274 0.01035257  0.09325036
	1 0.2377599729 0.05398724  0.02732927
	2 0.4640292922 0.07979414  0.09458665
	3 0.6941966000 0.16645638  0.14223467
	4 0.8925797685 0.14864868  0.30480988
	5 0.4387484356 0.09187280  0.15887714
}\BiasptnoheterSigmatwo  %p200 n100 heter S2

\pgfplotstableread{
	beta  Bayes       DebiasBayes  DebiasLasso
	0     0.0010033   0.01765966   0.1096036
	1     0.2443898   0.05233355   0.0271311
	2     0.4785253   0.09218970   0.1101223
	3     0.6974934   0.15740223   0.1260607
	4     0.9078971   0.13117534   0.2983622
	5     0.4587348   0.10312010   0.1683140
}\BiasptnohomochiSigmatwo  %p200 n100 homochi S2

\pgfplotstableread{
	beta  Bayes       DebiasBayes  DebiasLasso
	0     0.00540785  0.00046631   0.01902918
	1     0.19907592  0.03883401   0.03311536
	2     0.31594224  0.07267480   0.07249739
	3     0.32860568  0.07351858   0.12647247
	4     0.21528844  0.05249209   0.16446921
	5     0.06759288  0.01927046   0.10592719
}\BiasptnohomonSigmatwo  %p200 n100 homon S2

%F5
\pgfplotstableread{
	beta   Bayes       DebiasBayes  DebiasLasso
	0      0.08254837  0.5086305    0.5441416
	1      0.25971430  0.2637163    0.2786553
	2      0.46170014  0.3033568    0.3153846
	3      0.65905194  0.3591929    0.3726307
	4      0.84968007  0.4218454    0.4267199
	5      1.29128758  0.4603910    0.4594298
}\RMSEpfnoheterSigmaone  %p50 n100 heter Sigma1

\pgfplotstableread{
	beta  Bayes       DebiasBayes  DebiasLasso
	0     0.06209489  0.6611826     0.6778838
	1     0.25160765  0.2675794     0.2839598
	2     0.48733017  0.3771388     0.3906908
	3     0.72631279  0.4635989     0.4665322
	4     0.95027629  0.5102706     0.5193839
	5     1.73324747  0.5847290     0.5860584
}\RMSEpfnohomochiSigmaone  %p50 n100 homochi Sigma1

\pgfplotstableread{
	beta       Bayes       DebiasBayes   DebiasLasso
	0          0.06815679  0.2451978     0.2804317
	1          0.19176764  0.1098113     0.1133537
	2          0.26271838  0.1660429     0.1713026
	3          0.27672554  0.1957949     0.2049842
	4          0.28410034  0.2180002     0.2322905
	5          0.24572306  0.2394319     0.2590445
}\RMSEpfnohomonSigmaone  %p50 n100 homon Sigma1

\pgfplotstableread{
	beta       Bayes       DebiasBayes   DebiasLasso
	0          0.04539301  0.5078626     0.5396947
	1          0.25982750  0.2650866     0.2876093
	2          0.47776810  0.2918635     0.3088945
	3          0.68868893  0.3547176     0.3620187
	4          0.88052546  0.4314152     0.4370849
	5          1.36814845  0.4550914     0.4526235
}\RMSEponoheterSigmaone  %p100 n100 heter S1

\pgfplotstableread{
	beta       Bayes       DebiasBayes   DebiasLasso
	0          0.02811768  0.6416459     0.6581995
	1          0.24828050  0.2525080     0.2779783
	2          0.49103707  0.3650983     0.3790894
	3          0.73751676  0.4876046     0.4917899
	4          0.96280747  0.5483693     0.5549353
	5          1.78887747  0.5720341     0.5695781
}\RMSEponohomochiSigmaone  %p100 n100 homochi S1

\pgfplotstableread{
	beta       Bayes       DebiasBayes   DebiasLasso
	0          0.0353110   0.2342876     0.2740181
	1          0.2103085   0.1154194     0.1137970
	2          0.3150441   0.1729335     0.1697419
	3          0.3360473   0.2032149     0.2035512
	4          0.3418442   0.2322440     0.2403927
	5          0.2545867   0.2439282     0.2550367
}\RMSEponohomonSigmaone  %p100 n100 homon S1

\pgfplotstableread{
	beta       Bayes       DebiasBayes   DebiasLasso
	0          0.06266688  0.5109558     0.5367663
	1          0.25404252  0.2516392     0.2783975
	2          0.48008979  0.2982676     0.3109914
	3          0.70624598  0.3689816     0.3798191
	4          0.91121555  0.4225943     0.4320719
	5          1.47325863  0.4628739     0.4641672
}\RMSEptnoheterSigmaone  %p200 n100 heter S1

\pgfplotstableread{
	beta       Bayes       DebiasBayes   DebiasLasso
	0          0.03904932  0.6827244     0.6942065
	1          0.25422926  0.2530075     0.2775954
	2          0.49637675  0.3846165     0.3996284
	3          0.73710518  0.4641862     0.4659896
	4          0.97600696  0.5270730     0.5407012
	5          1.84099104  0.5693631     0.5718987
}\RMSEptnohomochiSigmaone  %p200 n100 homochi S1

\pgfplotstableread{
	beta       Bayes       DebiasBayes   DebiasLasso
	0          0.04626209  0.2380254     0.2807367
	1          0.22161738  0.1220298     0.1164194
	2          0.34942491  0.1761957     0.1717340
	3          0.39417550  0.2128253     0.2157517
	4          0.38773882  0.2362682     0.2429289
	5          0.27160848  0.2506018     0.2717188
}\RMSEptnohomonSigmaone  %p200 n100 homon S1

%F6
\pgfplotstableread{
	beta       Bayes       DebiasBayes   DebiasLasso
	0          0.09168509  0.2163082     0.2046474
	1          0.27463131  0.3054994     0.3104504
	2          0.47637559  0.2479545     0.2600198
	3          0.70733907  0.2712974     0.2602138
	4          0.90924111  0.2656003     0.3858627
	5          0.49621814  0.2501143     0.2675998
}\RMSEpfnoheterSigmatwo  %p50 n100 heter Sigma2

\pgfplotstableread{
	beta       Bayes       DebiasBayes   DebiasLasso
	0          0.08779023  0.2299180     0.2245792
	1          0.26106128  0.2656319     0.2720877
	2          0.48299388  0.2550444     0.2657801
	3          0.71286718  0.2826652     0.2627743
	4          0.92871305  0.2820502     0.4032482
	5          0.50143339  0.2504466     0.2700973
}\RMSEpfnohomochiSigmatwo  %p50 n100 homochi Sigma2

\pgfplotstableread{
	beta       Bayes       DebiasBayes   DebiasLasso
	0          0.02958598  0.09069449    0.08294056
	1          0.20510245  0.11858631    0.11920002
	2          0.23966992  0.12577262    0.13439824
	3          0.21783138  0.12206385    0.15986872
	4          0.19774494  0.11978397    0.17797528
	5          0.11721226  0.10809705    0.14070302
}\RMSEpfnohomonSigmatwo  %p50 n100 homon Sigma2

\pgfplotstableread{
	beta       Bayes       DebiasBayes   DebiasLasso
	0          0.01113178  0.2225581     0.2142471
	1          0.26707937  0.2972665     0.3013392
	2          0.48454461  0.2314504     0.2482455
	3          0.71034970  0.2735636     0.2609650
	4          0.91087610  0.2646515     0.3805860
	5          0.47610880  0.2380804     0.2650259
}\RMSEponoheterSigmatwo  %p100 n100 heter S2

\pgfplotstableread{
	beta       Bayes       DebiasBayes   DebiasLasso
	0          0.00019168  0.2336082     0.2315057
	1          0.25160764  0.2581977     0.2692475
	2          0.48637188  0.2539778     0.2669972
	3          0.72393887  0.2779722     0.2557865
	4          0.94569961  0.2771307     0.3924041
	5          0.48940920  0.2476144     0.2717588
}\RMSEponohomochiSigmatwo  %p100 n100 homochi S2

\pgfplotstableread{
	beta       Bayes       DebiasBayes   DebiasLasso
	0          0.02842996  0.09851303    0.08989901
	1          0.22966045  0.11707309    0.11997834
	2          0.33936494  0.12602572    0.13516457
	3          0.38408620  0.13598380    0.16585410
	4          0.30282318  0.12907162    0.19242741
	5          0.15056226  0.11501320    0.15130191
}\RMSEponohomonSigmatwo  %p100 n100 homon S2

\pgfplotstableread{
	beta       Bayes       DebiasBayes   DebiasLasso
	0          0.0154824   0.2169399     0.2146040
	1          0.2563912   0.2867962     0.2976811
	2          0.4870332   0.2355932     0.2488787
	3          0.7215580   0.2788700     0.2581164
	4          0.9449017   0.2617744     0.3629847
	5          0.4860447   0.2361386     0.2545549
}\RMSEptnoheterSigmatwo  %p200 n100 heter S2

\pgfplotstableread{
	beta       Bayes       DebiasBayes   DebiasLasso
	0          0.02227146  0.2315107     0.2322819
	1          0.25263231  0.2660762     0.2803563
	2          0.49126995  0.2557743     0.2762522
	3          0.72526834  0.2825386     0.2575012
	4          0.95687267  0.2664669     0.3690936
	5          0.49710271  0.2577885     0.2759601
}\RMSEptnohomochiSigmatwo  %p200 n100 homochi S2

\pgfplotstableread{
	beta       Bayes       DebiasBayes   DebiasLasso
	0          0.03176429  0.09775338    0.09127577
	1          0.23455020  0.11603385    0.11992921
	2          0.38856441  0.13797691    0.13812807
	3          0.48110861  0.15255231    0.17382285
	4          0.40918779  0.14570092    0.20269544
	5          0.19329238  0.12088637    0.15347705
}\RMSEptnohomonSigmatwo  %p200 n100 homon S2

We compare our procedure with both standard Bayesian inference based on the spike-and-slab prior and a frequentist benchmark, i.e., the debiased inference procedure using a LASSO pilot estimator \citep{VandeGeer_OnAsymptotically_2014}. The methods under comparison are summarized as follows.

\begin{table}[H]
\centering
\caption{Methods Compared in the Simulation Study}
\label{tab:methods}
\renewcommand{\arraystretch}{1.3}
\begin{tabular}{@{}>{\raggedright\arraybackslash}p{4cm}>{\raggedright\arraybackslash}p{10cm}@{}}
\toprule
\textbf{Method} & \textbf{Description} \\ 
\midrule
\textbf{Bayes} 
& Standard Bayesian inference using the spike-and-slab prior, 
where the posterior is approximated via variational Bayes. \\[0.2cm]
\textbf{Posterior-Corrected} 
& Our Posterior-Corrected inference procedure. \\[0.2cm]
\textbf{Debiased-LASSO} 
& The Debiased-LASSO estimator of \citet{VandeGeer_OnAsymptotically_2014}. \\
\bottomrule
\end{tabular}
\end{table}

\noindent\textbf{Simulation design.} We fix the sample size at $n = 100$ and vary the ambient dimensionality of the covariates by setting $p = 50, 100,$ and $200$, respectively. The true regression coefficient vector $\beta_0$ is sparse with $5$ nonzero entries. Specifically, the first five elements of $\beta_0$ are set to $(0.25, 0.5, 0.75, 1, 2)$, while the remaining coefficients are zero.

\noindent
\textbf{Data-generating processes.}
We consider six simulation scenarios designed to examine the robustness of each method:
\begin{itemize}
\item[] \textbf{S1 (Homoskedastic Normal):} $\varepsilon_i \sim N(0, 1)$ for all $i = 1, \ldots, n$.
\item[] \textbf{S2 (Non-Normal Errors):} $\varepsilon_i \sim \chi^2(3) - 3$, i.e., centered chi-squared noise.
\item[] \textbf{S3 (Heteroskedastic Normal):} $\varepsilon_i \sim N(0, \sigma_{i}^2)$, where $\sigma_i = 1 + |X_{i,1}|$.
\end{itemize}
The heteroskedastic specification in S3 is similar to that in Section 4.2 of \citet{HouMaWang2023CompositeQuantile}. Settings S2 and S3 are particularly relevant to Remark~\ref{Rem:NonGaussian}. Setting S2 violates Gaussianity, whereas setting S3 introduces conditional heteroskedasticity and violates the independence between $X_i$ and $\varepsilon_i$ imposed in the benchmark Assumption~\ref{Assump:Error}. These designs therefore provide complementary finite-sample evidence beyond the benchmark primitive specification. For S1–S3, the regressors are drawn independently as $X_i \sim N(0, \Theta_0^{-1})$, where $\Theta_0$ is a $p\times p$ diagonal matrix with entries $(1, 2, \ldots, p)$ on the diagonal.

\begin{itemize}
\item[] \textbf{S4–S6 (Correlated Covariates):}
The error terms are specified as in S1–S3, but the covariates $X_i$ are sampled from a correlated Gaussian distribution $N(0, \Theta_0^{-1})$, where $\Theta_0$ is a banded precision matrix with its $ij$th element defined by
\[
\theta_{0,ij} =
\begin{cases}
	1, & \text{if } i = j, \\
	0.5, & \text{if } |i - j| = 1, \\
	0, & \text{otherwise}.
\end{cases}
\]
\end{itemize}

\noindent
\textbf{Implementation details.}
For the standard Bayesian method, we employ the variational Bayes algorithm of \citet{RaySzabo2022Variational} to approximate the posterior distribution. Each posterior sample consists of $B = 8{,}000$ draws from the variational posterior. Each simulation scenario is replicated 1,000 times. The spike-and-slab prior uses a Laplace slab with $\lambda = 1$ and a Beta$(1, p^u)$ hyperprior on the inclusion probability with $u = 1$. No additional tuning is performed. For the precision matrix estimator $\hat{\Theta}_n$, we adopt the nodewise LASSO regression of \citet{VandeGeer_OnAsymptotically_2014}, as implemented in the R package \texttt{hdi} \citep{DBMM2015hdi}. This estimator is used for both the Posterior-Corrected and the Debiased-LASSO methods to ensure comparability. Our numerical experiments suggest that the default tuning parameters in the package yield stable and robust performance. We anticipate that further optimization of hyper-parameters could potentially improve finite-sample performance. Such exploration is beyond the scope of the present study.

Figures \ref{Fig: one} and \ref{Fig: two} illustrate the empirical coverage probabilities of the $95\%$  credible/confidence intervals across the simulation scenarios. Specifically, Figure \ref{Fig: one} reports results under specifications S1 (first column), S2 (second column), and S3 (third column), with each row corresponding to a different dimensionality setting for $p = 50$, $p = 100$, and $p = 200$. On the horizontal axis, the label 0 represents the average coverage across all zero coefficients in $\beta_0$, whereas labels 1–5 correspond to the coverage probabilities for the nonzero coefficients $(0.25, 0.5, 0.75, 1, 2)$, respectively. Figure \ref{Fig: two} displays the analogous coverage results for specifications S4–S6, following the same layout and interpretation. We also compare the estimation accuracy of the mean of the corrected draws with the uncorrected posterior mean and the Debiased-LASSO point estimator. Figures \ref{Fig: three} and \ref{Fig: four} report the empirical bias of the three estimators, while Figures \ref{Fig: five} and \ref{Fig: six} present the corresponding RMSE.

\begin{figure}[!ht]
\centering\scriptsize

\begin{tikzpicture} % columns*rows
	\begin{groupplot}[group style={group name=myplots,group size=3 by 3,horizontal sep= 0.8cm,vertical sep=1.1cm},
		grid = minor,
		width = 0.375\textwidth,
		xmax=5,xmin=0,
		ymax=1,ymin=0,
		every axis title/.style={below,at={(0.2,0.8)}},
		xlabel=$\beta_0$,
		x label style={at={(axis description cs:0.95,0.04)},anchor=south},
		xtick={0,1,2,3,4,5},
		ytick={0,0.5,0.95,1},
		tick label style={/pgf/number format/fixed},
		legend style={text=black,cells={align=center},row sep = 3pt,legend columns = -1, draw=none,fill=none},
		cycle list={%
			%{smooth,tension=0.5,color=NavyBlue, no markers,line width=0.25pt, densely dotted}, % alpha
			{smooth,tension=0,color=black, mark=halfsquare*,every mark/.append style={rotate=270},mark size=1.5pt,line width=0.5pt},% LSW-Op
			{smooth,tension=0,color=blue, mark=halfsquare*,every mark/.append style={rotate=90},mark size=1.5pt,line width=0.5pt}, % LSW-Un
			{smooth,tension=0,color=red, mark=10-pointed star,mark size=1.5pt,line width=0.5pt},% Chetverikov
			{smooth,tension=0,color=RoyalBlue1, mark=halfcircle*,every mark/.append style={rotate=90},mark size=1.5pt,line width=0.5pt}, % Quad,Knots 0/// Cub,Knots 3
			{smooth,tension=0,color=RoyalBlue2, mark=halfcircle*,every mark/.append style={rotate=180},mark size=1.5pt,line width=0.5pt}, % Quad,Knots/// Cub,Knots 5
			{smooth,tension=0,color=RoyalBlue3, mark=halfcircle*,every mark/.append style={rotate=270},mark size=1.5pt,line width=0.5pt},% Cub,Knots 0/// Cub,Knots 7
			{smooth,tension=0,color=RoyalBlue4, mark=halfcircle*,every mark/.append style={rotate=360},mark size=1.5pt,line width=0.5pt},% Cub,Knots 1
		}
		]
		% Monotonicity for the univariate case
		\nextgroupplot
		\node[anchor=north] at (axis description cs: 0.25,  0.95) {\fontsize{5}{4}\selectfont \shortstack{
				\\  \\
				Homo-Normal\\
				$p=50$
		}};
		\addplot[smooth,tension=0.5,color=NavyBlue, no markers,line width=0.25pt, densely dotted,forget plot] table[x = beta,y=alpha] from \CovpfnohomonSigmaone;
		\addplot table[x = beta,y=Bayes] from \CovpfnohomonSigmaone;
		\addplot table[x = beta,y=DebiasBayes] from \CovpfnohomonSigmaone;
		\addplot table[x = beta,y=DebiasLasso] from \CovpfnohomonSigmaone;
		
		\nextgroupplot

		\node[anchor=north] at (axis description cs: 0.25,  0.95) {\fontsize{5}{4}\selectfont \shortstack{
				\\  \\
				Homo-Chi\\
				$p=50$
		}};
		\addplot[smooth,tension=0.5,color=NavyBlue, no markers,line width=0.25pt, densely dotted,forget plot] table[x = beta,y=alpha] from \CovpfnohomochiSigmaone;
		\addplot table[x = beta,y=Bayes] from \CovpfnohomochiSigmaone;
		\addplot table[x = beta,y=DebiasBayes] from \CovpfnohomochiSigmaone;
		\addplot table[x = beta,y=DebiasLasso] from \CovpfnohomochiSigmaone;

		\nextgroupplot[legend style = {column sep = 7pt, legend to name = LegendMon1}]
		\addplot[smooth,tension=0.5,color=NavyBlue, no markers,line width=0.25pt, densely dotted,forget plot] table[x = beta,y=alpha] from \CovpfnoheterSigmaone;
		\addplot table[x = beta,y=Bayes] from \CovpfnoheterSigmaone;
		\addplot table[x = beta,y=DebiasBayes] from \CovpfnoheterSigmaone;
		\addplot table[x = beta,y=DebiasLasso] from \CovpfnoheterSigmaone;
		\node[anchor=north] at (axis description cs: 0.25,  0.95) {\fontsize{5}{4}\selectfont \shortstack{
				\\  \\
				Hetero\\
				$p=50$
		}};
		
		% Monotonicity for the bivariate case
		
		\nextgroupplot
		\node[anchor=north] at (axis description cs: 0.25,  0.95) {\fontsize{5}{4}\selectfont \shortstack{
				\\  \\
				Homo-Normal\\
				$p=100$
		}};
		\addplot[smooth,tension=0.5,color=NavyBlue, no markers,line width=0.25pt, densely dotted,forget plot] table[x = beta,y=alpha] from \CovponohomonSigmaone;
		\addplot table[x = beta,y=Bayes] from \CovponohomonSigmaone;
		\addplot table[x = beta,y=DebiasBayes] from \CovponohomonSigmaone;
		\addplot table[x = beta,y=DebiasLasso] from \CovponohomonSigmaone;

		\nextgroupplot
		\node[anchor=north] at (axis description cs: 0.25,  0.95) {\fontsize{5}{4}\selectfont \shortstack{
				\\  \\
				Homo-Chi\\
				$p=100$
		}};
		\addplot[smooth,tension=0.5,color=NavyBlue, no markers,line width=0.25pt, densely dotted,forget plot] table[x = beta,y=alpha] from \CovponohomochiSigmaone;
		\addplot table[x = beta,y=Bayes] from \CovponohomochiSigmaone;
		\addplot table[x = beta,y=DebiasBayes] from \CovponohomochiSigmaone;
		\addplot table[x = beta,y=DebiasLasso] from \CovponohomochiSigmaone;
		
		\nextgroupplot[legend style = {column sep = 3.5pt, legend to name = LegendMon2}]
		
		\node[anchor=north] at (axis description cs: 0.25,  0.95) {\fontsize{5}{4}\selectfont \shortstack{
				\\  \\
				Hetero\\
				$p=100$
		}};
		\addplot[smooth,tension=0.5,color=NavyBlue, no markers,line width=0.25pt, densely dotted,forget plot] table[x = beta,y=alpha] from \CovponoheterSigmaone;
		\addplot table[x = beta,y=Bayes] from \CovponoheterSigmaone;
		\addplot table[x = beta,y=DebiasBayes] from \CovponoheterSigmaone;
		\addplot table[x = beta,y=DebiasLasso] from \CovponoheterSigmaone;

			\nextgroupplot
		\node[anchor=north] at (axis description cs: 0.25,  0.95) {\fontsize{5}{4}\selectfont \shortstack{
				\\  \\
				Homo-Normal\\
				$p=200$
		}};
		\addplot[smooth,tension=0.5,color=NavyBlue, no markers,line width=0.25pt, densely dotted,forget plot] table[x = beta,y=alpha] from \ptonohomon;
		\addplot table[x = beta,y=Bayes] from \CovptnohomonSigmaone;
		\addplot table[x = beta,y=DebiasBayes] from \CovptnohomonSigmaone;
		\addplot table[x = beta,y=DebiasLasso] from \CovptnohomonSigmaone;

		\nextgroupplot
		\node[anchor=north] at (axis description cs: 0.25,  0.95) {\fontsize{5}{4}\selectfont \shortstack{
				\\  \\
				Homo-Chi\\
				$p=200$
		}};
		\addplot[smooth,tension=0.5,color=NavyBlue, no markers,line width=0.25pt, densely dotted,forget plot] table[x = beta,y=alpha] from \CovptnohomochiSigmaone;
		\addplot table[x = beta,y=Bayes] from \CovptnohomochiSigmaone;
		\addplot table[x = beta,y=DebiasBayes] from \CovptnohomochiSigmaone;
		\addplot table[x = beta,y=DebiasLasso] from \CovptnohomochiSigmaone;
		
		\nextgroupplot[legend style = {column sep = 3.5pt, legend to name = LegendMon3}]
		
		\node[anchor=north] at (axis description cs: 0.25,  0.95) {\fontsize{5}{4}\selectfont \shortstack{
				\\  \\
				Hetero\\
				$p=200$
		}};
		\addplot[smooth,tension=0.5,color=NavyBlue, no markers,line width=0.25pt, densely dotted,forget plot] table[x = beta,y=alpha] from \CovptnoheterSigmaone;
		\addplot table[x = beta,y=Bayes] from \CovptnoheterSigmaone;
		\addplot table[x = beta,y=DebiasBayes] from \CovptnoheterSigmaone;
		\addplot table[x = beta,y=DebiasLasso] from \CovptnoheterSigmaone;
		\addlegendentry{Bayes}
		\addlegendentry{Posterior-Corrected}
		\addlegendentry{Debiased-LASSO}

	\end{groupplot}
	\node at ($(myplots c2r1) + (0,-2.25cm)$) {\ref{LegendMon1}};
	\node at ($(myplots c2r2) + (0,-2.25cm)$) {\ref{LegendMon2}};
	\node at ($(myplots c2r3) + (0,-2.25cm)$) {\ref{LegendMon3}};
\end{tikzpicture}
\caption{Coverage rates corresponding to different values of $\beta_0$ under settings S1 (first column), S2 (second column), and S3 (third column). The dashed line represents the 95\% benchmark.} \label{Fig: one}
\end{figure}
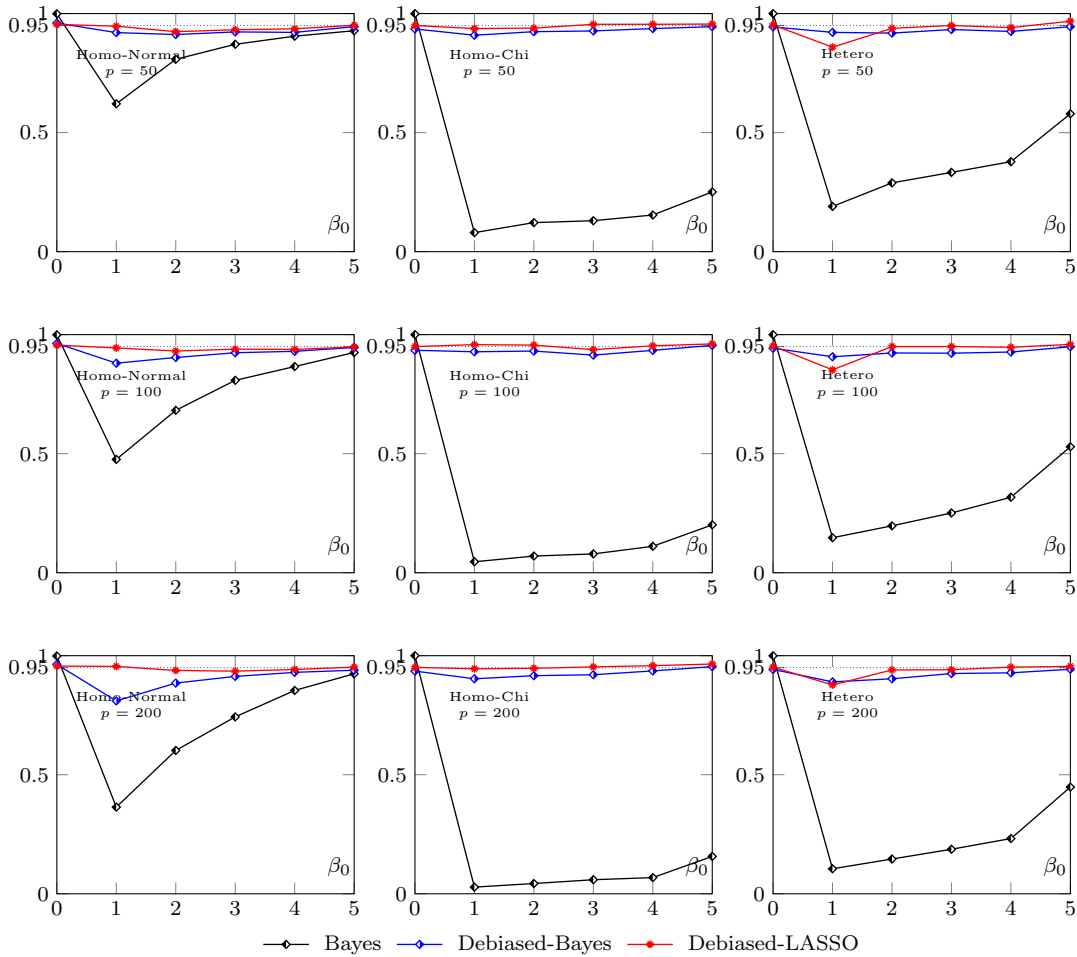

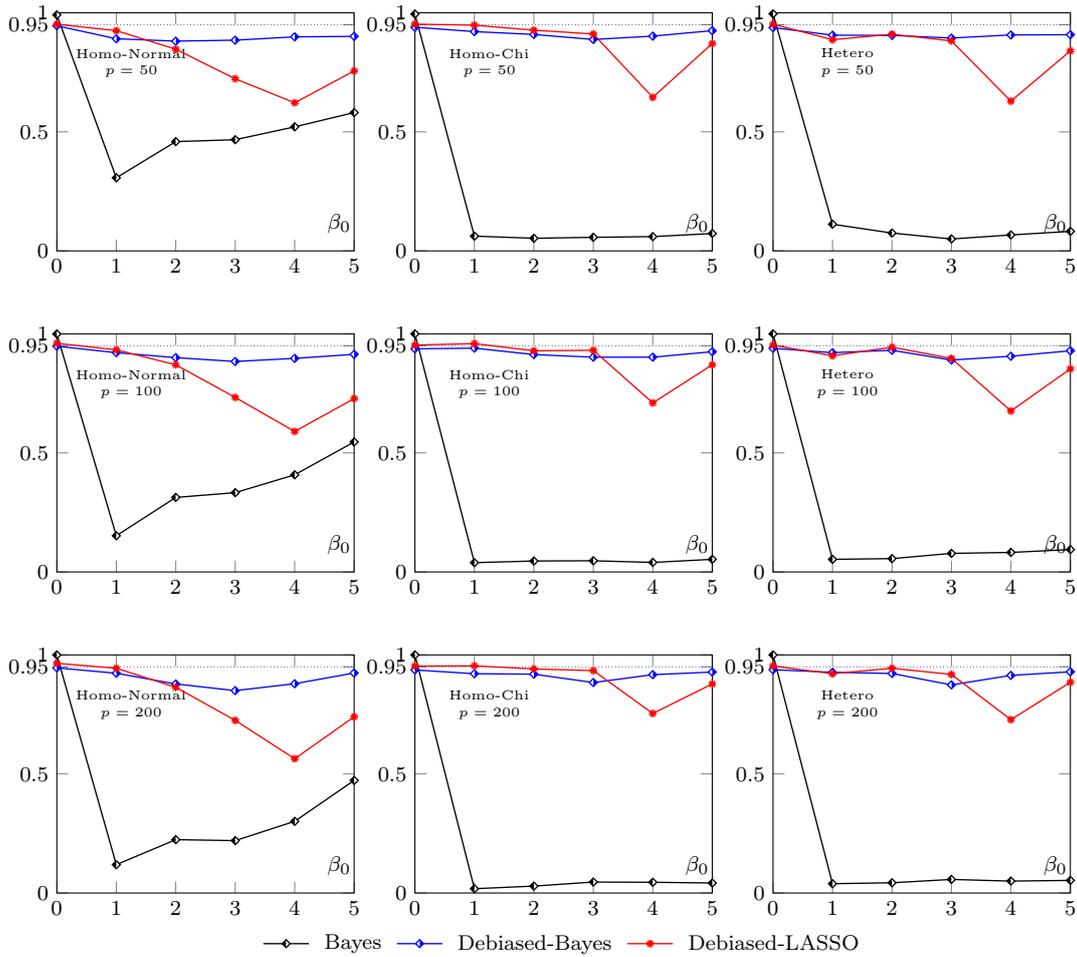
\begin{figure}[!ht]
	\centering\scriptsize
	\begin{tikzpicture} % columns*rows
		\begin{groupplot}[group style={group name=myplots,group size=3 by 3,horizontal sep= 0.8cm,vertical sep=1.1cm},
			grid = minor,
			width = 0.375\textwidth,
			xmax=5,xmin=0,
			ymax=1,ymin=0,
			every axis title/.style={below,at={(0.2,0.8)}},
			xlabel=$\beta_0$,
			x label style={at={(axis description cs:0.95,0.04)},anchor=south},
			xtick={0,1,2,3,4,5},
			ytick={0,0.5,0.95,1},
			tick label style={/pgf/number format/fixed},
			legend style={text=black,cells={align=center},row sep = 3pt,legend columns = -1, draw=none,fill=none},
			cycle list={%
				%{smooth,tension=0.5,color=NavyBlue, no markers,line width=0.25pt, densely dotted}, % alpha
				{smooth,tension=0,color=black, mark=halfsquare*,every mark/.append style={rotate=270},mark size=1.5pt,line width=0.5pt},% LSW-Op
				{smooth,tension=0,color=blue, mark=halfsquare*,every mark/.append style={rotate=90},mark size=1.5pt,line width=0.5pt}, % LSW-Un
				{smooth,tension=0,color=red, mark=10-pointed star,mark size=1.5pt,line width=0.5pt},% Chetverikov
				{smooth,tension=0,color=RoyalBlue1, mark=halfcircle*,every mark/.append style={rotate=90},mark size=1.5pt,line width=0.5pt}, % Quad,Knots 0/// Cub,Knots 3
				{smooth,tension=0,color=RoyalBlue2, mark=halfcircle*,every mark/.append style={rotate=180},mark size=1.5pt,line width=0.5pt}, % Quad,Knots/// Cub,Knots 5
				{smooth,tension=0,color=RoyalBlue3, mark=halfcircle*,every mark/.append style={rotate=270},mark size=1.5pt,line width=0.5pt},% Cub,Knots 0/// Cub,Knots 7
				{smooth,tension=0,color=RoyalBlue4, mark=halfcircle*,every mark/.append style={rotate=360},mark size=1.5pt,line width=0.5pt},% Cub,Knots 1
			}
			]
			% Monotonicity for the univariate case
			
			\nextgroupplot
			\node[anchor=north] at (axis description cs: 0.25,  0.95) {\fontsize{5}{4}\selectfont \shortstack{
					\\  \\
					Homo-Normal\\
					$p=50$
			}};
			\addplot[smooth,tension=0.5,color=NavyBlue, no markers,line width=0.25pt, densely dotted,forget plot] table[x = beta,y=alpha] from \pfnohomon;
			\addplot table[x = beta,y=Bayes] from \pfnohomon;
			\addplot table[x = beta,y=DebiasBayes] from \pfnohomon;
			\addplot table[x = beta,y=DebiasLasso] from \pfnohomon;

			\nextgroupplot

			\node[anchor=north] at (axis description cs: 0.25,  0.95) {\fontsize{5}{4}\selectfont \shortstack{
					\\  \\
					Homo-Chi\\
					$p=50$
			}};
			\addplot[smooth,tension=0.5,color=NavyBlue, no markers,line width=0.25pt, densely dotted,forget plot] table[x = beta,y=alpha] from \pfnohomochi;
			\addplot table[x = beta,y=Bayes] from \pfnohomochi;
			\addplot table[x = beta,y=DebiasBayes] from \pfnohomochi;
			\addplot table[x = beta,y=DebiasLasso] from \pfnohomochi;

			\nextgroupplot[legend style = {column sep = 7pt, legend to name = LegendMon12}]
			\addplot[smooth,tension=0.5,color=NavyBlue, no markers,line width=0.25pt, densely dotted,forget plot] table[x = beta,y=alpha] from \pfnoheter;
			\addplot table[x = beta,y=Bayes] from \pfnoheter;
			\addplot table[x = beta,y=DebiasBayes] from \pfnoheter;
			\addplot table[x = beta,y=DebiasLasso] from \pfnoheter;
			\node[anchor=north] at (axis description cs: 0.25,  0.95) {\fontsize{5}{4}\selectfont \shortstack{
					\\  \\
					Hetero\\
					$p=50$
			}};
			
			% Monotonicity for the bivariate case
			
			\nextgroupplot
			\node[anchor=north] at (axis description cs: 0.25,  0.95) {\fontsize{5}{4}\selectfont \shortstack{
					\\  \\
					Homo-Normal\\
					$p=100$
			}};
			\addplot[smooth,tension=0.5,color=NavyBlue, no markers,line width=0.25pt, densely dotted,forget plot] table[x = beta,y=alpha] from \ponohomon;
			\addplot table[x = beta,y=Bayes] from \ponohomon;
			\addplot table[x = beta,y=DebiasBayes] from \ponohomon;
			\addplot table[x = beta,y=DebiasLasso] from \ponohomon;

			\nextgroupplot
			\node[anchor=north] at (axis description cs: 0.25,  0.95) {\fontsize{5}{4}\selectfont \shortstack{
					\\  \\
					Homo-Chi\\
					$p=100$
			}};
			\addplot[smooth,tension=0.5,color=NavyBlue, no markers,line width=0.25pt, densely dotted,forget plot] table[x = beta,y=alpha] from \ponohomochi;
			\addplot table[x = beta,y=Bayes] from \ponohomochi;
			\addplot table[x = beta,y=DebiasBayes] from \ponohomochi;
			\addplot table[x = beta,y=DebiasLasso] from \ponohomochi;
			
				\nextgroupplot[legend style = {column sep = 3.5pt, legend to name = LegendMon22}]
			
			\node[anchor=north] at (axis description cs: 0.25,  0.95) {\fontsize{5}{4}\selectfont \shortstack{
					\\  \\
					Hetero\\
					$p=100$
			}};
			\addplot[smooth,tension=0.5,color=NavyBlue, no markers,line width=0.25pt, densely dotted,forget plot] table[x = beta,y=alpha] from \ponoheter;
			\addplot table[x = beta,y=Bayes] from \ponoheter;
			\addplot table[x = beta,y=DebiasBayes] from \ponoheter;
			\addplot table[x = beta,y=DebiasLasso] from \ponoheter;
			
			\nextgroupplot
			\node[anchor=north] at (axis description cs: 0.25,  0.95) {\fontsize{5}{4}\selectfont \shortstack{
					\\  \\
					Homo-Normal\\
					$p=200$
			}};
			\addplot[smooth,tension=0.5,color=NavyBlue, no markers,line width=0.25pt, densely dotted,forget plot] table[x = beta,y=alpha] from \ptonohomon;
			\addplot table[x = beta,y=Bayes] from \ptonohomon;
			\addplot table[x = beta,y=DebiasBayes] from \ptonohomon;
			\addplot table[x = beta,y=DebiasLasso] from \ptonohomon;

			\nextgroupplot
			\node[anchor=north] at (axis description cs: 0.25,  0.95) {\fontsize{5}{4}\selectfont \shortstack{
					\\  \\
					Homo-Chi\\
					$p=200$
			}};
			\addplot[smooth,tension=0.5,color=NavyBlue, no markers,line width=0.25pt, densely dotted,forget plot] table[x = beta,y=alpha] from \ptonohomochi;
			\addplot table[x = beta,y=Bayes] from \ptonohomochi;
			\addplot table[x = beta,y=DebiasBayes] from \ptonohomochi;
			\addplot table[x = beta,y=DebiasLasso] from \ptonohomochi;
			
			\nextgroupplot[legend style = {column sep = 3.5pt, legend to name = LegendMon32}]
			
			\node[anchor=north] at (axis description cs: 0.25,  0.95) {\fontsize{5}{4}\selectfont \shortstack{
					\\  \\
					Hetero\\
					$p=200$
			}};
			\addplot[smooth,tension=0.5,color=NavyBlue, no markers,line width=0.25pt, densely dotted,forget plot] table[x = beta,y=alpha] from \ptonoheter;
			\addplot table[x = beta,y=Bayes] from \ptonoheter;
			\addplot table[x = beta,y=DebiasBayes] from \ptonoheter;
			\addplot table[x = beta,y=DebiasLasso] from \ptonoheter;
			\addlegendentry{Bayes}
			\addlegendentry{Posterior-Corrected}
			\addlegendentry{Debiased-LASSO}

		\end{groupplot}
		\node at ($(myplots c2r1) + (0,-2.25cm)$) {\ref{LegendMon12}};
		\node at ($(myplots c2r2) + (0,-2.25cm)$) {\ref{LegendMon22}};
		\node at ($(myplots c2r3) + (0,-2.25cm)$) {\ref{LegendMon32}};
	\end{tikzpicture}
	\caption{Coverage rates corresponding to different values of $\beta_0$ under settings S4 (first column), S5 (second column), and S6 (third column). The dashed line represents the 95\% benchmark.} \label{Fig: two}
\end{figure}

We summarize our findings as follows. The standard Bayesian approach based on the spike-and-slab prior performs well in identifying the zero coefficients, as reflected by favorable average coverage, low bias, and small RMSE for these components. However, its performance deteriorates substantially for the nonzero coefficients,\footnote{By contrast, \citet{RaySzabo2022Variational} considered much larger nonzero entries of $\beta_0$ (as large as 10), whereas ours are more moderate.} and becomes more fragile when the error distribution departs from the baseline homoscedastic Gaussian specification. In contrast, our posterior-correction procedure consistently outperforms the standard method, achieving markedly better estimation accuracy and more reliable uncertainty quantification across all designs. 

Relative to the frequentist Debiased-LASSO benchmark, our method delivers comparable performance with independent covariates (S1–S3). Notably, in the more challenging scenarios involving correlated covariates (S4–S6), the posterior-correction procedure tends to attain substantially improved coverage, lower bias, and smaller RMSE for relatively larger coefficients. For scenario S4 with $p=
100$, the finite-sample coverage for the coefficient (true value = 1) is only $59\%$ under the Debiased-LASSO method, while our approach achieves a notably higher coverage of $89\%$. In scenario S5 with $p=200$, the finite-sample bias for the same coefficient is reduced from $0.298$ with Debiased-LASSO to $0.131$ with our Posterior-Corrected method, which cuts the bias by more than half. Correspondingly, the RMSEs in this setting are $0.262$ for Posterior-Corrected compared to $0.363$ for Debiased-LASSO, further highlighting the finite-sample efficiency gains of our method. Overall, these results demonstrate the robustness and accuracy of the posterior-correction procedure across diverse conditions and the importance of debiasing initial posterior draws in high-dimensional settings. It provides clear improvements over conventional Bayesian inference and exhibits complementary strengths relative to the frequentist benchmark.

\begin{figure}[!ht]
	\centering\scriptsize
	\begin{tikzpicture} % columns*rows
		\begin{groupplot}[group style={group name=myplots,group size=3 by 3,horizontal sep= 0.8cm,vertical sep=1.1cm},
			grid = minor,
			width = 0.375\textwidth,
			xmax=5,xmin=0,
			ymax=1.66,ymin=0,
			every axis title/.style={below,at={(0.2,0.8)}},
			xlabel=$\beta_0$,
			x label style={at={(axis description cs:0.95,0.04)},anchor=south},
			xtick={0,1,2,3,4,5},
			ytick={0,0.5,1,1.5},
			tick label style={/pgf/number format/fixed},
			legend style={text=black,cells={align=center},row sep = 3pt,legend columns = -1, draw=none,fill=none},
			cycle list={%
				%{smooth,tension=0.5,color=NavyBlue, no markers,line width=0.25pt, densely dotted}, % alpha
				{smooth,tension=0,color=black, mark=halfsquare*,every mark/.append style={rotate=270},mark size=1.5pt,line width=0.5pt},% LSW-Op
				{smooth,tension=0,color=blue, mark=halfsquare*,every mark/.append style={rotate=90},mark size=1.5pt,line width=0.5pt}, % LSW-Un
				{smooth,tension=0,color=red, mark=10-pointed star,mark size=1.5pt,line width=0.5pt},% Chetverikov
				{smooth,tension=0,color=RoyalBlue1, mark=halfcircle*,every mark/.append style={rotate=90},mark size=1.5pt,line width=0.5pt}, % Quad,Knots 0/// Cub,Knots 3
				{smooth,tension=0,color=RoyalBlue2, mark=halfcircle*,every mark/.append style={rotate=180},mark size=1.5pt,line width=0.5pt}, % Quad,Knots/// Cub,Knots 5
				{smooth,tension=0,color=RoyalBlue3, mark=halfcircle*,every mark/.append style={rotate=270},mark size=1.5pt,line width=0.5pt},% Cub,Knots 0/// Cub,Knots 7
				{smooth,tension=0,color=RoyalBlue4, mark=halfcircle*,every mark/.append style={rotate=360},mark size=1.5pt,line width=0.5pt},% Cub,Knots 1
			}
			]
			% Monotonicity for the univariate case
			\nextgroupplot
			\node[anchor=north] at (axis description cs: 0.25,  0.95) {\fontsize{5}{4}\selectfont \shortstack{
					\\  \\
					Homo-Normal\\
					$p=50$
			}};
			\addplot table[x = beta,y=Bayes] from \BiaspfnohomonSigmaone;
			\addplot table[x = beta,y=DebiasBayes] from \BiaspfnohomonSigmaone;
			\addplot table[x = beta,y=DebiasLasso] from \BiaspfnohomonSigmaone;
			
			\nextgroupplot

			\node[anchor=north] at (axis description cs: 0.25,  0.95) {\fontsize{5}{4}\selectfont \shortstack{
					\\  \\
					Homo-Chi\\
					$p=50$
			}};
			\addplot table[x = beta,y=Bayes] from \BiaspfnohomochiSigmaone;
			\addplot table[x = beta,y=DebiasBayes] from \BiaspfnohomochiSigmaone;
			\addplot table[x = beta,y=DebiasLasso] from \BiaspfnohomochiSigmaone;

			\nextgroupplot[legend style = {column sep = 7pt, legend to name = LegendMon13}]
			\addplot table[x = beta,y=Bayes] from \BiaspfnoheterSigmaone;
			\addplot table[x = beta,y=DebiasBayes] from \BiaspfnoheterSigmaone;
			\addplot table[x = beta,y=DebiasLasso] from \BiaspfnoheterSigmaone;
			\node[anchor=north] at (axis description cs: 0.25,  0.95) {\fontsize{5}{4}\selectfont \shortstack{
					\\  \\
					Hetero\\
					$p=50$
			}};
			
			% Monotonicity for the bivariate case
			\nextgroupplot
			\node[anchor=north] at (axis description cs: 0.25,  0.95) {\fontsize{5}{4}\selectfont \shortstack{
					\\  \\
					Homo-Normal\\
					$p=100$
			}};
			\addplot table[x = beta,y=Bayes] from \BiasponohomonSigmaone;
			\addplot table[x = beta,y=DebiasBayes] from \BiasponohomonSigmaone;
			\addplot table[x = beta,y=DebiasLasso] from \BiasponohomonSigmaone;

			\nextgroupplot
			\node[anchor=north] at (axis description cs: 0.25,  0.95) {\fontsize{5}{4}\selectfont \shortstack{
					\\  \\
					Homo-Chi\\
					$p=100$
			}};
			\addplot table[x = beta,y=Bayes] from \BiasponohomochiSigmaone;
			\addplot table[x = beta,y=DebiasBayes] from \BiasponohomochiSigmaone;
			\addplot table[x = beta,y=DebiasLasso] from \BiasponohomochiSigmaone;
			
			\nextgroupplot[legend style = {column sep = 3.5pt, legend to name = LegendMon23}]
			
			\node[anchor=north] at (axis description cs: 0.25,  0.95) {\fontsize{5}{4}\selectfont \shortstack{
					\\  \\
					Hetero\\
					$p=100$
			}};
			\addplot table[x = beta,y=Bayes] from \BiasponoheterSigmaone;
			\addplot table[x = beta,y=DebiasBayes] from \BiasponoheterSigmaone;
			\addplot table[x = beta,y=DebiasLasso] from \BiasponoheterSigmaone;
			
			\nextgroupplot
			\node[anchor=north] at (axis description cs: 0.25,  0.95) {\fontsize{5}{4}\selectfont \shortstack{
					\\  \\
					Homo-Normal\\
					$p=200$
			}};
			
			\addplot table[x = beta,y=Bayes] from \BiasptnohomonSigmaone;
			\addplot table[x = beta,y=DebiasBayes] from \BiasptnohomonSigmaone;
			\addplot table[x = beta,y=DebiasLasso] from \BiasptnohomonSigmaone;
			
			\nextgroupplot
			\node[anchor=north] at (axis description cs: 0.25,  0.95) {\fontsize{5}{4}\selectfont \shortstack{
					\\  \\
					Homo-Chi\\
					$p=200$
			}};
			\addplot table[x = beta,y=Bayes] from \BiasptnohomochiSigmaone;
			\addplot table[x = beta,y=DebiasBayes] from \BiasptnohomochiSigmaone;
			\addplot table[x = beta,y=DebiasLasso] from \BiasptnohomochiSigmaone;
			
			\nextgroupplot[legend style = {column sep = 3.5pt, legend to name = LegendMon33}]
			
			\node[anchor=north] at (axis description cs: 0.25,  0.95) {\fontsize{5}{4}\selectfont \shortstack{
					\\  \\
					Heter\\
					$p=200$
			}};
			
			\addplot table[x = beta,y=Bayes] from \BiasptnoheterSigmaone;
			\addplot table[x = beta,y=DebiasBayes] from \BiasptnoheterSigmaone;
			\addplot table[x = beta,y=DebiasLasso] from \BiasptnoheterSigmaone;
			\addlegendentry{Bayes}
			\addlegendentry{Posterior-Corrected}
			\addlegendentry{Debiased-LASSO}

		\end{groupplot}
		\node at ($(myplots c2r1) + (0,-2.25cm)$) {\ref{LegendMon13}};
		\node at ($(myplots c2r2) + (0,-2.25cm)$) {\ref{LegendMon23}};
		\node at ($(myplots c2r3) + (0,-2.25cm)$) {\ref{LegendMon33}};
	\end{tikzpicture}
	\caption{Bias corresponding to different values of $\beta_0$ under settings S1 (first column), S2 (second column), and S3 (third column).} \label{Fig: three}
\end{figure}
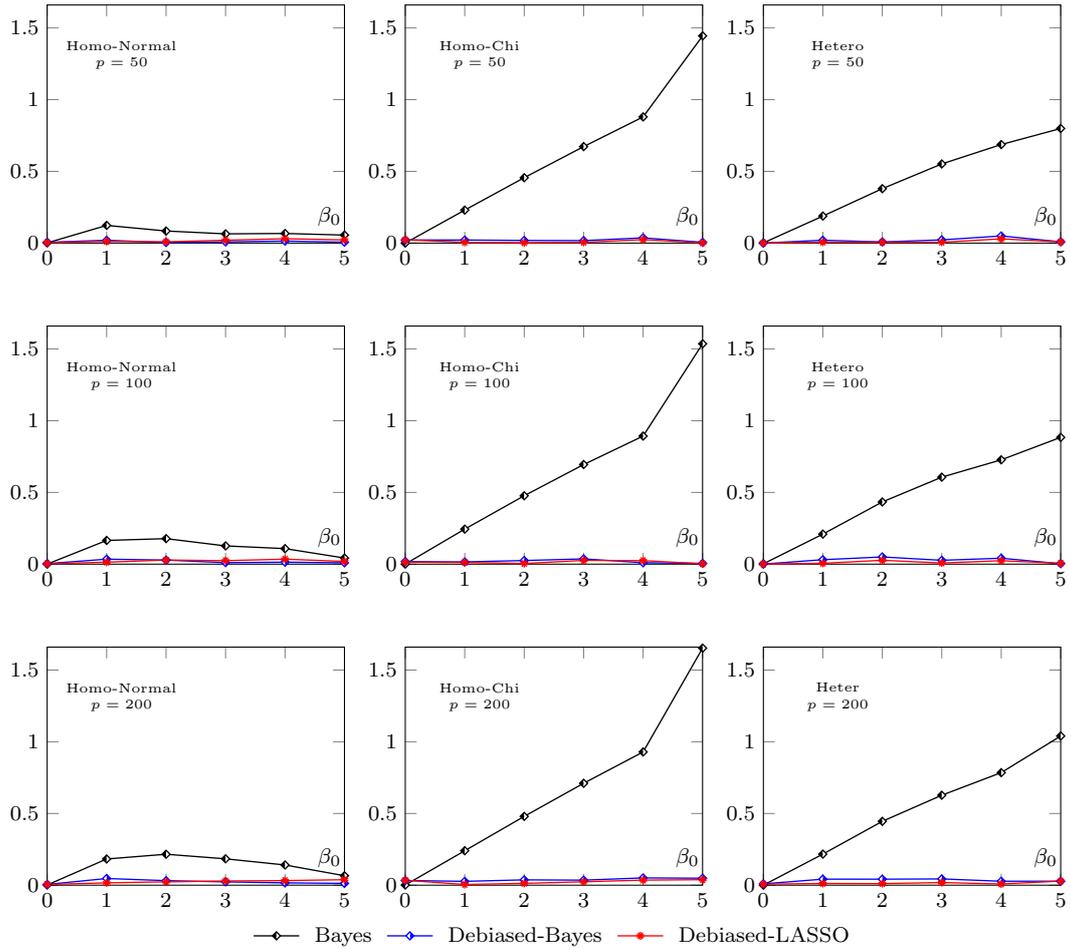

\begin{figure}[!ht]
	\centering\scriptsize
	\begin{tikzpicture} % columns*rows
		\begin{groupplot}[group style={group name=myplots,group size=3 by 3,horizontal sep= 0.8cm,vertical sep=1.1cm},
			grid = minor,
			width = 0.375\textwidth,
			xmax=5,xmin=0,
			ymax=1,ymin=0,
			every axis title/.style={below,at={(0.2,0.8)}},
			xlabel=$\beta_0$,
			x label style={at={(axis description cs:0.95,0.04)},anchor=south},
			xtick={0,1,2,3,4,5},
			ytick={0,0.5,1},
			tick label style={/pgf/number format/fixed},
			legend style={text=black,cells={align=center},row sep = 3pt,legend columns = -1, draw=none,fill=none},
			cycle list={%
				%{smooth,tension=0.5,color=NavyBlue, no markers,line width=0.25pt, densely dotted}, % alpha
				{smooth,tension=0,color=black, mark=halfsquare*,every mark/.append style={rotate=270},mark size=1.5pt,line width=0.5pt},% LSW-Op
				{smooth,tension=0,color=blue, mark=halfsquare*,every mark/.append style={rotate=90},mark size=1.5pt,line width=0.5pt}, % LSW-Un
				{smooth,tension=0,color=red, mark=10-pointed star,mark size=1.5pt,line width=0.5pt},% Chetverikov
				{smooth,tension=0,color=RoyalBlue1, mark=halfcircle*,every mark/.append style={rotate=90},mark size=1.5pt,line width=0.5pt}, % Quad,Knots 0/// Cub,Knots 3
				{smooth,tension=0,color=RoyalBlue2, mark=halfcircle*,every mark/.append style={rotate=180},mark size=1.5pt,line width=0.5pt}, % Quad,Knots/// Cub,Knots 5
				{smooth,tension=0,color=RoyalBlue3, mark=halfcircle*,every mark/.append style={rotate=270},mark size=1.5pt,line width=0.5pt},% Cub,Knots 0/// Cub,Knots 7
				{smooth,tension=0,color=RoyalBlue4, mark=halfcircle*,every mark/.append style={rotate=360},mark size=1.5pt,line width=0.5pt},% Cub,Knots 1
			}
			]
			% Monotonicity for the univariate case
			\nextgroupplot
			\node[anchor=north] at (axis description cs: 0.25,  0.95) {\fontsize{5}{4}\selectfont \shortstack{
					\\  \\
					Homo-Normal\\
					$p=50$
			}};
			\addplot table[x = beta,y=Bayes] from \BiaspfnohomonSigmatwo;
			\addplot table[x = beta,y=DebiasBayes] from \BiaspfnohomonSigmatwo;
			\addplot table[x = beta,y=DebiasLasso] from \BiaspfnohomonSigmatwo;

			\nextgroupplot

			\node[anchor=north] at (axis description cs: 0.25,  0.95) {\fontsize{5}{4}\selectfont \shortstack{
					\\  \\
					Homo-Chi\\
					$p=50$
			}};
			\addplot table[x = beta,y=Bayes] from \BiaspfnohomochiSigmatwo;
			\addplot table[x = beta,y=DebiasBayes] from \BiaspfnohomochiSigmatwo;
			\addplot table[x = beta,y=DebiasLasso] from \BiaspfnohomochiSigmatwo;
			
			\nextgroupplot[legend style = {column sep = 7pt, legend to name = LegendMon14}]
			\addplot table[x = beta,y=Bayes] from \BiaspfnoheterSigmatwo;
			\addplot table[x = beta,y=DebiasBayes] from \BiaspfnoheterSigmatwo;
			\addplot table[x = beta,y=DebiasLasso] from \BiaspfnoheterSigmatwo;
			\node[anchor=north] at (axis description cs: 0.25,  0.95) {\fontsize{5}{4}\selectfont \shortstack{
					\\  \\
					Hetero\\
					$p=50$
			}};
			
			% Monotonicity for the bivariate case
			\nextgroupplot
			\node[anchor=north] at (axis description cs: 0.25,  0.95) {\fontsize{5}{4}\selectfont \shortstack{
					\\  \\
					Homo-Normal\\
					$p=100$
			}};
			\addplot table[x = beta,y=Bayes] from \BiasponohomonSigmatwo;
			\addplot table[x = beta,y=DebiasBayes] from \BiasponohomonSigmatwo;
			\addplot table[x = beta,y=DebiasLasso] from \BiasponohomonSigmatwo;

			\nextgroupplot
			\node[anchor=north] at (axis description cs: 0.25,  0.95) {\fontsize{5}{4}\selectfont \shortstack{
					\\  \\
					Homo-Chi\\
					$p=100$
			}};
			\addplot table[x = beta,y=Bayes] from \BiasponohomochiSigmatwo;
			\addplot table[x = beta,y=DebiasBayes] from \BiasponohomochiSigmatwo;
			\addplot table[x = beta,y=DebiasLasso] from \BiasponohomochiSigmatwo;
			
			\nextgroupplot[legend style = {column sep = 3.5pt, legend to name = LegendMon24}]
			
			\node[anchor=north] at (axis description cs: 0.25,  0.95) {\fontsize{5}{4}\selectfont \shortstack{
					\\  \\
					Heter\\
					$p=100$
			}};
			\addplot table[x = beta,y=Bayes] from \BiasponoheterSigmatwo;
			\addplot table[x = beta,y=DebiasBayes] from \BiasponoheterSigmatwo;
			\addplot table[x = beta,y=DebiasLasso] from \BiasponoheterSigmatwo;
			
			\nextgroupplot
			\node[anchor=north] at (axis description cs: 0.25,  0.95) {\fontsize{5}{4}\selectfont \shortstack{
					\\  \\
					Homo-Normal\\
					$p=200$
			}};
			
			\addplot table[x = beta,y=Bayes] from \BiasptnohomonSigmatwo;
			\addplot table[x = beta,y=DebiasBayes] from \BiasptnohomonSigmatwo;
			\addplot table[x = beta,y=DebiasLasso] from \BiasptnohomonSigmatwo;

			\nextgroupplot
			\node[anchor=north] at (axis description cs: 0.25,  0.95) {\fontsize{5}{4}\selectfont \shortstack{
					\\  \\
					Homo-Chi\\
					$p=200$
			}};
			\addplot table[x = beta,y=Bayes] from \BiasptnohomochiSigmatwo;
			\addplot table[x = beta,y=DebiasBayes] from \BiasptnohomochiSigmatwo;
			\addplot table[x = beta,y=DebiasLasso] from \BiasptnohomochiSigmatwo;
			
			\nextgroupplot[legend style = {column sep = 3.5pt, legend to name = LegendMon34}]
			
			\node[anchor=north] at (axis description cs: 0.25,  0.95) {\fontsize{5}{4}\selectfont \shortstack{
					\\  \\
					Hetero\\
					$p=200$
			}};
			
			\addplot table[x = beta,y=Bayes] from \BiasptnoheterSigmatwo;
			\addplot table[x = beta,y=DebiasBayes] from \BiasptnoheterSigmatwo;
			\addplot table[x = beta,y=DebiasLasso] from \BiasptnoheterSigmatwo;
			\addlegendentry{Bayes}
			\addlegendentry{Posterior-Corrected}
			\addlegendentry{Debiased-LASSO}

		\end{groupplot}
		\node at ($(myplots c2r1) + (0,-2.25cm)$) {\ref{LegendMon14}};
		\node at ($(myplots c2r2) + (0,-2.25cm)$) {\ref{LegendMon24}};
		\node at ($(myplots c2r3) + (0,-2.25cm)$) {\ref{LegendMon34}};
	\end{tikzpicture}
	\caption{Bias corresponding to different values of $\beta_0$ under settings S4 (first column), S5 (second column), and S6 (third column).} \label{Fig: four}
\end{figure}

\begin{figure}[!ht]
	\centering\scriptsize
	\begin{tikzpicture} % columns*rows
		\begin{groupplot}[group style={group name=myplots,group size=3 by 3,horizontal sep= 0.8cm,vertical sep=1.1cm},
			grid = minor,
			width = 0.375\textwidth,
			xmax=5,xmin=0,
			ymax=1.85,ymin=0,
			every axis title/.style={below,at={(0.2,0.8)}},
			xlabel=$\beta_0$,
			x label style={at={(axis description cs:0.95,0.04)},anchor=south},
			xtick={0,1,2,3,4,5},
			ytick={0,0.5,1,1.5},
			tick label style={/pgf/number format/fixed},
			legend style={text=black,cells={align=center},row sep = 3pt,legend columns = -1, draw=none,fill=none},
			cycle list={%
				%{smooth,tension=0.5,color=NavyBlue, no markers,line width=0.25pt, densely dotted}, % alpha
				{smooth,tension=0,color=black, mark=halfsquare*,every mark/.append style={rotate=270},mark size=1.5pt,line width=0.5pt},% LSW-Op
				{smooth,tension=0,color=blue, mark=halfsquare*,every mark/.append style={rotate=90},mark size=1.5pt,line width=0.5pt}, % LSW-Un
				{smooth,tension=0,color=red, mark=10-pointed star,mark size=1.5pt,line width=0.5pt},% Chetverikov
				{smooth,tension=0,color=RoyalBlue1, mark=halfcircle*,every mark/.append style={rotate=90},mark size=1.5pt,line width=0.5pt}, % Quad,Knots 0/// Cub,Knots 3
				{smooth,tension=0,color=RoyalBlue2, mark=halfcircle*,every mark/.append style={rotate=180},mark size=1.5pt,line width=0.5pt}, % Quad,Knots/// Cub,Knots 5
				{smooth,tension=0,color=RoyalBlue3, mark=halfcircle*,every mark/.append style={rotate=270},mark size=1.5pt,line width=0.5pt},% Cub,Knots 0/// Cub,Knots 7
				{smooth,tension=0,color=RoyalBlue4, mark=halfcircle*,every mark/.append style={rotate=360},mark size=1.5pt,line width=0.5pt},% Cub,Knots 1
			}
			]
			% Monotonicity for the univariate case
			\nextgroupplot
			\node[anchor=north] at (axis description cs: 0.25,  0.95) {\fontsize{5}{4}\selectfont \shortstack{
					\\  \\
					Homo-Normal\\
					$p=50$
			}};
			\addplot table[x = beta,y=Bayes] from \RMSEpfnohomonSigmaone;
			\addplot table[x = beta,y=DebiasBayes] from \RMSEpfnohomonSigmaone;
			\addplot table[x = beta,y=DebiasLasso] from \RMSEpfnohomonSigmaone;
						
			\nextgroupplot

			\node[anchor=north] at (axis description cs: 0.25,  0.95) {\fontsize{5}{4}\selectfont \shortstack{
					\\  \\
					Homo-Chi\\
					$p=50$
			}};
			\addplot table[x = beta,y=Bayes] from \RMSEpfnohomochiSigmaone;
			\addplot table[x = beta,y=DebiasBayes] from \RMSEpfnohomochiSigmaone;
			\addplot table[x = beta,y=DebiasLasso] from \RMSEpfnohomochiSigmaone;
			
			\nextgroupplot[legend style = {column sep = 7pt, legend to name = LegendMon15}]
			\addplot table[x = beta,y=Bayes] from \RMSEpfnoheterSigmaone;
			\addplot table[x = beta,y=DebiasBayes] from \RMSEpfnoheterSigmaone;
			\addplot table[x = beta,y=DebiasLasso] from \RMSEpfnoheterSigmaone;
			\node[anchor=north] at (axis description cs: 0.25,  0.95) {\fontsize{5}{4}\selectfont \shortstack{
					\\  \\
					Hetero\\
					$p=50$
			}};

			% Monotonicity for the bivariate case
			\nextgroupplot
			\node[anchor=north] at (axis description cs: 0.25,  0.95) {\fontsize{5}{4}\selectfont \shortstack{
					\\  \\
					Homo-Normal\\
					$p=100$
			}};
			\addplot table[x = beta,y=Bayes] from \RMSEponohomonSigmaone;
			\addplot table[x = beta,y=DebiasBayes] from \RMSEponohomonSigmaone;
			\addplot table[x = beta,y=DebiasLasso] from \RMSEponohomonSigmaone;

			\nextgroupplot
			\node[anchor=north] at (axis description cs: 0.25,  0.95) {\fontsize{5}{4}\selectfont \shortstack{
					\\  \\
					Homo-Chi\\
					$p=100$
			}};
			\addplot table[x = beta,y=Bayes] from \RMSEponohomochiSigmaone;
			\addplot table[x = beta,y=DebiasBayes] from \RMSEponohomochiSigmaone;
			\addplot table[x = beta,y=DebiasLasso] from \RMSEponohomochiSigmaone;
			
			\nextgroupplot[legend style = {column sep = 3.5pt, legend to name = LegendMon25}]
			
			\node[anchor=north] at (axis description cs: 0.25,  0.95) {\fontsize{5}{4}\selectfont \shortstack{
					\\  \\
					Hetero\\
					$p=100$
			}};
			\addplot table[x = beta,y=Bayes] from \RMSEponoheterSigmaone;
			\addplot table[x = beta,y=DebiasBayes] from \RMSEponoheterSigmaone;
			\addplot table[x = beta,y=DebiasLasso] from \RMSEponoheterSigmaone;
			
			\nextgroupplot
			\node[anchor=north] at (axis description cs: 0.25,  0.95) {\fontsize{5}{4}\selectfont \shortstack{
					\\  \\
					Homo-Normal\\
					$p=200$
			}};
			
			\addplot table[x = beta,y=Bayes] from \RMSEptnohomonSigmaone;
			\addplot table[x = beta,y=DebiasBayes] from \RMSEptnohomonSigmaone;
			\addplot table[x = beta,y=DebiasLasso] from \RMSEptnohomonSigmaone;

			\nextgroupplot
			\node[anchor=north] at (axis description cs: 0.25,  0.95) {\fontsize{5}{4}\selectfont \shortstack{
					\\  \\
					Homo-Chi\\
					$p=200$
			}};
			\addplot table[x = beta,y=Bayes] from \RMSEptnohomochiSigmaone;
			\addplot table[x = beta,y=DebiasBayes] from \RMSEptnohomochiSigmaone;
			\addplot table[x = beta,y=DebiasLasso] from \RMSEptnohomochiSigmaone;
			
				\nextgroupplot[legend style = {column sep = 3.5pt, legend to name = LegendMon35}]
			
			\node[anchor=north] at (axis description cs: 0.25,  0.95) {\fontsize{5}{4}\selectfont \shortstack{
					\\  \\
					Hetero\\
					$p=200$
			}};
			
			\addplot table[x = beta,y=Bayes] from \RMSEptnoheterSigmaone;
			\addplot table[x = beta,y=DebiasBayes] from \RMSEptnoheterSigmaone;
			\addplot table[x = beta,y=DebiasLasso] from \RMSEptnoheterSigmaone;
			\addlegendentry{Bayes}
			\addlegendentry{Posterior-Corrected}
			\addlegendentry{Debiased-LASSO}

		\end{groupplot}
		\node at ($(myplots c2r1) + (0,-2.25cm)$) {\ref{LegendMon15}};
		\node at ($(myplots c2r2) + (0,-2.25cm)$) {\ref{LegendMon25}};
		\node at ($(myplots c2r3) + (0,-2.25cm)$) {\ref{LegendMon35}};
	\end{tikzpicture}
	\caption{RMSE corresponding to different values of $\beta_0$ under settings S4 (first column), S5 (second column), and S6 (third column).} \label{Fig: five}
\end{figure}

\begin{figure}[!ht]
	\centering\scriptsize
	\begin{tikzpicture} % columns*rows
		\begin{groupplot}[group style={group name=myplots,group size=3 by 3,horizontal sep= 0.8cm,vertical sep=1.1cm},
			grid = minor,
			width = 0.375\textwidth,
			xmax=5,xmin=0,
			ymax=1,ymin=0,
			every axis title/.style={below,at={(0.2,0.8)}},
			xlabel=$\beta_0$,
			x label style={at={(axis description cs:0.95,0.04)},anchor=south},
			xtick={0,1,2,3,4,5},
			ytick={0,0.5,1},
			tick label style={/pgf/number format/fixed},
			legend style={text=black,cells={align=center},row sep = 3pt,legend columns = -1, draw=none,fill=none},
			cycle list={%
				%{smooth,tension=0.5,color=NavyBlue, no markers,line width=0.25pt, densely dotted}, % alpha
				{smooth,tension=0,color=black, mark=halfsquare*,every mark/.append style={rotate=270},mark size=1.5pt,line width=0.5pt},% LSW-Op
				{smooth,tension=0,color=blue, mark=halfsquare*,every mark/.append style={rotate=90},mark size=1.5pt,line width=0.5pt}, % LSW-Un
				{smooth,tension=0,color=red, mark=10-pointed star,mark size=1.5pt,line width=0.5pt},% Chetverikov
				{smooth,tension=0,color=RoyalBlue1, mark=halfcircle*,every mark/.append style={rotate=90},mark size=1.5pt,line width=0.5pt}, % Quad,Knots 0/// Cub,Knots 3
				{smooth,tension=0,color=RoyalBlue2, mark=halfcircle*,every mark/.append style={rotate=180},mark size=1.5pt,line width=0.5pt}, % Quad,Knots/// Cub,Knots 5
				{smooth,tension=0,color=RoyalBlue3, mark=halfcircle*,every mark/.append style={rotate=270},mark size=1.5pt,line width=0.5pt},% Cub,Knots 0/// Cub,Knots 7
				{smooth,tension=0,color=RoyalBlue4, mark=halfcircle*,every mark/.append style={rotate=360},mark size=1.5pt,line width=0.5pt},% Cub,Knots 1
			}
			]
			% Monotonicity for the univariate case
			\nextgroupplot
			\node[anchor=north] at (axis description cs: 0.25,  0.95) {\fontsize{5}{4}\selectfont \shortstack{
					\\  \\
					Homo-Normal\\
					$p=50$
			}};
			\addplot table[x = beta,y=Bayes] from \RMSEpfnohomonSigmatwo;
			\addplot table[x = beta,y=DebiasBayes] from \RMSEpfnohomonSigmatwo;
			\addplot table[x = beta,y=DebiasLasso] from \RMSEpfnohomonSigmatwo;

			\nextgroupplot

			\node[anchor=north] at (axis description cs: 0.25,  0.95) {\fontsize{5}{4}\selectfont \shortstack{
					\\  \\
					Homo-Chi\\
					$p=50$
			}};
			\addplot table[x = beta,y=Bayes] from \RMSEpfnohomochiSigmatwo;
			\addplot table[x = beta,y=DebiasBayes] from \RMSEpfnohomochiSigmatwo;
			\addplot table[x = beta,y=DebiasLasso] from \RMSEpfnohomochiSigmatwo;
			
			\nextgroupplot[legend style = {column sep = 7pt, legend to name = LegendMon16}]
			\addplot table[x = beta,y=Bayes] from \RMSEpfnoheterSigmatwo;
			\addplot table[x = beta,y=DebiasBayes] from \RMSEpfnoheterSigmatwo;
			\addplot table[x = beta,y=DebiasLasso] from \RMSEpfnoheterSigmatwo;
			\node[anchor=north] at (axis description cs: 0.25,  0.95) {\fontsize{5}{4}\selectfont \shortstack{
					\\  \\
					Hetero\\
					$p=50$
			}};
			
			% Monotonicity for the bivariate case
			\nextgroupplot
			\node[anchor=north] at (axis description cs: 0.25,  0.95) {\fontsize{5}{4}\selectfont \shortstack{
					\\  \\
					Homo-Normal\\
					$p=100$
			}};
			\addplot table[x = beta,y=Bayes] from \RMSEponohomonSigmatwo ;
			\addplot table[x = beta,y=DebiasBayes] from \RMSEponohomonSigmatwo ;
			\addplot table[x = beta,y=DebiasLasso] from \RMSEponohomonSigmatwo ;

			\nextgroupplot
			\node[anchor=north] at (axis description cs: 0.25,  0.95) {\fontsize{5}{4}\selectfont \shortstack{
					\\  \\
					Homo-Chi\\
					$p=100$
			}};
			\addplot table[x = beta,y=Bayes] from \RMSEponohomochiSigmatwo;
			\addplot table[x = beta,y=DebiasBayes] from \RMSEponohomochiSigmatwo;
			\addplot table[x = beta,y=DebiasLasso] from \RMSEponohomochiSigmatwo;
			
			\nextgroupplot[legend style = {column sep = 3.5pt, legend to name = LegendMon26}]
			
			\node[anchor=north] at (axis description cs: 0.25,  0.95) {\fontsize{5}{4}\selectfont \shortstack{
					\\  \\
					Heter\\
					$p=100$
			}};
			\addplot table[x = beta,y=Bayes] from \RMSEponoheterSigmatwo ;
			\addplot table[x = beta,y=DebiasBayes] from \RMSEponoheterSigmatwo ;
			\addplot table[x = beta,y=DebiasLasso] from \RMSEponoheterSigmatwo ;

			\nextgroupplot
			\node[anchor=north] at (axis description cs: 0.25,  0.95) {\fontsize{5}{4}\selectfont \shortstack{
					\\  \\
					Homo-Normal\\
					$p=200$
			}};
			
			\addplot table[x = beta,y=Bayes] from \RMSEptnohomonSigmatwo;
			\addplot table[x = beta,y=DebiasBayes] from \RMSEptnohomonSigmatwo;
			\addplot table[x = beta,y=DebiasLasso] from \RMSEptnohomonSigmatwo;

			\nextgroupplot
			\node[anchor=north] at (axis description cs: 0.25,  0.95) {\fontsize{5}{4}\selectfont \shortstack{
					\\  \\
					Homo-Chi\\
					$p=200$
			}};
			\addplot table[x = beta,y=Bayes] from \RMSEptnohomochiSigmatwo;
			\addplot table[x = beta,y=DebiasBayes] from \RMSEptnohomochiSigmatwo;
			\addplot table[x = beta,y=DebiasLasso] from \RMSEptnohomochiSigmatwo;
			
			\nextgroupplot[legend style = {column sep = 3.5pt, legend to name = LegendMon36}]
			
			\node[anchor=north] at (axis description cs: 0.25,  0.95) {\fontsize{5}{4}\selectfont \shortstack{
					\\  \\
					Hetero\\
					$p=200$
			}};
			
			\addplot table[x = beta,y=Bayes] from \RMSEptnoheterSigmatwo;
			\addplot table[x = beta,y=DebiasBayes] from \RMSEptnoheterSigmatwo;
			\addplot table[x = beta,y=DebiasLasso] from \RMSEptnoheterSigmatwo;
			\addlegendentry{Bayes}
			\addlegendentry{Posterior-Corrected}
			\addlegendentry{Debiased-LASSO}

		\end{groupplot}
		\node at ($(myplots c2r1) + (0,-2.25cm)$) {\ref{LegendMon16}};
		\node at ($(myplots c2r2) + (0,-2.25cm)$) {\ref{LegendMon26}};
		\node at ($(myplots c2r3) + (0,-2.25cm)$) {\ref{LegendMon36}};
	\end{tikzpicture}
	\caption{RMSE corresponding to different values of $\beta_0$ under settings S4 (first column), S5 (second column), and S6 (third column).} \label{Fig: six}
\end{figure}
\newpage

\subsection{Calibrated Simulation Results}\label{Sec: calibrated simulations}

We complement the benchmark simulation with additional experiments calibrated to both empirical applications in Section~3\footnote{We would like to thank an anonymous referee for making the suggestion.}. These designs retain the observed covariates and preserve residual heterogeneity, allowing empirically motivated heteroskedasticity. After standardizing the variables, we fit a sparse model by OLS, where the relevant regressors are based on the post-selection inference and also motivated from the existing literature. Let the calibration residuals be $\hat\varepsilon_i$. Then, we embed its coefficients in the full predictor vector $\hat\beta^{\mathrm{cal}}$ by setting the remaining entries to zero with all regressors included. These sparse specifications define the target in our simulation. They do not imply sparsity in the population for either empirical application. Finally, we continue to generate
\begin{equation*}
Y_i^*=X_i^\intercal\hat\beta^{\mathrm{cal}}+\hat\varepsilon_i R_i,
\qquad \Pr(R_i=1)=\Pr(R_i=-1)=\tfrac12,
\end{equation*}
by repeatedly drawing independent Rademacher weights. Conditional on the observed data, the simulated errors have mean zero and variance $\hat\varepsilon_i^2$. The wild bootstrap retains the residual variance pattern without imposing homoskedasticity or Gaussian errors. The covariates and calibrated model specifications remain fixed across 1,000 replications. In each replication, all candidate predictors enter the estimation problem. Then we compare the uncorrected Bayesian and our posterior-corrected (debiased) spike-and-slab and horseshoe procedures with the debiased LASSO. We use 8,000 retained draws for each Bayesian procedure. The baseline Laplace slab rate is $\lambda=1$ in the spike-and-slab prior.

For the first example on the economic growth (Growth in Figure \ref{Fig: calibrated coverage}), the fixed design has $n=90$ and $p=61$ predictors. The calibration regression includes initial GDP and nine controls: the seven controls selected across the procedures in Table~6 of \citet{BelloniChernozhukovHansen2013InferHigh}, supplemented by government consumption and political instability. These controls capture human capital, health, trade, external price shocks, government size, and political stability. The other 51 coefficients are zero. OLS without an intercept on the standardized variables gives the pseudo-true initial-GDP coefficient $-0.4628$.

For the second example about the crime rate (Crime in Figure \ref{Fig: calibrated coverage}), the design has $n=576$ and $p=297$ predictors. The calibration includes the effective abortion rate, year effects, and nine controls motivated by the murder specification in \citet{BelloniChernozhukovHansen2014HighATE}: the eight controls in their footnote~20 plus the initial first difference of the murder-related effective abortion rate. The controls cover welfare generosity, alcohol consumption, income, and initial abortion conditions. The other 275 control coefficients are zero.  The pseudo-true abortion coefficient is $-0.0334$.

Figure~\ref{Fig: calibrated coverage} reports coverage at both nominal levels. In the first example, the uncorrected spike-and-slab and horseshoe 95\% intervals cover only 2.2\% and 52.0\% of the time, respectively. Posterior correction raises coverage to 91.4\% and 90.4\%, compared with 90.0\% for the debiased LASSO. The corresponding average coefficient estimates move from $-0.054$ and $-0.150$ to about $-0.41$, substantially reducing shrinkage bias. All coverage remains below nominal possibly due to the small sample size here. In the second example, the spike-and-slab approximation collapses to zero and never covers the small nonzero truth; horseshoe coverage is 27.7\% with the nominal 95\% coverage. After correction, this jumps to 94.8\% and 94.3\%, compared with 95.5\% for the debiased LASSO. The exercises thus show substantial improvements over the uncorrected procedures and performance comparable to the debiased LASSO under the calibrated variance patterns. Appendix~\ref{Sec: lambda sensitivity} examines sensitivity to the Laplace slab rate $\lambda$, which confirms the same qualitative patterns reported here.

\begin{figure}[!htbp]
\centering
\includegraphics[width=\linewidth]{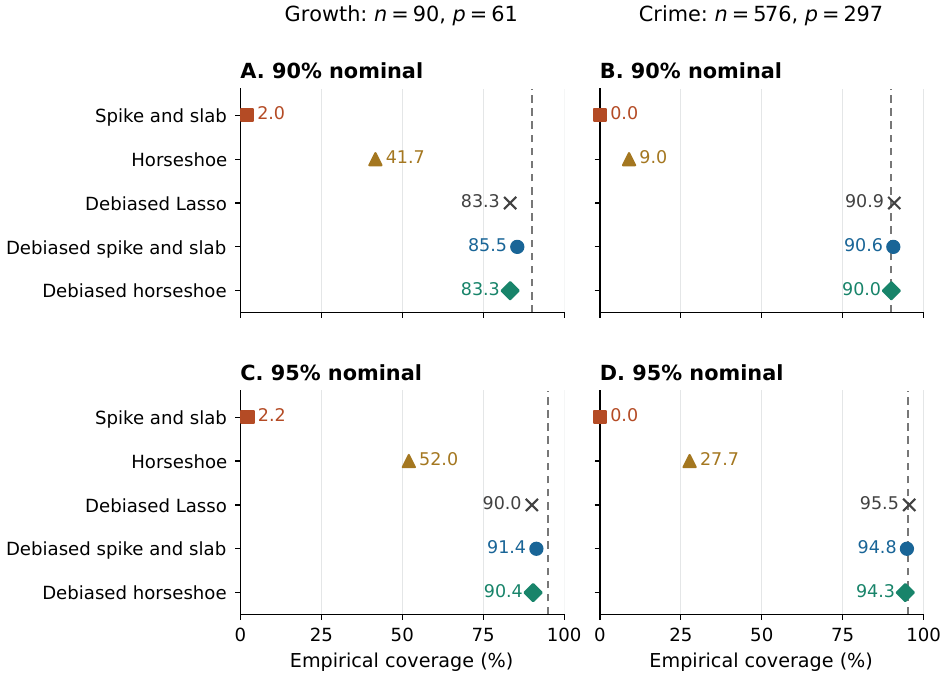}
\caption{Coverage in the calibrated simulations. Each panel reports empirical coverage over 1,000 replications for the indicated nominal level; dashed lines mark that level. Growth uses $(n,p)=(90,61)$ and pseudo-true coefficient $-0.4628$. Crime uses $(576,297)$ and $-0.0334$. Spike-and-slab results use $\lambda=1$. }\label{Fig: calibrated coverage}
\end{figure}

\section{Conclusion \& Extensions}\label{Sec: Conclusion}
In this paper, we develop a posterior-correction inferential method for high-\allowbreak dimensional  regression. The construction resembles the frequentist debiasing step, whereas the key difference is that we correct the entire posterior distribution rather than the point estimator. Our approach is tailored to building confidence intervals based on sparsity-\allowbreak inducing priors such as spike-and-slab or horse\-shoe type priors of the regression coefficients. We establish the frequentist validity of our proposal in the general setup and also provide low-level conditions. It is straightforward to observe that our debiasing step easily extends to other parametric or semiparametric models. To demonstrate the versatility of our general methodology, we mention two possible extensions.

One may consider the generalized linear model for the conditional density function of the dependent variable given covariates as $p(y|x)\propto\exp\left(y\cdot x^\intercal\beta_0-d(x^\intercal\beta_0)\right)$, where $d$ is a convex link function. In this case, our posterior-correction method builds on
\begin{align*}	\tilde{\beta}=\beta+\hat{\Theta}_n\left[\sum_{i=1}^{n}W_{ni}X_i(Y_i-d^\prime(X_i^\intercal\beta))\right],
\end{align*}
where $\hat{\Theta}_n$ stands for a regularized inverse of the Hessian matrix $n^{-1}\sum_{i=1}^nd^{\prime\prime} (X_{i}^{\intercal}\hat\beta_n)X_iX_i^\intercal$, given some pilot estimator $\hat{\beta}_n$. See Section 3 of \cite{VandeGeer_OnAsymptotically_2014}. %When it comes to the high-level posterior contration result in Assumption \ref{Assump:Beta}, one can resort to Theorem 3 in \cite{JeongGhosal_Posterior_2021}.

One may also incorporate group structure which often occurs in additive models \citep{Baietal2022Group}, multivariate outcome variables regression, and regressors collected over mixed frequencies \citep{MoglianiSimoni2021BMidas}. For example, consider the linear regression specified by the following form $Y=\sum_{g=1}^G X_g^{\intercal}\beta_{g,0}+\varepsilon$, where $\beta_{g,0}$ is an $m_g\times 1$ vector of coefficients, and $X_g$ is an $m_g\times 1$ vector of covariates corresponding to group $g = 1, \cdots,G$. The number of groups $G$ is potentially larger than the sample size $n$. In this case, our debiasing procedure is based on
	\begin{align*}
	\tilde{\beta}_{g}=\beta_{g}+\mathcal{E}_g^\intercal \hat{\Theta}_n\left[\frac{1}{n}\sum_{i=1}^{n}X_i(Y_i-X_i^\intercal\beta)\right],
\end{align*}
where $\mathcal{E}_g$ here is the matrix formed by columns of $I_p$ such that $\mathcal{E}_g^\intercal\beta=\beta_g$ with $\beta=(\beta_1^\intercal,\ldots,\beta_G^\intercal)^\intercal$ and $p=m_1+\cdots+m_G$, $X_i=(X_{1,i}^\intercal,\ldots,X_{G,i}^\intercal)^\intercal$, and the initial posterior for $\beta$ can be obtained by the group spike-and-slab prior from \cite{Baietal2022Group}. 

For both extensions, we believe our posterior-correction framework offers interesting new insights. We will defer the detailed theoretical development, as well as practical performance to some future work.

\begin{appendices}
\titleformat{\section}{\Large\center}{{\sc Appendix} \thesection}{1em}{}

\section{Proofs of Main Results}\label{Sec: Main proofs}
\renewcommand{\theequation}{A.\arabic{equation}}
\setcounter{equation}{0}
We provide some heuristic discussions about proofs of our theoretical results. 
Starting from the definition of $\tilde{\beta}$ in \eqref{Eqn:DebiasBayes} and the expansion of $\hat{\beta}_{n}$ in Assumption \ref{Assump:Frequentist}, 
we obtain from the model \eqref{Eqn:Model} the following decomposition:
\begin{align}\label{Eqn: Decomposition}
\tilde{\beta}-\hat{\beta}_{n} 
&= 
\underbrace{\Theta_0 \left[\sum_{i=1}^{n}\left(W_{ni}-\frac{1}{n}\right)X_i\varepsilon_i\right]}_{S_n}
+\underbrace{(\hat{\Theta}_n -\Theta_0)\left[\sum_{i=1}^{n}W_{ni}X_i\varepsilon_i\right]}_{R_{1n}} \notag\\
&\hspace{0.5cm} +
\underbrace{\left\{\left[I_p - \hat{\Theta}_n \hat{\Omega}_n\right]
 -\hat{\Theta}_n \left[\sum_{i=1}^{n}\left(W_{ni}-\frac{1}{n}\right)X_iX_i^{\intercal}\right]\right\}}_{R_{2n}}(\beta - \beta_0) 
- {\Delta}_n \notag\\
&= S_n + R_{1n} + R_{2n}(\beta - \beta_0) - {\Delta}_n,
\end{align}
where $\Delta_n$ is defined in Assumption \ref{Assump:Frequentist}. The decomposition in \eqref{Eqn: Decomposition} forms the foundation of our analysis. 
The leading term $S_n$ captures the first-order stochastic component of the debiased law and drives its asymptotic Gaussianity. Considering its close connection to the bootstrap central limit theorem, we rely on the recent strong approximation result for the exchangeable bootstrap in the high dimensional setup in \cite{FangSantosShaikhTorgovitsky2023LP}; see also Theorem \ref{Thm:strong approx} for a version of the result. 

When it comes to remainder terms, we show that for every $\epsilon>0$,
\begin{align}
\Pi\big(\sqrt{n}\|R_{1n}\|_{\infty}>\epsilon \mid Z^{(n)}\big)=o_{P_0}(1),~ 
\Pi\big(\sqrt{n}\|R_{2n}(\beta-\beta_0)\|_{\infty}>\epsilon \mid Z^{(n)}\big)=o_{P_0}(1),
\end{align}
and that $\sqrt{n}\|\Delta_n\|_{\infty}=o_{P_0}(1)$, which follows from Assumption \ref{Assump:Frequentist}. 
Hence, the remainder terms $R_{1n}$, $R_{2n}(\beta-\beta_0)$, and ${\Delta}_n$ are uniformly negligible across coordinates. The subtle difference about $R_{1n}$ and $R_{2n}(\cdot)$ is whether it depends on $\beta-\beta_0$ or not. As a consequence, one needs to show the convergence of certain random matrices in the operator norm or in the elementwise sup-norm.

Throughout, we denote by $\mathbb{E}_{Z,W}[\cdot]$ the expectation evaluated under the joint distribution of $Z^{(n)}$ and $W^{(n)}$. For a sequence $r_n>0$, the symbols $o_{P_{Z,W}}(r_n)$ and $O_{P_{Z,W}}(r_n)$ are defined with respect to the same underlying probability measure $P_{Z,W}$.

\begin{proof}[Proof of Theorem \ref{Thm:BvM1}]
As shown in the sequel, it suffices to establish the conditional weak convergence of $e_{j}^{\intercal}\sqrt{n}(\tilde{\beta}-\hat{\beta}_{n})$ restricted to some set $\Gamma_n$ such that $\Pi(\Gamma_n^c\mid Z^{(n)})=o_{P_0}(1)$. To construct this set $\Gamma_n$, let $\eta_n:=\sqrt{\log p/n} + \log^2 n\log^2 p/n$. Since $\sqrt{n}(\epsilon_n \|\Theta_0\|_{\infty}+\delta_n)\eta_n\to0$, there exists a positive sequence $\xi_n$ such that $\xi_n\to\infty$ and $\xi_n\sqrt{n}(\epsilon_n\|\Theta_0\|_{\infty}+\delta_n)\eta_n \to0$. Given the positive constant $C$ in Assumption \ref{Assump:Beta}, define the following sets: 
\begin{align}\label{Thm:BvM1:Eqn:1}
\mathcal{A}_{n}&: = \left\{\left\|\sum_{i=1}^{n}W_{ni}X_i\varepsilon_i\right\|_{\infty}\leq  C \xi_{n} \eta_n\right\},\notag\\
\mathcal{B}_{n}&: = \left\{\left\|\sum_{i=1}^{n}(W_{ni}-\frac{1}{n})X_iX_i^{\intercal}\right\|_{\max}\leq C \xi_{n} \eta_n\right\},\\
\mathcal{C}_{n}&: = \{\|\beta - \beta_0\|_{1}\leq C\epsilon_n\}.\notag
\end{align}
Then, we take $\Gamma_n:=\mathcal{A}_n\bigcap \mathcal{B}_n\bigcap\mathcal{C}_n$. Since $\xi_n\to\infty$, $\mathbb{E}_0[\Pi_{W}(\mathcal{A}^{c}_{n}\mid Z^{(n)})] = P_{Z,W} (\mathcal{A}^{c}_{n}) \to 0$ and $\mathbb{E}_0[\Pi_{W}(\mathcal{B}^{c}_{n}\mid Z^{(n)})] = P_{Z,W}(\mathcal{B}^{c}_{n}) \to 0$ by Lemma \ref{Lem:BootstrapRate}. And by Assumption \ref{Assump:Beta}, $\mathbb{E}_0[\Pi_{\beta}(\mathcal{C}_n^{c})\mid Z^{(n)}]\to 0$. 
Using both results, Markov inequality yields 
\begin{align}\label{Thm:BvM1:Eqn:4}
\Pi(\Gamma_n^c\mid Z^{(n)})\leq \Pi_{W}(\mathcal A_n^c\mid Z^{(n)}) +\Pi_{W}(\mathcal B_n^c\mid Z^{(n)})+\Pi_{\beta}(\mathcal C_n^c\mid Z^{(n)})\overset{P_0}{\to} 0,
\end{align}
where the inequality holds since $\Pi(\,\cdot\,|Z^{(n)}) = \Pi_{\beta}(\,\cdot\,|Z^{(n)})\times \Pi_W(\,\cdot\,|Z^{(n)})$, and both $\mathcal A_n^c$ and $\mathcal B_n^c$ do not depend on $\beta$ while $\mathcal C_n^c$ does not depend on $W^{(n)}$. 

By the triangle inequality, for each $j=1,\ldots, p$,
\begin{align}\label{Thm:BvM1:Eqn:2}
 &d_{BL}\left(\mathcal{L}_{\Pi}(e_{j}^{\intercal}\sqrt{n}(\tilde{\beta}-\hat{\beta}_{n})\mid Z^{(n)}),N(0,\sigma^2_{0,j})\right)\notag\\
 &\hspace{-0.5cm}\leq d_{BL}\left(\mathcal{L}_{\Pi}(e_{j}^{\intercal}\sqrt{n}S_{n}\mid Z^{(n)}),N(0,\sigma^2_{0,j}) \right) \notag\\
		+&d_{BL}\left(\mathcal{L}_{\Pi}(e_{j}^{\intercal}\sqrt{n}(\tilde{\beta}-\hat{\beta}_{n})\mid Z^{(n)}),\mathcal{L}_{\Pi}(e_{j}^{\intercal}\sqrt{n}S_{n}\mid Z^{(n)}) \right).
\end{align}
By Lemma \ref{Lem:BootApprox}, the first term on the right hand side of the inequality in \eqref{Thm:BvM1:Eqn:2} is convergent to zero in probability since $\mathcal{L}_{\Pi}(e_{j}^{\intercal}\sqrt{n}S_{n}\mid Z^{(n)})$ is identical to $\mathcal{L}_{\Pi_{W}}(e_{j}^{\intercal}\sqrt{n}S_{n}\mid Z^{(n)})$. For the second term, it follows from \eqref{Eqn: Decomposition} that
\begin{align}\label{Thm:BvM1:Eqn:3}
&d_{BL}\left(\mathcal{L}_{\Pi}(e_{j}^{\intercal}\sqrt{n}(\tilde{\beta}-\hat{\beta}_{n})\mid Z^{(n)}),\mathcal{L}_{\Pi}(e_{j}^{\intercal}\sqrt{n}S_{n}\mid Z^{(n)}) \right)\notag\\
&\hspace{-0.5cm}= \sup_{\|f\|_{BL}\le 1}\left|\int [f((e_{j}^{\intercal}\sqrt{n}(\tilde{\beta}-\hat{\beta}_{n}))-f(e_{j}^{\intercal}\sqrt{n}S_{n})]d\Pi(\beta,W^{(n)}\mid Z^{(n)})\right|\notag\\
&\hspace{-0.5cm}\leq \sup_{\|f\|_{BL}\le 1}\left|\int_{\Gamma_n} [f((e_{j}^{\intercal}\sqrt{n}(\tilde{\beta}-\hat{\beta}_{n}))-f(e_{j}^{\intercal}\sqrt{n}S_{n})]d\Pi(\beta,W^{(n)}\mid Z^{(n)})\right|\notag\\
&+ \sup_{\|f\|_{BL}\le 1}\int_{\Gamma_n^c} \left|[f((e_{j}^{\intercal}\sqrt{n}(\tilde{\beta}-\hat{\beta}_{n}))\right|d\Pi(\beta,W^{(n)}\mid Z^{(n)})\notag\\
&+\sup_{\|f\|_{BL}\le 1}\int_{\Gamma_n^c} \left|f(e_{j}^{\intercal}\sqrt{n}S_{n})]\right|d\Pi(\beta,W^{(n)}\mid Z^{(n)})\notag\\
        &\hspace{-0.5cm}\le \int_{\Gamma_n}|e_{j}^{\intercal}[\sqrt{n}R_{1n}+\sqrt{n}R_{2n}(\beta-\beta_0) - \sqrt{n}{\Delta}_{n}]|d\Pi(\beta,W^{(n)}\mid Z^{(n)})\notag\\
&+2\Pi(\Gamma_n^{c} \mid Z^{(n)}),
\end{align}
where the first inequality holds by the triangle inequality, and the second inequality holds since $\sup_{x}|f(x)|\leq 1$ and $|f(x) - f(y)|\leq |x-y|$ for any $x, y$ whenever $\|f\|_{BL}\leq 1$. 

By \eqref{Thm:BvM1:Eqn:4}, it remains to show that the first term on the right hand side of the inequality in \eqref{Thm:BvM1:Eqn:3} is convergent to zero in probability. By simple algebra, we have 
\begin{align}\label{Thm:BvM1:Eqn:5}
&\int_{\Gamma_n} |e_{j}^{\intercal}[\sqrt{n}R_{1n}+\sqrt{n}R_{2n}(\beta-\beta_0) - \sqrt{n}{\Delta}_{n}]|d\Pi(\beta,W^{(n)}\mid Z^{(n)})\notag\\
&\hspace{-0.5cm} \leq \int_{\mathcal A_n \cap \mathcal B_n \cap \mathcal C_n} (\|\sqrt{n}R_{1n}\|_{\infty}+\|\sqrt{n}R_{2n}\|_{\max}\|\beta-\beta_0\|_{1} + \|\sqrt{n}{\Delta}_{n}\|_{\infty})d\Pi(\beta,W^{(n)}\mid Z^{(n)}) \notag\\
&\hspace{-0.5cm} \leq \int_{\mathcal B_n}\|\sqrt{n} R_{2n}\|_{\max} d\Pi_{W}(W^{(n)}\mid Z^{(n)}) \int_{\mathcal C_n} \|\beta-\beta_0\|_1 d\Pi_{\beta}(\beta \mid Z^{(n)})\notag\\
 +& \int_{\mathcal A_n}\|\sqrt{n}R_{1n}\|_{\infty}d\Pi_{W}(W^{(n)}\mid Z^{(n)}) + \|\sqrt{n}{\Delta}_{n}\|_{\infty},
\end{align}
where the first inequality holds since $\|Ax\|_{\infty}\leq \|A\|_{\max}\|x\|_{1}$ for matrix $A$ and vector $x$, and the second inequality holds since $\Pi(\,\cdot\,|Z^{(n)}) = \Pi_{\beta}(\,\cdot\,|Z^{(n)})\times \Pi_W(\,\cdot\,|Z^{(n)})$. We next show that each term on the right hand side of the second inequality in \eqref{Thm:BvM1:Eqn:5} is convergent to zero in probability. First, $\|\sqrt{n}{\Delta}_{n}\|_{\infty}\overset{P_{0}}{\to} 0$ by Assumption \ref{Assump:Frequentist}. Second, 
\begin{align}\label{Thm:BvM1:Eqn:6}
&\int_{\mathcal A_n}\|\sqrt{n} R_{1n}\|_{\infty} d\Pi_{W}(W^{(n)}\mid Z^{(n)})\notag\\
&\hspace{-0.5cm} \leq \sqrt{n} \|\hat{\Theta}_n -\Theta_0)\|_{\infty}\int_{\mathcal A_n}\left\|\sum_{i=1}^{n}W_{ni}X_i\varepsilon_i\right\|_{\infty}d\Pi_{W}(W^{(n)}\mid Z^{(n)})\notag\\
&\hspace{-0.5cm} \leq \sqrt{n} \|\hat{\Theta}_n -\Theta_0)\|_{\infty} C\xi_n \eta_n \overset{P_{0}}{\to} 0,
\end{align}
where the first inequality holds since $\|Ax\|_{\infty}\leq \|A\|_{\infty} \|x\|_{\infty}$, and the convergence holds due to Assumption \ref{Assump:Precision} and $\xi_n\sqrt{n}\delta_n \eta_n \to0$. Third, similarly,
\begin{align}\label{Thm:BvM1:Eqn:7}
&\int_{\mathcal B_n}\|\sqrt{n} R_{2n}\|_{\max} d\Pi_{W}(W^{(n)}\mid Z^{(n)})\notag\\
&\hspace{-0.5cm} \leq \sqrt{n} \|\hat{\Theta}_n\hat{\Omega}_n - I_{p}\|_{\max} + \sqrt{n}\|\hat{\Theta}_n\|_{\infty}\int_{\mathcal B_n}\left\|\sum_{i=1}^{n}(W_{ni}-\frac{1}{n})X_iX_i^{\intercal}\right\|_{\max} d\Pi_{W}(W^{(n)}\mid Z^{(n)}) \notag\\
&\hspace{-0.5cm} \leq \sqrt{n} \|\hat{\Theta}_n\hat{\Omega}_n - I_{p}\|_{\max} + \sqrt{n}(\|\hat{\Theta}_n-{\Theta}_0\|_{\infty} + \|{\Theta}_0\|_{\infty})C \xi_n\eta_n,
\end{align}
where the first inequality holds since $\|A B\|_{\max} \leq \|A\|_{\infty}\|B\|_{\max}$ for matrices $A$ and $B$, and trivially, 
\begin{align}\label{Thm:BvM1:Eqn:8}
\int_{\mathcal C_n} \|\beta-\beta_0\|_1 d\Pi_{\beta}(\beta \mid Z^{(n)}) \leq C \epsilon_n.
\end{align}
Combining \eqref{Thm:BvM1:Eqn:7} and \eqref{Thm:BvM1:Eqn:8} implies the first term on the right hand side of the second inequality in \eqref{Thm:BvM1:Eqn:5} is convergent to zero in probability due to Assumption \ref{Assump:Precision}, $\sqrt{n}\epsilon_n\gamma_n\to 0$, and $\xi_n\sqrt{n}\epsilon_n\eta_n\|\Theta_0\|_{\infty}\to 0$. This completes the proof of the theorem. 
\end{proof}

\begin{proof}[Proof of Corollary \ref{Cor:BvM1}]
The essential idea is to establish convergence of the quantiles of the debiased law of $e_{j}^{\intercal}\sqrt n(\tilde{\beta}-\hat\beta_n)$. We also make use of the symmetry of the limiting Gaussian distribution to show the validity of the debiased confidence interval.

By Lemma \ref{Lem:StrongApprox}, there exists some Gaussian $\mathbb G_{j}\sim N(0,\sigma^2_{0,j})$ such that
\begin{align}\label{Cor:BvM1 aux1}
\limsup_{n\to\infty}P_0(|e_{j}^{\intercal}\sqrt n(\hat\beta_n-\beta_0)-\mathbb G_{j}|>\eta_n) = 0,
\end{align}
for some $\eta_n\downarrow 0$. By Lemmas \ref{Lem:BootApprox}, \eqref{Eqn: Decomposition}, and the inequalities used in the arguments of \eqref{Thm:BvM1:Eqn:4}, \eqref{Thm:BvM1:Eqn:5}, and \eqref{Thm:BvM1:Eqn:6} and by making $\eta_n$ larger if necessary, we have that
\begin{align}\label{Cor:BvM1 aux2}
\limsup_{n\to\infty}\mathbb{E}_0[\Pi(|e_{j}^{\intercal}\sqrt n(\tilde \beta-\hat\beta_{n})-\bar{\mathbb G}_{j}|>\eta_n\mid Z^{(n)})] = 0,    
\end{align}
for a copy $\bar{\mathbb G}_{j}$ of $\mathbb G_{j}$ that is independent of $Z^{(n)}$. By Markov’s inequality, Fubini’s theorem, and \eqref{Cor:BvM1 aux2}, there is some $\varpi_n\downarrow 0$ such that
\begin{align}\label{Cor:BvM1 aux3}
\limsup_{n\to\infty}\Pi(|e_{j}^{\intercal}\sqrt n(\tilde \beta-\hat\beta_{n})-\bar{\mathbb G}_{j}|>\eta_n\mid Z^{(n)}) = o_{P_0}(\varpi_n).    
\end{align}
For $0<\tau<1$, let $c_{j}(\tau)$ be the $\tau$th quantile of $N(0,\sigma^2_{0,j})$ and $\hat c_{n,j}(\tau)$ be the $\tau$th conditional quantile of $e_{j}^{\intercal}\sqrt n(\tilde \beta-\hat\beta_{n})$ given $Z^{(n)}$. Since $\bar{\mathbb G}_{j}$ is independent of $Z^{(n)}$, we note that $c_{j}(\tau)$ is also the $\tau$th conditional quantile of $\bar{\mathbb G}_{j}$ given $Z^{(n)}$. By Lemma 11 in \citet{ChernozhukovLeeRosen2013Intersection}, it follows from \eqref{Cor:BvM1 aux3} that, with probability approaching one, for some $\varsigma_n=o(\varpi_n)$, and for any $\tau\in(0,1-\varsigma_n)$,
\begin{align}\label{Cor:BvM1 aux4}
\hat c_{n,j}(\tau)\le c_{j}(\tau+\varsigma_n)+\eta_n,\quad c_{j}(\tau)\le \hat c_{n,j}(\tau+\varsigma_n)+\eta_n.
\end{align}
By equivariance of quantiles (see, e.g., \citet[p.39]{Koenker2005QR}), we note that
\begin{align}\label{Cor:BvM1 aux5}
P_0(\beta_{0,j}\in\mathcal C_{n,j}(\alpha)) = P_0\left(-\hat c_{n,j}(1-{\alpha}/{2})\leq e_{j}^{\intercal}\sqrt{n}(\hat{\beta}_{n}-\beta_{0})\leq -\hat c_{n,j}({\alpha}/{2})\right).
\end{align}
Let $H_{j}$ be the cdf of $N(0,\sigma^2_{0,j})$. By \eqref{Cor:BvM1 aux1} and \eqref{Cor:BvM1 aux4}, we have:
\begin{align}\label{Cor:BvM1 aux6}
&\hspace{0.5cm}\limsup_{n\to\infty} P_0(e_{j}^{\intercal}\sqrt{n}(\hat{\beta}_{n}-\beta_{0})
\le -\hat c_{n,j}({\alpha}/{2}))\notag\\
&\le \limsup_{n\to\infty} P_0(\mathbb G_{j}-\eta_n \le -c_{j}({\alpha}/{2}-\varsigma_n)+\eta_n)\notag\\
&= \limsup_{n\to\infty}H_{j}(-c_{j}({\alpha}/{2}-\varsigma_n)+2\eta_n).
\end{align}
Since $\sigma^2_{0,j}$ is bounded away from zero by assumption, the density of $H_{j}$ is bounded. Thus,  the fundamental theorem of calculus implies
\begin{align}\label{Cor:BvM1 aux7}
H_{j}(-c_{j}({\alpha}/{2}-\varsigma_n)+2\eta_n)  - H_{j}(-c_{j}({\alpha}/{2}-\varsigma_n)) = O(\eta_n).
\end{align}
Note that $H_{j}(-c_{j}(\alpha/2-\varsigma_n))=1-\alpha/2+\varsigma_n$ by the symmetry of $N(0,\sigma^2_{0,j})$ around zero. Since $\varsigma_n\vee \eta_n=o(1)$, we may then obtain from \eqref{Cor:BvM1 aux6} and \eqref{Cor:BvM1 aux7} that
\begin{align}\label{Cor:BvM1 aux8}
\limsup_{n\to\infty} P_0(e_{j}^{\intercal}\sqrt{n}(\hat{\beta}_{n}-\beta_{0})\le -\hat c_{n,j}({\alpha}{/2}))\le 1-\alpha/2.
\end{align}
For the reverse direction, we have by \eqref{Cor:BvM1 aux1} and \eqref{Cor:BvM1 aux4} that
\begin{align}\label{Cor:BvM1 aux9}
&\hspace{0.5cm}\liminf_{n\to\infty} P_0(e_{j}^{\intercal}\sqrt{n}(\hat{\beta}_{n}-\beta_{0})\le -\hat c_{n,j}({\alpha}/{2}))\notag\\
&\ge \liminf_{n\to\infty} P_0(\mathbb G_{j}+\eta_n \le -c_{j}({\alpha}/{2}+\varsigma_n)-\eta_n)\notag\\
&= \liminf_{n\to\infty}H_{j}(-c_{j}({\alpha}/{2}+\varsigma_n)-2\eta_n)=1-{\alpha}/{2},
\end{align}
where the last step follows by another application of the fundamental theorem of calculus and $\varsigma_n\vee \eta_n=o(1)$. Combining \eqref{Cor:BvM1 aux8} and \eqref{Cor:BvM1 aux9} yields:
\begin{align}\label{Cor:BvM1 aux10}
\lim_{n\to\infty} P_0(e_{j}^{\intercal}\sqrt{n}(\hat{\beta}_{n}-\beta_{0})\le -\hat c_{n,j}({\alpha}/{2}))=    1-\alpha/2.
\end{align}
By similar arguments, we may also obtain that
\begin{align}\label{Cor:BvM1 aux11}
\lim_{n\to\infty} P_0(e_{j}^{\intercal}\sqrt{n}(\hat{\beta}_{n}-\beta_{0})< -\hat c_{n,j}(1-{\alpha}/{2}))=    \alpha/2.
\end{align}
The conclusion of the corollary then follows from \eqref{Cor:BvM1 aux5}, \eqref{Cor:BvM1 aux10}, and \eqref{Cor:BvM1 aux11}.  
\end{proof}

\begin{proof}[Proof of Theorem \ref{Thm:BvM2}] 
	We first establish the strong approixmation with an explicit control on the remainder term. By Lemma \ref{Lem:BootstrapRate}, we have:
	\begin{align}\label{Thm:BvM2,aux1}
	&	\|\sqrt n E_J^\intercal R_{1n}\|_\infty  \le \sqrt n\| R_{1n}\|_\infty 
		\le \sqrt n
	\|\hat\Theta_n-\Theta_0\|_\infty \|\sum_{i=1}^{n}W_{ni}X_i\varepsilon_i\|_\infty \notag\\
	=&O_{P_{Z,W}}\left( \delta_n\log^{1/2} p+\frac{\delta_n\log^{2} n\log^{2} p}{n^{1/2}}\right).
	\end{align}
	Inspecting the proof of Theorem \ref{Thm:BvM1}, we also have:
	\begin{align}\label{Thm:BvM2,aux2}
		&\|\sqrt n  E_J^\intercal R_{2n}(\beta-\beta_0)\|_\infty  \le \sqrt n\|R_{2n}\|_{\max}\|\beta-\beta_0\|_1 \notag\\
		\le&  \sqrt n \left[O_{P_0}(\gamma_n)+\{O_{P_0}(\delta_n)+\|\Theta_0\|_\infty\}O_{P_{Z,W}}\left( \sqrt{\frac{\log p}{n}}+\frac{\log^{2} n\log^{2} p}{n}\right)\right] O_{P_{Z,W}}(\epsilon_n)\notag\\
		 =&  O_{P_{Z,W}}\left(\sqrt n \gamma_n\epsilon_n + \epsilon_n(\delta_n+\|\Theta_0\|_\infty)\left( \log^{1/2} p+\frac{\log^{2} n\log^{2} p}{n^{1/2}}\right)\right),
	\end{align}
	where the second inequality follows by Assumption \ref{Assump:Beta}, Fubini's theorem, and the dominated convergence theorem. Combining results \eqref{Thm:BvM2,aux1} and \eqref{Thm:BvM2,aux2} with Lemma \ref{Lem:BootApproxSimultaneous} leads to the first claim of the theorem. 

	When it comes to the second claim, fix $\epsilon>0$ and define the event:
	\begin{align}
		\mathcal D_n=\{\|E_J^\intercal \sqrt n (\tilde\beta-\hat\beta_n)-\bar{\mathbb G}_J\|_\infty>\epsilon\}.
	\end{align}
	By $\bar{\mathbb G}_J\sim N(0,\Sigma_{0,J})$ being indepedent of $Z^{(n)}$ and Jensen's inequality, we have
	\begin{align}\label{Thm:BvM2,aux3}
		& d_{BL}\! (
\mathcal{L}_{\Pi}\big(E_{J}^{\intercal}\sqrt{n}(\tilde{\beta}-\hat{\beta}_{n})\mid Z^{(n)}\big),  \,
N(0,\Sigma_{0,J})
)\notag\\
 = &\sup_{\|f\|_{BL}\le 1} \left|\mathbb E\left[f\left(E_{J}^{\intercal}\sqrt{n}(\tilde{\beta}-\hat{\beta}_{n})\right)\mid Z^{(n)}\right] - E\left[f(\bar{\mathbb G}_J)\right]\right|\notag\\
\le & \sup_{\|f\|_{BL}\le 1} \mathbb E\left[\left|f\left(E_{J}^{\intercal}\sqrt{n}(\tilde{\beta}-\hat{\beta}_{n})\right)-f(\bar{\mathbb G}_J)\right|\mid Z^{(n)}\right].  
	\end{align}
	By the definition of $\|\cdot\|_{BL}$, we obtain from \eqref{Thm:BvM2,aux3} that
	\begin{align}\label{Thm:BvM2,aux4}
		d_{BL}\! (
\mathcal{L}_{\Pi}\big(E_{J}^{\intercal}\sqrt{n}(\tilde{\beta}-\hat{\beta}_{n})& \mid Z^{(n)}\big),  \,
N(0,\Sigma_{0,J})\le \epsilon + 2 P(\mathcal D_n|Z^{(n)}) .
	\end{align}
	Since \(P(\mathcal D_n|Z^{(n)}) = o_p(1)\), we conclude that the second claim.  
\end{proof}

\begin{proof}[Proof of Corollary \ref{Cor:BvM2}] By our assumption, we may select some $\ell_n=o(1)$ such that
	\begin{align}\label{Cor:BvM2 aux0}
		a_n\vee b_n=o(\ell_n), \quad \ell_n[\frac{\bar\sigma_J\log (1+J)}{\varpi_J^2} +\frac{1}{\varpi_J} ]= o(1).
	\end{align}
	Indeed, one may simply set
	\[
\ell_n=\frac{[a_n\vee b_n + (a_n\vee b_n)\bar\sigma_J\log(1+J)/\varpi_n^2 +(a_n\vee b_n)/\varpi_n+1/n]^{1/2}}{1+\bar\sigma_J\log(1+J)/\varpi_n^2+1/\varpi_n}.
	\]
	By Lemma \ref{Lem:StrongApprox2}, there exists a $\mathbb G_{J}\sim N(0,\Sigma_{0,J})$ such that
\begin{align}\label{Cor:BvM2 aux1}
\limsup_{n\to\infty}P_0(\|E_{J}^{\intercal}\sqrt n(\hat\beta_n-\beta_0)-\mathbb G_{J}\|_\infty>\ell_n) = 0.
\end{align}
	By Theorem \ref{Thm:BvM2}, we may obtain a copy $\bar{\mathbb G}_{J}$ of $\mathbb G_{J}$ independent of $Z^{(n)}$ while satisfying: 
\begin{align}\label{Cor:BvM2 aux3}
\limsup_{n\to\infty}\Pi(\|E_{J}^{\intercal}\sqrt n(\tilde \beta-\hat\beta_{n})-\bar{\mathbb G}_{J}\|_\infty>\ell_n\mid Z^{(n)}) = o_{P_0}(\eta_n),   
\end{align}
for some $\eta_n\downarrow 0$. Note that $c_{J}(\tau)$ equals the $\tau$th conditional quantile of $\bar{\mathbb G}_{J}$ given $Z^{(n)}$. By Lemma 11 in \citet{ChernozhukovLeeRosen2013Intersection}, result \eqref{Cor:BvM2 aux3} implies that with probability approaching one and for some $\varsigma_n=o(\eta_n)$,
\begin{align}\label{Cor:BvM2 aux4}
\hat c_{n,J}(\tau)\le c_{J}(\tau+\varsigma_n)+\ell_n,\quad c_{J}(\tau)\le \hat c_{n,J}(\tau+\varsigma_n)+\ell_n.
\end{align}

Letting $T_{n,J}:=\|E_{J}^{\intercal}\sqrt n(\hat\beta_n-\beta_0)\|_\infty$, we note that
\begin{align}\label{Cor:BvM2 aux5}
P_0\left(\beta_{0,j}\in\mathcal C_{n,j}(\alpha),~~\forall j=1,\ldots,J\right)
=P_0\left(T_{n,J}\leq \hat c_{n,J}(1-\alpha)\right).
\end{align}
By results \eqref{Cor:BvM2 aux1} and \eqref{Cor:BvM2 aux4}, we obtain that, 
\begin{align}\label{Cor:BvM2 aux6}
\limsup_{n\to\infty}P_0(T_{n,J}\leq \hat c_{n,J}(1-\alpha))&\le \limsup_{n\to\infty}P_0(\|\mathbb G_J\|_\infty\leq  c_{J}(1-\alpha+\varsigma_n)+2\ell_n)\notag \\
&= \limsup_{n\to\infty}H_{J}(c_{J}(1-\alpha+\varsigma_n)+2\ell_n).
\end{align}
By Assumption \ref{Ass:anti-concentration}(i)(ii), Lemma D.3 in \citet{FangSeo2019Shape}, and the fundamental theorem of calculus, we know that the cdf $H_J$ of $\|\mathbb G_J\|_\infty$ is absolutely continuous and strictly increasing on $(0,\infty)$ and that, for $m_J=c_J(0.5)$,
\begin{align}\label{Cor:BvM2 aux7}
&|H_{J}(c_{J}(1-\alpha+\varsigma_n)+2\ell_n)-H_{J}(c_{J}(1-\alpha+\varsigma_n))| \\
= &\int_{c_J(1-\alpha+\varsigma_n)}^{c_{J}(1-\alpha+\varsigma_n)+2\ell_n} H_J'(u)\,\mathrm du\le \int_{c_J(1-\alpha+\varsigma_n)}^{c_{J}(1-\alpha+\varsigma_n)+2\ell_n} \frac{2u-m_J}{(u-m_J)^2}\,\mathrm du~.
\end{align}
Since $u\mapsto (2u-m_J)/(u-m_J)^2$ is decreasing on $(m_J,\infty)$, we may combine Assumption \ref{Ass:anti-concentration}(ii), \(H_{J}(c_{J}(1-\alpha+\varsigma_n))=c_{J}(1-\alpha+\varsigma_n)\), and result \eqref{Cor:BvM2 aux7} to obtain that
\begin{align}\label{Cor:BvM2 aux8}
|H_{J}(c_{J}(1-\alpha+\varsigma_n)+2\ell_n)-(1-\alpha+\varsigma_n)|\le 2\ell_n \frac{2(m_J+\varpi_J)-m_J}{\varpi_J^2}~.
\end{align}

Exploiting $\|\mathbb G_J\|_\infty=\max_{j=1}^J\max\{\mathbb G_{J,j},-\mathbb G_{J,j}\}$ with $\mathbb G_{J,j}$ the $j$th entry of $\mathbb G_J$, Assumption \ref{Ass:anti-concentration}(ii), and Borell's inequality \citep[p.~82]{DavydovLifshitsSmorodina1998Local}, we note that
\begin{align}\label{Cor:BvM2 aux9}
\varpi_J\le c_J(1-\alpha)\le m_J + z_{1-\alpha}\bar\sigma_J,
\end{align}
where $z_{1-\alpha}$ is the $1-\alpha$ quantile of the standard normal distribution. By \citet{Kwapien1994Median} and a maximal inequality (see, e.g., Lemma A.8 in \citet{FangSantosShaikhTorgovitsky2023LP}), we have
\begin{align}\label{Cor:BvM2 aux10}
	m_J\le E[\|\mathbb G_J\|_\infty] \le \bar\sigma_J\log^{1/2}(1+J).
\end{align}
Combining results \eqref{Cor:BvM2 aux6}--\eqref{Cor:BvM2 aux10} with \(\varsigma_n=o(1)\) and \eqref{Cor:BvM2 aux0}, we thus obtain that
\begin{align}\label{Cor:BvM2 aux11}
\limsup_{n\to\infty}P_0(T_{n,J}\leq \hat c_{n,J}(1-\alpha))\le 1-\alpha.
\end{align}

By results \eqref{Cor:BvM2 aux1} and \eqref{Cor:BvM2 aux4}, we may also obtain that, 
\begin{align}\label{Cor:BvM2 aux12}
\liminf_{n\to\infty}P_0(T_{n,J}\leq \hat c_{n,J}(1-\alpha))&\ge \liminf_{n\to\infty}P_0(\|\mathbb G_J\|_\infty\leq  c_{J}(1-\alpha-\varsigma_n)-2\ell_n)\notag \\
&= \liminf_{n\to\infty}H_{J}(c_{J}(1-\alpha-\varsigma_n)-2\ell_n).
\end{align}
Since $\ell_n=o(\varpi_J)$ by \eqref{Cor:BvM2 aux0} and \(\varsigma_n=o(1)\), we note by Assumption \ref{Ass:anti-concentration}(ii) that
\begin{align}\label{Cor:BvM2 aux13}
	c_{J}(1-\alpha-\varsigma_n)-2\ell_n - m_J\ge \varpi_J-2\ell_n \ge \frac{1}{2}\varpi_J
\end{align}
for all $n$ sufficiently large. Given \eqref{Cor:BvM2 aux13}, we may then employ arguments analogus to those leading to result \eqref{Cor:BvM2 aux11} to conclude that
\begin{align}\label{Cor:BvM2 aux14}
\liminf_{n\to\infty}P_0(T_{n,J}\leq \hat c_{n,J}(1-\alpha))\ge 1-\alpha.
\end{align}
The conclusion of the corollary then follows from \eqref{Cor:BvM2 aux5}, \eqref{Cor:BvM2 aux11}, and \eqref{Cor:BvM2 aux14}.  
\end{proof}

\begin{proof}[Proof of Proposition \ref{Pro:Primitive}]
We first verify Assumption \ref{Assump:Beta} for the exact posterior under the spike-and-slab prior. By Examples 4 and 5 of  \citet{Castilloetal_BayesianSparse_2015}, their Assumption 1 is satisfied by the spike-and-slab prior. Let $A$ be the value of the fourth constant in their Assumption 1 in this case. We define the following notation. For a subset $S\subset\{1,\ldots,p\}$, define the compatibility number
\begin{align}\label{CompNumber}
	\phi(S):= \inf\left\{\frac{\sqrt{|S|\beta^\intercal \hat{\Omega}_n\beta } }{\nu(\hat{\Omega}_n)\|\beta_S\|_1}: \|\beta_{S^c}\|_1\le 7\|\beta_S\|_1, \beta_S\neq 0\right\}.
\end{align}
% The compatibility number does not directly require sparsity, but reduces
% the problem to approximate sparsity by considering only $\beta_0$ whose coordinates are
% small outside $S_{\beta_0}$.
For a sparsity level $s\leq p$, define the uniform compatibility number for sparse vectors
\begin{align}\label{UCompNumber}
	\bar\phi(s):= \inf\left\{\frac{\sqrt{|S_{\beta}|\beta^\intercal \hat{\Omega}_n\beta}}{\nu(\hat{\Omega}_n)\|\beta\|_1}: 0\neq |S_\beta|\le s\right\}.
\end{align}
% Requirements on these compatibility numbers are standard for high-dimensional problem
% , see Section 6.13 of \citet{BuhlmannvandeGeer_HighDimension_2011} for further discussion.
The posterior contraction rate in Theorem 2 of \cite{Castilloetal_BayesianSparse_2015} depends on both $\phi$ and $\bar\phi$ through the following quantity: for a subset $S\subset\{1,\ldots,p\}$,
\begin{align}\label{MixofCompNumber}
\bar\psi(S) =\bar\phi\left(\left[2+\frac{3}{A}+\frac{33}{\phi(S)^2}\frac{\lambda}{\bar\lambda}\right]|S|\right),
\end{align}
where $\lambda$ is the hyperparameter in the Laplace density for the slab part in \eqref{SSprior}, and $\bar\lambda=2\nu(\hat{\Omega}_n)\sqrt{n\log p}$. Note that \eqref{CompNumber}, \eqref{UCompNumber}, and \eqref{MixofCompNumber} are Definitions 2.1 and 2.2 and Equation (2.5) of \cite{Castilloetal_BayesianSparse_2015}, respectively. Their condition (2.1) is satisfied by Assumption \ref{Assump:Tuning}(i). Thus, by the law of iterated expectation and the dominated convergence theorem, Theorem 2 implies that 
\begin{align}\label{Pro:Primitive:Eqn:1}
\mathbb{E}_{0}\left[\Pi_{\beta}\left(\|\beta-\beta_0\|_1>\frac{M}{\phi(S_{\beta_0})^2\bar\psi(S_{\beta_0})^2\nu(\hat{\Omega}_n)}s_{\beta_0}\sqrt{\frac{ \log p}{n}}\Biggm|Z^{(n)}\right)\right] \to 0
\end{align}
for sufficiently large $M>0$. Below we show that $\nu(\hat{\Omega}_n)$, $\phi(S_{\beta_0})$, and $\bar{\psi}(S_{\beta_0})$ are bounded away from zero with probability approaching one, so that Assumption \ref{Assump:Beta} holds with 
\begin{align}\label{Pro:Primitive:Eqn:rateExact}
\epsilon_n = s_{\beta_0}\sqrt{\frac{\log p}{n}}.
\end{align}

First, by Assumptions \ref{Assump:Moments}(i) and (ii), following the same arguments as in the proof of Lemma \ref{Lem:BootstrapRate} (see \eqref{Lem:BootstrapRate:Eqn:3}), we have
\begin{align}\label{Pro:Primitive:Eqn:2}
\|\hat{\Omega}_n - \Omega_0\|_{\max} = O_{P_{0}}\left( \sqrt{\frac{\log p}{n}}\right) = O_{P_{0}}(1). 
\end{align}
By \eqref{Pro:Primitive:Eqn:2} and Assumption~\ref{Assump:Error}(ii),
$\nu(\hat{\Omega}_n)$ is bounded above and away from zero with
probability approaching one. Second, since $\|\beta_{S}\|_1\leq |S|^{1/2}\|\beta_{S}\|_2\leq |S|^{1/2}\|\beta\|_2$, it suffices to show that $\phi^{\ast}(S_{\beta_0})$ and $\bar{\psi}^{\ast}(S_{\beta_0})$ are bounded away from zero with probability approaching one, where
\begin{align}\label{CompNumberstar}
	\phi^{\ast}(S):= \inf\left\{\frac{\sqrt{\beta^\intercal \hat{\Omega}_n\beta } }{\|\beta\|_2}: \|\beta_{S^c}\|_1\le 7\|\beta_S\|_1, \beta_S\neq 0\right\}
\end{align}
and, for $\bar\phi^{\ast}(s):= \inf\left\{\sqrt{\beta^\intercal \hat{\Omega}_n\beta}/\|\beta\|_2: 0\neq |S_\beta|\le s\right\}$,
\begin{align}\label{MixofCompNumberstar}
\bar\psi^{\ast}(S)=\bar\phi^{\ast}\left(\left[2+\frac{3}{A}+\frac{33}{\phi^{\ast}(S)^2}\frac{\lambda}{\bar\lambda}\right]|S|\right).
\end{align}
By Theorem 1 (or Corollary 1) of \citet{Raskuttietal_RestrictedEigenvalue_2010} and its sub-gaussian extension \citep{Zhou_RestrictedSubgaussian_2009}, there is a constant $L>0$ such that $\phi^{\ast}(S)\geq L$ with probability approaching one uniformly over all subsets $S\subset\{1,\ldots,p\}$ such that $|S|=o(n/\log p)$. So, $\phi^{\ast}(S_{\beta_0})$ is bounded away from zero with probability approaching one, since $s_{\beta_0}\log p/\sqrt{n}\to0$. The uniformity further implies that $\bar\phi^{\ast}(s)\geq L$ with probability approaching one for $s=o(n/\log p)$, since $0=\|\beta_{S_{\beta}^{c}}\|_1\leq 7\|\beta_{S_{\beta}}\|_1$ is always true. Given the range of $\lambda$ in Assumption \ref{Assump:Tuning} (i), $\bar\psi^{\ast}(S_{\beta_0})\geq \bar\phi^{\ast}(\eta_n s_{\beta_0})$ with probability approaching one for any $\eta_n\to\infty$. So, $\bar{\psi}^{\ast}(S_{\beta_0})$ are bounded away from zero with probability approaching one.

For the variational Bayesian approximate posterior in \eqref{VBFamily} and \eqref{VBPost}, we apply Theorem 1 in \cite{RaySzabo2022Variational}. The posterior contraction rate depends on both $\phi$ and $\bar\phi$ instead through the following quantity: for a subset $S\subset\{1,\ldots,p\}$,
\begin{align}\label{MixofCompNumberVB}
\bar\psi_{M}(S) =\bar\phi\left(\left[2+\frac{4M}{A}\left(1+\frac{16}{\phi(S)^2}\frac{\lambda}{\bar\lambda}\right)\right]|S|\right) \text{ for }M>0.
\end{align}
The preceding sparse eigenvalue bounds also imply that
$\bar\psi_{\eta_n}(S_{\beta_0})$ and its counterpart
$\widetilde\psi_{\eta_n}(S_{\beta_0})$ in
\citet{RaySzabo2022Variational} are both bounded away from zero with
probability approaching one whenever
$s_{\beta_0}\eta_n=o(n/\log p)$. For any fixed $C>0$, we apply their
Theorem~1 conditional on designs satisfying these compatibility
bounds and
$\lambda s_{\beta_0}\leq C\nu(\hat{\Omega}_n)\sqrt{n\log p}$,
using $\eta_n^{1/2}$ as the diverging factor. The factor $\eta_n$
in the bound below absorbs any fixed constant depending on $C$.
Integrating over the design and letting $C\to\infty$ yields, for
some constant $M_1>0$,
\begin{align}\label{Pro:Primitive:Eqn:VB}
\mathbb{E}_{0}\left[\tilde\Pi_{\beta}\left(\|\beta-\beta_0\|_1>\frac{M_1\eta_n}{\phi(S_{\beta_0})^2\bar\psi_{\eta_n}(S_{\beta_0})^2\nu(\hat{\Omega}_n)}s_{\beta_0}\sqrt{\frac{ \log p}{n}}\Biggm|Z^{(n)}\right)\right] \to 0
\end{align}
for any $\eta_n\to\infty$ satisfying
$s_{\beta_0}\eta_n=o(n/\log p)$. Taking $\eta_n=\rho_n$ and using
the lower bounds on the denominator in
\eqref{Pro:Primitive:Eqn:VB}, we obtain
Assumption~\ref{Assump:Beta} with
\begin{align}\label{Pro:Primitive:Eqn:rateVB}
\epsilon_n = s_{\beta_0}\rho_n\sqrt{\frac{\log p}{n}}
\end{align}
for any $\rho_n\to\infty$ such that $s_{\beta_0}\rho_n = o(\sqrt{n/\log p})$.
 
We next verify Assumption \ref{Assump:Precision} for $\hat\Theta_n$ obtained by the nodewise LASSO regression. A rate for $\|\hat{\Theta}_n-\Theta_0\|_{\infty}$ is directly available from the first result in Theorem 2.4 of \citet{VandeGeer_OnAsymptotically_2014} such that the second condition of Assumption \ref{Assump:Precision} holds with
\begin{align}\label{Pro:Primitive:Eqn:3}
\delta_n = \max_{1\leq j\leq p} s_{\Theta_0,j}\sqrt{\frac{\log p}{n}}.
\end{align}
For the second condition, we make use of (10) in \citet{VandeGeer_OnAsymptotically_2014}: $\|\hat{\Theta}_n\hat{\Omega}_n - I_{p}\|_{\max} \leq \max_{1\leq j\leq p} {\lambda_j}/{\hat{\tau}_j}$.
By the proof of Theorem 2.4 in \citet{VandeGeer_OnAsymptotically_2014} (p.1196), $\sup_{1\leq j\leq p}|\hat\tau^{-1}_{j}-\tau^{-1}_j| = O_{P_0}(\sqrt{\max_{1\leq j\leq p}s_{\Theta_0,j}\log p/n})$ for some $\tau_j$ which is uniformly bounded away zero from over $j$. Since $\max_{1\leq j\leq p}\lambda_j = O(\sqrt{\log p/n})$, it then follows that the first condition of Assumption \ref{Assump:Precision} holds with
\begin{align}\label{Pro:Primitive:Eqn:4}
\gamma_n = \sqrt{\frac{\log p}{n}}.
\end{align}

We now verify Assumption \ref{Assump:Frequentist} for the debiased LASSO estimator. By Theorem 2.4 of \citet{VandeGeer_OnAsymptotically_2014}, we have
\begin{align}\label{Pro:Primitive:Eqn:5}
\hat\beta_n = \beta_0  + \hat\Theta_{n}\frac{1}{n} \sum_{i=1}^{n}X_i\varepsilon_i + \Delta_{1n}, ~\text{ where } \|\Delta_{1n}\|_{\infty} = o_{P_0}(n^{-1/2}).
\end{align}
Since $\|\sum_{i=1}^{n}X_i\varepsilon_i/n\|_{\infty} = O_{P_0}(\sqrt{{\log p}/{n}})$, it follows from \eqref{Pro:Primitive:Eqn:3} that 
\begin{align}\label{Pro:Primitive:Eqn:6}
\left\|(\hat\Theta_{n}-\Theta_0)\frac{1}{n}\sum_{i=1}^{n}X_i\varepsilon_i\right\|_{\infty}&\leq \|\hat\Theta_{n}-\Theta_0\|_{\infty}\left\|\frac{1}{n}\sum_{i=1}^{n}X_i\varepsilon_i\right\|_{\infty}= O_{P_0}\left(\max_{1\leq j\leq p} s_{\Theta_0,j}{\frac{\log p}{n}}\right).
\end{align} 
Therefore, Assumption \ref{Assump:Frequentist} holds with $\Delta_n = (\hat\Theta_{n}-\Theta_0)\sum_{i=1}^{n}X_i\varepsilon_i/n + \Delta_{1n}$ by noting that $\max_{1\leq j\leq p} s_{\Theta_0,j}\log p/\sqrt{n}\to 0$. 

Substituting \eqref{Pro:Primitive:Eqn:rateExact} and
\eqref{Pro:Primitive:Eqn:rateVB}, respectively, together with
\eqref{Pro:Primitive:Eqn:3}--\eqref{Pro:Primitive:Eqn:4}, into
the definition of $a_n$ gives the rates stated in
Proposition~\ref{Pro:Primitive}, up to constant factors. Here,
$\|\Theta_0\|_\infty$ is bounded away from zero because
$(\Theta_0)_{jj}\geq 1/(\Omega_0)_{jj}$ and the diagonal entries of
$\Omega_0$ are bounded. By Assumption~\ref{Assump:Sparsity} and the
condition on $J$ in Proposition~\ref{Pro:Primitive}, these rates
converge to zero for both posteriors when $\rho_n$ diverges
sufficiently slowly. This verifies the rate conditions of
Theorems~\ref{Thm:BvM1} and \ref{Thm:BvM2}.
\end{proof}

\bibliographystyle{ecta}
\phantomsection
\addcontentsline{toc}{section}{References}
%20\bibliography{D:/Dropbox/Common/Latex/bibliography.bib}
\bibliography{bibliography}

\newpage

\setcounter{page}{1}
\section*{\centering\Large\sc Supplemental Appendix}
\addcontentsline{toc}{section}{Supplemental Appendix}

The supplemental appendix is organized as follows. Appendix \ref{Sec: Tech lemmasproofs} collects some supporting technical results, Appendix \ref{Sec: AdditionalTheory} provides alternative verifications of our main assumptions, and Appendix \ref{Sec: additional simulations} conducts additional Monte Carlo simulations.

\section{Technical Lemmas and Proofs}\label{Sec: Tech lemmasproofs}
\renewcommand{\theequation}{B.\arabic{equation}}
\setcounter{equation}{0}

\begin{lem}\label{Lem:BootstrapRate}
Suppose Assumptions \ref{Assump:Moments}(i)-(iii) hold. Then
\[\left\|\sum_{i=1}^n W_{ni}X_i\varepsilon_i\right\|_{\infty}=O_{P_{Z,W}}\left( \sqrt{\frac{\log p}{n}}+\frac{\log^{2} n\log^{2} p}{n}\right)\]
and
\[\left\|\sum_{i=1}^n(W_{ni}-\frac{1}{n})X_iX_i^\intercal\right\|_{\max}=O_{P_{Z,W}}\left( \sqrt{\frac{\log p}{n}}+\frac{\log^{2} n\log^{2} p}{n}\right).\]
\end{lem}
\begin{proof}
We prove the results by applying a maximal inequality for sub-Weibull random variables \citep{KuchibhotlaChakrabortty_MovingBeyond_2022}. See Theorem \ref{Lem:SubWeibull}. We first prove the first result. Since $W_{ni} = \omega_i/\sum_{j=1}^{n}\omega_j$ for $i=1,\ldots, n$ where $\{\omega_i\}_{i=1}^{n} \stackrel{i.i.d.}{\sim} \textup{Exp}(1)$, it follows that 
\begin{align}\label{Lem:BootstrapRate:Eqn:1}
\sum_{i=1}^n W_{ni}X_i\varepsilon_i = \left(\frac{1}{n}\sum_{i=1}^{n}\omega_i\right)^{-1} \frac{1}{n}\sum_{i=1}^n \omega_i X_i\varepsilon_i.
\end{align}
Since $\sum_{i=1}^{n}\omega_i/n\overset{P_{Z,W}}{\to} 1$ by the law of large numbers, by \eqref{Lem:BootstrapRate:Eqn:1} it suffices to show 
\begin{align}\label{Lem:BootstrapRate:Eqn:2}
\left\|\frac{1}{n}\sum_{i=1}^n \omega_i X_i\varepsilon_i\right\|_{\infty} = O_{P_{Z,W}}\left( \sqrt{\frac{\log p}{n}}+\frac{\log^{2} n\log^{2} p}{n}\right). 
\end{align}
Since both $X_i$ and $\varepsilon_i$ are sub-Gaussian by Assumptions \ref{Assump:Moments}(ii) and (iii),  $\omega_i X_i\varepsilon_i$ is sub-Weibull of order $\theta=1/2$ \citep{KuchibhotlaChakrabortty_MovingBeyond_2022}. Thus, \eqref{Lem:BootstrapRate:Eqn:2} is obtained by applying Theorem \ref{Lem:SubWeibull} (by taking $t = M \log p$ for sufficiently large $M$).

For the second result, since $\sum_{i=1}^{n}(W_{ni}-n^{-1}) =0$, we have
\begin{align}\label{Lem:BootstrapRate:Eqn:3}
&\hspace{0.5cm}\left\|\sum_{i=1}^n(W_{ni}-\frac{1}{n})X_iX_i^\intercal\right\|_{\max} \notag\\
&= \left\|\sum_{i=1}^n(W_{ni}-\frac{1}{n})\{X_iX_i^\intercal - E_0[X_iX_i^\intercal]\}\right\|_{\max}\notag\\
&\leq \left\|\sum_{i=1}^nW_{ni}\{X_iX_i^\intercal - \mathbb{E}_0[X_iX_i^\intercal]\}\right\|_{\max} + \left\|\frac{1}{n}\sum_{i=1}^n\{X_iX_i^\intercal - \mathbb{E}_0[X_iX_i^\intercal]\}\right\|_{\max},
\end{align}
where the inequality follows by the triangle inequality. Since $X_i$ is sub-Gaussian by Assumption \ref{Assump:Moments}, $X_iX_i^\intercal-\mathbb{E}_0[X_iX_i^\intercal]$ is sub-exponential and $\omega_i \{X_iX_i^\intercal-\mathbb{E}_0[X_iX_i^\intercal]\}$ is sub-Weibull of order $\theta=1/2$. The second result follows by following the arguments as in \eqref{Lem:BootstrapRate:Eqn:1} and \eqref{Lem:BootstrapRate:Eqn:2}. This completes the proof of the lemma. 
\end{proof}

\begin{lem}\label{Lem:BootApprox}
Let Assumptions \ref{Assump:Moments}(i) and (iv) hold and $S_{n}$ be given as in \eqref{Eqn: Decomposition}. Then for each $j=1,\ldots, p$, there is some $\bar{\mathbb G}_{j}\sim N(0,\sigma^2_{0,j})$ independent of $Z^{(n)}$ such that
\[e_{j}^{\intercal} \sqrt{n} S_{n}=  \bar{\mathbb G}_{j}+ o_{P_{Z,W}}(1),\]
and thus, for $\mathcal{L}_{\Pi_{W}}(e_{j}^{\intercal} \sqrt{n} S_{n} \mid Z^{(n)})$ the conditional law of $e_{j}^{\intercal} \sqrt{n} S_{n}$ given $Z^{(n)}$,
\[d_{BL}\left(\mathcal{L}_{\Pi_{W}}(e_{j}^{\intercal} \sqrt{n} S_{n} \mid Z^{(n)}), N(0,\sigma^2_{0,j}) \right)\overset{P_{0}}{\to} 0.
\]
\end{lem}
\begin{proof}
We prove the first result by applying a strong approximation result for the exchangeable bootstrap (see Theorem \ref{Thm:strong approx}). It is observed that 
\begin{align}\label{Eqn:BootApprox aux1}
e_{j}^{\intercal}\sqrt{n}S_{n} & =  \frac{1}{\sqrt n}\sum_{i=1}^{n}n(W_{ni}-\frac{1}{n}) e_j^\intercal\Theta_0 X_i\varepsilon_i\notag\\
& = \frac{1}{\sqrt n}\sum_{i=1}^{n}nW_{ni}  \left( e_j^\intercal\Theta_0X_i\varepsilon_i-\frac{1}{n}\sum_{i=1}^{n}  e_j^\intercal\Theta_0X_i\varepsilon_i\right).
\end{align}
Consider random variables $\{ e_j^\intercal\Theta_0X_i\varepsilon_i\}_{i=1}^{n}$ and exchangeable weights $\{nW_{ni}\}_{i=1}^{n}$. First, Assumption \ref{Ass:IID} is satisfied by Assumption \ref{Assump:Moments}(i) and (iv) and Lyapunov’s inequality.  Since $W_{ni}$ follows the distribution $\mathrm{Beta}(1,n-1)$ by Corollary G.4 in \citet{GhosalVaart2017Bayesian}, Theorem 1 in \citet{Skorski2023Beta} implies that, for all $t>0$ and $n>2$,
\begin{align}\label{Eqn:BootApprox aux2}
P_{Z,W}(|nW_{ni}-\mathbb E_{Z,W}[nW_{ni}]|>t) &= P_{Z,W}\left(W_{ni}-\mathbb E_{Z,W}[W_{ni}]|>\frac{t}{n}\right)\notag\\
&\le 2\exp\left\{-\frac{t^2}{\frac{2(n-1)}{n+1}+\frac{4(n-2)}{3(n+2)}t}\right\}\notag\\
&\le 2\exp\left\{-\frac{t^2}{2+2t}\right\}.
\end{align}
By simple algebra, we have 
\begin{align}\label{Eqn:BootApprox aux3}
\frac{1}{n}\sum_{i=1}^{n}(nW_{ni}-n\bar W_n)^2-1= \frac{1}{\bar \omega_n^2 }\frac{1}{n}\sum_{i=1}^{n}\omega_i^2-2 = O_{P_{Z,W}}(n^{-1/2}),
\end{align}
where the last equality follows by the central limit theorem together with $\mathbb E_{Z,W}[\omega_i]=1$ and $\mathbb E_{Z,W}[\omega_i^2]=2$. By the virtue of $W_{ni}\sim\mathrm{Beta}(1,n-1)$, we may obtain that
\begin{align}\label{Eqn:BootApprox aux4}
\mathbb E_{Z,W}[|nW_{ni}|^3] = n^3 \frac{1}{n}\frac{2}{n+1}\frac{3}{n+2}
\end{align}
which is uniformly bounded in $n$. Thus, together results \eqref{Eqn:BootApprox aux2}, \eqref{Eqn:BootApprox aux3}, and \eqref{Eqn:BootApprox aux4} verify Assumptions \ref{Ass:Boot}(ii) and (iii), while Assumptions \ref{Ass:Boot}(i) and (iv) are trivially satisfied with $q=3/2$. By Theorem \ref{Thm:strong approx}, we may conclude that there exists some $\bar{\mathbb G}_{j}\sim N(0,\sigma^2_{0,j})$ that is independent of $Z^{(n)}$ and
\begin{align}\label{Eqn:BootApprox aux5}
	e_{j}^{\intercal}\sqrt{n}S_{n} = \bar{\mathbb G}_{j}+ O_{P_{Z,W}}(n^{-1/12}).
\end{align}
The first result of the lemma follows from \eqref{Eqn:BootApprox aux5}.

For the second result, since $\bar{\mathbb G}_{j}$ is independent of $Z^{(n)}$, it follows that 
\begin{align}\label{Eqn:BootApprox aux6}
&\hspace{0.5cm}d_{BL}\left(\mathcal{L}_{\Pi_{W}}(e_{j}^{\intercal} \sqrt{n} S_{n} \mid Z^{(n)}), N(0,\sigma^2_{0,j}) \right) \notag\\
&= \sup_{\|f\|_{BL}\leq 1} |\mathbb E_{Z,W}[f(e_{j}^{\intercal} \sqrt{n} S_{n})\mid Z^{(n)}]-\mathbb E_{Z,W}[f(\bar{\mathbb G}_{j})\mid Z^{(n)}]|\notag\\
&\le \sup_{\|f\|_{BL}\leq 1} \mathbb E_{Z,W}[|f(e_{j}^{\intercal} \sqrt{n} S_{n})-f(\bar{\mathbb G}_{j})|\mid Z^{(n)}],
\end{align}
where $\mathbb E_{Z,W}[\cdot|Z^{(n)}]$ denotes the conditional expectation given $Z^{(n)}$, and the inequality follows by Jensen's inequality. Fix $\epsilon>0$. By the triangle inequality and \eqref{Eqn:BootApprox aux6}, we have
\begin{align}\label{Eqn:BootApprox aux7}
&\hspace{0.5cm}d_{BL}\left(\mathcal{L}_{\Pi_{W}}(e_{j}^{\intercal} \sqrt{n} S_{n} \mid Z^{(n)}), N(0,\sigma^2_{0,j}) \right) \notag\\
	& \le \epsilon + \sup_{\|f\|_{BL}\leq 1} \mathbb E_{Z,W}[|f(e_{j}^{\intercal} \sqrt{n} S_{n})-f(\bar{\mathbb G}_{j})|1\{|e_{j}^{\intercal} \sqrt{n} S_{n}-  \bar{\mathbb G}_{j}|>\epsilon\}\mid Z^{(n)}]\notag\\
	&\le \epsilon + 2 P_{Z,W}(|e_{j}^{\intercal} \sqrt{n} S_{n}-  \bar{\mathbb G}_{j}| >\epsilon \mid Z^{(n)}),
\end{align}
where $P_{Z,W}(\cdot|Z^{(n)})$ denotes the conditional probability given $Z^{n}$, and we have used the fact that $\|f\|_{BL}\leq 1$ implies $\sup_{x}|f(x)|\leq 1$ and $|f(x) - f(y)|\leq |x-y|$ for any $x,y$. By Markov's inequality, Fubini's theorem, and \eqref{Eqn:BootApprox aux5}, we have that $P_{Z,W}(|e_{j}^{\intercal} \sqrt{n} S_{n}-  \bar{\mathbb G}_{j}| >\epsilon \mid Z^{(n)})$ is convergent to zero in probability for each $\epsilon>0$. Since $\epsilon>0$ is arbitrary, the second result of the lemma follows from \eqref{Eqn:BootApprox aux7}.
\end{proof}

\begin{lem}\label{Lem:StrongApprox}
Suppose Assumptions \ref{Assump:Frequentist} and \ref{Assump:Moments} (i) and (iv) hold. Then for each $j=1,\ldots, p$, there exists some $\mathbb G_{j}\sim N(0,\sigma^2_{0,j})$ satisfying
\[e_{j}^{\intercal}\sqrt n(\hat\beta_n-\beta_0) = \mathbb G_{j} +o_{P_0}(1).\]
\end{lem}
\begin{proof}
By Yurinskii's coupling (see, e.g., Theorem 10.10 in \citet{Pollard2002User}), we may obtain that, for each $\delta>0$ and some $\mathbb G_{j}\sim N(0,\sigma^2_{0,j})$,
\begin{align}\label{Lem:StrongApprox aux1}
P_0\left(\left|\frac{1}{\sqrt n} \sum_{i=1}^n e_j^\intercal \Theta_0 X_i\varepsilon_i - \mathbb G_{j}\right|>3\delta\right)\lesssim b_n(1+|\log b_n|),
\end{align}
where $b_n\:= \mathbb E_0[|e_j^\intercal\Theta_0 X_i\varepsilon_i|^3] n^{-1/2}/\delta^3$. Setting $\delta=M(\mathbb E_0[|e_j^\intercal\Theta_0 X_i \varepsilon_i|^3] n^{-1/2})^{1/3}$ in \eqref{Lem:StrongApprox aux1} for any fixed constant $M>0$ yields
\begin{multline}\label{Lem:StrongApprox aux2}
P_0\left(\left|\frac{1}{\sqrt n} \sum_{i=1}^n e_j^\intercal \Theta_0 X_i\varepsilon_i - \mathbb G_{j}\right|>3M\left(\frac{\mathbb E_0[|e_j^\intercal\Theta_0 X_i\varepsilon_i|^3]}{\sqrt n}\right)^{1/3}\right) \lesssim \frac{1}{M^3}(1+3|\log M|).
\end{multline}
By Assumption \ref{Assump:Moments}(iv), it follows from \eqref{Lem:StrongApprox aux2} that
\begin{align}\label{Lem:StrongApprox aux3}
\frac{1}{\sqrt n} \sum_{i=1}^n e_j^\intercal \Theta_0 X_i\varepsilon_i - \mathbb G_{j}=O_{P_0}(n^{-1/6}).
\end{align}
The conclusion of the lemma then follows from \eqref{Lem:StrongApprox aux3} and Assumption \ref{Assump:Frequentist}.
\end{proof}

\begin{lem}\label{Lem:BootApproxSimultaneous}
Let Assumptions \ref{Assump:Moments}(i) and \ref{Assump:MomentsSimultaneous} hold and $S_{n}$ be given in \eqref{Eqn: Decomposition}. If $\nu_{J}^{6}J^{3}\log^{9}(1+J)/n\to 0$, there is some $\tilde{\mathbb G}_{J}\sim N(0,\Sigma_{0,J})$ independent of $Z^{(n)}$ such that
\[\|E_{J}^{\intercal} \sqrt{n} S_{n} - \tilde{\mathbb G}_{J}\|_{\infty} = O_{P_{Z,W}}\left(\left(\frac{\nu_{J}^{2}J\log^{3}(1+J)}{n^{1/3}}\right)^{1/4}\right),\]
and, for $\mathcal{L}_{\Pi_{W}}(E_{J}^{\intercal} \sqrt{n} S_{n} \mid Z^{(n)})$ the conditional law of $E_{J}^{\intercal} \sqrt{n} S_{n}$ given $Z^{(n)}$,
\[d_{BL}\left(\mathcal{L}_{\Pi_{W}}(E_{J}^{\intercal} \sqrt{n} S_{n} \mid Z^{(n)}), N(0,\Sigma_{0,J}) \right)\overset{P_{0}}{\to} 0.\]

\end{lem}
\begin{proof} Following the arguments in \eqref{Eqn:BootApprox aux1}-\eqref{Eqn:BootApprox aux5}, there exists some $\tilde{\mathbb G}_{J}\sim N(0,\Sigma_{0,J})$ that is independent of $Z^{(n)}$ and
\begin{align}\label{Lem:BootApproxSimultaneous:Eqn:1}
\|E_{J}^{\intercal} \sqrt{n} S_{n}-\tilde{\mathbb G}_{J}\|_{\infty} =  O_{P_{Z,W}}\left(\left(\frac{\nu_{J}^{2}J\log^{3}(1+J)}{n^{1/3}}\right)^{1/4}\right).
\end{align}
The second result of the lemma follows by the arguments in \eqref{Eqn:BootApprox aux6} and \eqref{Eqn:BootApprox aux7}. 
\end{proof}

\begin{lem}\label{Lem:StrongApprox2}
Suppose Assumptions \ref{Assump:Frequentist}, \ref{Assump:Moments} (i), and \ref{Assump:MomentsSimultaneous} hold. If $\nu_J^6J^2\log^5(1+J)=o(n)$, then there exists some $\mathbb G_{J}\sim N(0,\Sigma_{0,J})$ satisfying
\[
\left\|E_{J}^{\intercal}\sqrt n(\hat\beta_n-\beta_0) - \mathbb G_{J}\right\|_\infty = O_{P_0}\left(\frac{\nu_JJ^{1/3}\log^{5/6}(1+J)}{n^{1/6}}+\sqrt n\varrho_n\right).
\]
\end{lem}
\begin{proof} By Theorem \ref{Thm:strong approx0} (an extension of Yurinskii's coupling), we may obtain that
\begin{align}\label{Lem:StrongApprox2 aux3}
\left\|\frac{1}{\sqrt n} \sum_{i=1}^n E_J^\intercal \Theta_0 X_i\varepsilon_i - \mathbb G_{J}\right\|_\infty=O_{P_0}\left(\frac{\nu_JJ^{1/3}\log^{5/6}(1+J)}{n^{1/6}}\right).
\end{align}
The conclusion then follows from \eqref{Lem:StrongApprox2 aux3} and Assumptions \ref{Assump:Frequentist} and \ref{Ass:anti-concentration}(iii).
\end{proof}

\subsection{Technical Tools}
We record a maximal inequality for sub-Weibull random variables. Sub-Weibull random variables form a broad class of light-tailed random variables that generalize both sub-Gaussian and sub-exponential distributions. In particular, the sub-Weibull family nests these two as special cases: sub-Gaussian random variables correspond to the sub-Weibull case of order $\theta = 2$, while sub-exponential random variables correspond to $\theta = 1$. 

\begin{defn}[Sub-Weibull Random Variable]
A real-valued random variable $Z$ is said to be \emph{sub-Weibull of order} $\theta>0$, denoted by $Z\sim\mathrm{subW}(\theta)$, if 
\[
  \|Z\|_{\psi_\theta}
  := \inf\left\{\eta>0:\; \mathbb{E}\!\left[\exp\!\left({|X|^{\theta}}/{\eta^{\theta}}\right)\right]\le 2\right\}
  <\infty.
\]
\end{defn}

An important property of sub-Weibull random variables is that their products remain sub-Weibull. Specifically, let $Z_1,\ldots,Z_L$ be (possibly dependent) random variables such that 
$\|Z_\ell\|_{\psi_{\theta_\ell}}<\infty$ for each $\ell=1,\ldots,L$. 
Then the product $\prod_{\ell=1}^{L} Z_\ell$ is sub-Weibull with order $\bar{\theta}=(\sum_{\ell=1}^L \theta_\ell^{-1})^{-1}$ and satisfies $\|\prod_{\ell=1}^{L} Z_\ell\|_{\psi_{\bar{\theta}}}\le \prod_{\ell=1}^{L} \|Z_\ell\|_{\psi_{\theta_\ell}}$. The following result is taken from Theorem 3.4 of \citet{KuchibhotlaChakrabortty_MovingBeyond_2022}.

\begin{thm}\label{Lem:SubWeibull}
Suppose $Z_1,\ldots,Z_n$ are independent mean zero random vectors in $\mathbf{R}^k$, such that for some $\theta>0$ and $K_{n,k}>0$, $\max_{1\leq j\leq k}\max_{1\leq i\leq n}\| e_{j}^{\intercal}Z_{i}\|_{\psi_{\theta}}\leq K_{n,k}$.
Define $\Gamma_{n,k}:=\max_{1\leq j\leq k}\sum_{i=1}^n\mathbb{E}[(e_{j}^{\intercal}Z_i)^2]/n$. Then for any $t\geq 0$, with probability $1-3\exp(-t)$,  
\[\left\|\frac{1}{n}\sum_{i=1}^n Z_i \right\|_{\infty}\leq 7\sqrt{\frac{\Gamma_{n,k}(t+\log k)}{n}}+\frac{C_{\theta}K_{n,k}(\log(2n))^{1/\theta}(t+\log k)^{1/\theta^*}}{n},\]
where $\theta^\ast:=\min\{\theta,1\}$ and $C_{\theta}>0$ is some constant depending only on $\theta$.
\end{thm}

We gather two strong approximation results. The first result, due to \citet[Lemma 38]{BelloniChernozhukovChetverikovFernandez2019QR}, extends the usual Yurinskii coupling to an arbitrary norm, while the second provides a strong approximation for the exchangeable bootstrap and is a version of Lemma A.5 in \citet{FangSantosShaikhTorgovitsky2023LP}.

\begin{ass}\label{Ass:IID}
(i) $\{Z_i\}_{i=1}^n$ are i.i.d. random vectors in $\mathbf R^k$ with common law $P$; (ii) the eigenvalues of $\Sigma:=\mathrm{Var}(Z_i)$ are bounded uniformly in $n$ and $k$ (iii) $\{\mathbb E[\|Z_i\|_\infty^3]\}^{1/3}<M_1$ uniformly in $n$ and $k$ for $M_1\ge 1$.
\end{ass}

\begin{ass}\label{Ass:Boot}
(i) $\{\zeta_{ni}\}_{i=1}^n$ are exchangeable and independent of $\{Z_i\}_{i=1}^n$; (ii) for some $a,b>0$, $P(|\zeta_{in}-\mathbb E[\zeta_{in}]|>t)\le 2\exp\{-{t^2}/{(b+at)}\}$ for all $t\ge 0$ and $n$; (iii) $|\sum_{i=1}^{n}(\zeta_{in}-\bar \zeta_n)^2/n-1|=O_P(n^{-1/2})$ where $\bar{\zeta}_n:=\sum_{i=1}^{n}\zeta_{ni}/n$, and $\sup_n \mathbb E[|\zeta_{in}|^3]<\infty$; (iv) for some $q\in(1,\infty]$, $\{\mathbb E[\|Z_i\|_\infty^{2q}]\}^{1/q}<M_2<\infty$.
\end{ass} % (ii) is implied by the so-called Bernstein condition.

\begin{thm}\label{Thm:strong approx0}
If Assumption \ref{Ass:IID} holds, then there is a $\mathbb G_n\sim N(0,\Sigma)$ such that
\begin{align*}
\left\|\frac{1}{\sqrt n}\sum_{i=1}^n\{Z_i-E[Z_i]\}-\mathbb G_n\right\|_\infty = O_p(r_n)
\end{align*}
provided $r_n\equiv M_1p^{1/3}(\log(1+p))^{5/6}/n^{1/6}=o(1)$.
\end{thm}

\begin{thm}\label{Thm:strong approx}
Suppose Assumptions \ref{Ass:IID} and \ref{Ass:Boot} hold. Then there exists a $\bar{\mathbb G}_n\sim N(0,\Sigma)$ that is independent of $\{Z_i\}_{i=1}^n$ such that
\begin{align*}
\left\|\frac{1}{\sqrt n}\sum_{i=1}^n\zeta_{ni}(Z_i-\bar Z_n)-\bar{\mathbb G}_n\right\|_\infty = O_P(b_n)
\end{align*}
provided $b_n: = \frac{M_1\sqrt{k\log(1+n)}}{n^{1/4}}+(\frac{M_1k\log^{{5}/{2}}(1+k)}{\sqrt n})^{{1}/{3}}+(\frac{M_2k\log^3(1+k)n^{{1}/{q}}}{n})^{{1}/{4}}=o(1)$.
\end{thm}

\section{Additional Theoretical Results}\label{Sec: AdditionalTheory}
\renewcommand{\theequation}{C.\arabic{equation}}
\setcounter{equation}{0}

In this appendix, we provide alternative verifications of Assumptions \ref{Assump:Beta} and \ref{Assump:Precision}. The spike-and-slab prior distribution can be viewed as a finite mixture model. Beyond finite mixtures, one may also consider continuous mixtures. One popular choice is the horseshoe prior of \citet{CarvalhoPolsonScott2010Horseshoe}, which is a continuous scale mixture of Gaussian distributions. The marginal density of such a prior (after integrating out the random scale) exhibits both a pole at zero and heavy tails; see \cite{Bhadraetal_Lassohorsehoe_2019}. It is often studied under the framework of global-local shrinkage priors, as it achieves adaptive sparsity through a global shrinkage parameter and selective signal identification through local shrinkage parameters. For the regression coefficients $\beta = (\beta_1,\ldots,\beta_p)^{\intercal}$ in \eqref{Eqn:Model}, the prior is specified as follows: independently for all $j$,
\begin{align}\label{Horseshoe1}
	\beta_{j}|\kappa_j,\tau \sim N(0,\kappa_j^2\tau^2),\quad \kappa_j\sim \mathrm{C}^+(0,1), 
\end{align}
where $\mathrm C^+(0,1)$ denotes the standard half Cauchy distribution on $(0,\infty)$. Here, $\tau$ acts as a global shrinkage parameter governing overall sparsity, while $\kappa_j$ serves as a local shrinkage parameter that allows selective shrinkage of individual coefficients. In general, global-local shrinkage priors shrink small coefficients strongly toward zero, while their heavy tails prevent over-shrinkage of large coefficients. Unlike the spike-and-slab prior, the horseshoe prior typically does not produce exact sparsity in the posterior. For variable selection, additional post-processing methods, such as thresholding or projection techniques, are often employed; see \citet{CarvalhoPolsonScott2010Horseshoe}.

%The state of the art seems to be the horseshoe-like priors developed in \citet{BhadraDattaPolsonWillard2021Horseshoe}---see also \citet{Bhadraetal_Lassohorsehoe_2019} (check out Section 3.3 and note in particular Tables 3 and 4). Both horseshoe and horseshoe-like priors admit closed form expressions. For additional follow-up work, see \citet{PasSzaboVaart2017Horseshoe}.

\citet{SongLiang_NearlyOptimal_2023} consider a general class of shrinkage priors that include the horseshoe prior as a special case. Below we summarize their contraction rate result specialized to the following horseshoe prior: for $\beta = (\beta_1,\ldots,\beta_p)^{\intercal}$, 
\begin{align}\label{Horseshoe2}
	\beta_{j}|\kappa_j,\tau\sim N(0,\kappa_j^2\tau^2 \sigma_0^2), \quad \kappa_j^{2}\sim\pi(\kappa_j^{2}), \quad \text{independently for all }j, 
\end{align}
where $\sigma_0^2>0$ is the variance of the error term in \eqref{Eqn:Model} that is assumed to be centered Gaussian, $\tau>0$ is a global shrinkage parameter, and $\kappa_j$'s are local shrinkage parameters. Note that $\kappa_j^2\tau^2$ and $\tau^2$ correspond to $\lambda_j^2$ and $s_n$ respectively in \citet{SongLiang_NearlyOptimal_2023}. The density of $\kappa_j^2\tau^2$ is $x\mapsto \pi({x}/{\tau^2})/{\tau^2}$, 
which is precisely the structure considered in \citet[p.419]{SongLiang_NearlyOptimal_2023}.

\begin{ass}\label{Assump:Horseshoe}
(i) The support of $X$ is contained in a fixed bounded set; (ii) $n=O(p)$; (iii) $\pi(x)=Cx^{-\varkappa}L(x)$ for some constant $C>0$, some constant $\varkappa>1$, and some function $L$ satisfying $|L(x)-1|=O(x^{-\varsigma})$ as $|x|\to\infty$ with $\varsigma>0$; (iv) $\tau$ satisfies $\tau^2=O(\sqrt{s_{\beta_0}\log p/n}/p^{(\varpi+2\varkappa-1)/(2\varkappa-2)})$ for some $\varpi>0$.
\end{ass}

Assumptions \ref{Assump:Horseshoe}(i) and (ii) are Conditions A1(1) and (2) in \citet{SongLiang_NearlyOptimal_2023}. By the arguments below \eqref{MixofCompNumberstar}, with probability approaching one, the smallest eigenvalue of $\sum_{i=1}^{n}X_{i,S}X_{i,S}^{\intercal}/n$ is bounded below by a fixed constant whenever $|S| = o(n/\log p)$ under Assumptions \ref{Assump:Moments}(i) and  \ref{Assump:Error}(ii). Thus, their Condition A1 (3) holds with probability one if $s_{\beta_0}=o(n/\log p)$.  By their Lemma 3.3, Assumption \ref{Assump:Horseshoe}(iii) implies that the prior of $\beta_j$ is proportional to $\beta_j^{-(2\varkappa-1)}$ as $|\beta_j|\to\infty$ (i.e., polynomial-tailed). The assumption is satisfied by the standard half Cauchy distribution with $\varkappa= 2$. To meet Assumption \ref{Assump:Horseshoe}(iv), \citet{SongLiang_NearlyOptimal_2023} suggest choosing $\tau^2$ deterministically. In particular, it is sufficient to choose $\tau^2$ such that $\log \tau^2=-c\log p$ for some constant $c>0$ satisfying $c\ll (\varpi+2\varkappa-1)/(2\varkappa-2)$. %They also note that certain Bayesian approaches of assigning a hyper prior for $\tau^2$ in the literature does not work well, including those in \citet{PasSzaboVaart2017Adaptive} and \citet{CarvalhoPolsonScott2010Horseshoe} (the standard half-Cauchy distribution).

The following result follows directly from Corollary 3.2 of \citet{SongLiang_NearlyOptimal_2023}. 

\begin{pro}\label{Pro:Horseshoe}
Consider the horseshoe prior in \eqref{Horseshoe2}. Suppose Assumptions \ref{Assump:Moments}(i), \ref{Assump:Error}(ii), and \ref{Assump:Horseshoe} hold. Assume additionally that $X_i$ and $\varepsilon_i$ are independent and $\varepsilon_i\sim N(0,\sigma_0^2)$ for a fixed $\sigma_0^2>0$. If in addition $s_{\beta_0}\log p=o(n)$ and $\log\|\beta_0\|_\infty=O(\log p)$,  then
\begin{align*}
	\mathbb E_{0}\left[\Pi_{\beta}\left(\|\beta-\beta_0\|_1\geq  M\sigma_0s_{\beta_0} \sqrt{\frac{\log p}{n}}\Biggm|Z^{(n)}\right)\right]\to 0,
\end{align*}
for all sufficiently large $M>0$. That is, Assumption \ref{Assump:Beta} is satisfied with 
\[\epsilon_n = s_{\beta_0}\sqrt{\frac{\log p}{n}}.\] 
\end{pro}

The CLIME approach of \citet{Caietal_ConstrainedSparse_2011} provides an alternative construction of $\hat{\Theta}_n$. Specifically, we may estimate the precision matrix by $\hat{\Theta}_n$ that solves the following constrained $\ell_1$ minimization problem:
\begin{align}\label{CLIME1}
\min_{\Theta}\sum_{j=1}^{p}\|\Theta^{\intercal}e_{j}\|_{1}
\quad \text{subject to}\quad 
\|\Theta\hat{\Omega}_n - I_p\|_{\max}\leq \kappa,
\end{align}
where $\kappa>0$ is a tuning parameter. As noted in their Lemma 1, the problem in \eqref{CLIME1} is equivalent to solving $p$ rowwise linear programs:
\begin{align}\label{CLIME2}
\min_{\theta}\|\theta\|_{1}\quad \text{subject to}\quad 
\|\hat{\Omega}_n\theta - e_j\|_{\max}\leq \kappa, 
\quad j=1,\ldots, p. 
\end{align}
Denoting the solutions to \eqref{CLIME2} by $\hat{\theta}_1,\ldots,\hat\theta_p$, we have 
$\hat{\Theta}_n = (\hat{\theta}_1,\ldots,\hat\theta_p)^{\intercal}$.
Each subproblem in \eqref{CLIME2} can be solved efficiently by linear programming, making $\hat{\Theta}_n$ computationally attractive. A symmetrized version can be obtained following (2) of \citet{Caietal_ConstrainedSparse_2011}. 

While \citet{Caietal_ConstrainedSparse_2011} focus on the inverse covariance matrix, 
we instead consider the inverse of the second moment matrix. 
Using the same arguments as in Theorem 6 of their paper, we can establish a rate for $\|\hat{\Theta}_n-\Theta_0\|_{\max}$.
Specifically, if $\kappa \geq \|\Theta_0\|_{\infty}\|\hat{\Omega}_n - \Omega_0\|_{\max}$, then
$\|\hat{\Theta}_n-\Theta_0\|_{\max}\leq 4\kappa \|\Theta_0\|_{\infty}$.
Under Assumptions \ref{Assump:Moments}(i) and (ii), 
$\|\hat{\Omega}_n-\Omega_0\|_{\max} = O_{P_{0}}(\sqrt{\log p /n})$ 
(also see \eqref{Pro:Primitive:Eqn:2}).
Hence, by choosing 
\begin{align}\label{CLIME3}
\kappa \asymp \|\Theta_0\|_{\infty}\sqrt{\frac{\log p}{n}},
\end{align}
we obtain $\|\hat{\Theta}_n-\Theta_0\|_{\max} = 
O_{P_{0}}(\|\Theta_0\|_{\infty}^{2}\sqrt{\log p /n})$.
To establish a bound on $\|\hat{\Theta}_n-\Theta_0\|_{\infty}$, 
we note that $\|\Theta_0\hat{\Omega}_n - I_p\|_{\max}\leq \kappa$ for 
$\kappa \geq \|\Theta_0\|_{\infty}\|\hat{\Omega}_n - \Omega_0\|_{\max}$.
Writing ${\Theta}_0 = (\theta_{0,1},\ldots,\theta_{0,p})^{\intercal}$, 
$\Theta_0$ satisfies the constraint in \eqref{CLIME1}, 
so $\|\hat{\theta}_j\|_1\leq \|\theta_{0,j}\|_1$ for all $j$. 
For any $x,y\in\mathbf{R}^{p}$ with $\|x\|_1\leq \|y\|_1$, we have 
\begin{align} \label{CLIME4}
\|x-y\|_{1} 
= \|x_{S_y^c}\|_{1} + \|x_{S_y}-y_{S_y}\|_{1}
\leq 2\|x_{S_y}-y_{S_y}\|_{1},
\end{align}
where the inequality holds since $\|x_{S_y^c}\|_{1}= \|x\|_1-\|x_{S_y}\|_{1}\leq \|y\|_1-\|x_{S_y}\|_{1} = \|y_{S_y}\|_1-\|x_{S_y}\|_{1}\leq \|x_{S_y}-y_{S_y}\|_{1}$. Hence, $\|\hat{\theta}_j-\theta_{0,j}\|_{1} \leq 2s_{\Theta_0,j}\|\hat{\Theta}_n-\Theta_0\|_{\max}$ for all $j=1,\ldots, p$. 
Therefore, the second condition of Assumption \ref{Assump:Precision} holds with 
\begin{align} \label{CLIME5}
\delta_n = \max_{1\leq j\leq p}s_{\Theta_0,j}\|\Theta_0\|^{2}_{\infty}\sqrt{\frac{\log p}{n}},
\end{align}
while the first condition holds with 
\begin{align} \label{CLIME6}
\gamma_n = \|\Theta_0\|_{\infty}\sqrt{\frac{\log p}{n}}.
\end{align}

\begin{pro}\label{Pro:CLIME}
Consider the precision matrix estimator in \eqref{CLIME1}. Suppose Assumptions \ref{Assump:Moments}(i) and (ii) hold. If in addition $\kappa$ satisfies \eqref{CLIME3},  Assumption \ref{Assump:Precision} is satisfied with $\delta_n$ and $\gamma_n$ given in \eqref{CLIME5} and \eqref{CLIME6}. 
\end{pro}

\section{Additional Simulation Results}\label{Sec: additional simulations}
\subsection{Sensitivity to the Hyperparameter}\label{Sec: lambda sensitivity}

We examine the sensitivity of our simulation results with respect to the key hyper-parameter in the spike-and-slab prior, based on the variational approximation of \citet{RaySzabo2022Variational} and its posterior-corrected counterpart at $\lambda\in\{1/2,1,2\}$, halving and doubling the default value $\lambda=1$. Recall that the slab part, i.e., the Laplace density, is $g_\lambda(b)=(\lambda/2)\exp(-\lambda|b|)$. The reported comparisons assess the sensitivity of this Laplace slab rate $\lambda$ around the default value at one. All other settings are held fixed within each exercise, which uses 1,000 replications and 8,000 retained posterior draws per sample. The benchmark experiment also uses 8,000 Bayesian-bootstrap draws per sample, whereas the calibrated comparisons reuse the same simulated samples across $\lambda$.

\paragraph{Benchmark simulation.}
The additional experiment uses design S1 in Section~5.1 with $n=100$, $p\in\{50,100,200\}$, independent Gaussian errors, $X_i\sim N(0,\operatorname{diag}(1,1/2,\ldots,1/p))$, and the first five coefficients equal to $(0.25,0.5,0.75,1,2)$; the remaining coefficients are zero. Figure~\ref{Fig: lambda benchmark} shows coverage of nominal 95\% intervals for every nonzero coefficient and average coverage across the zero coefficients. The uncorrected procedure remains conservative at zero (98.55\%--99.92\%) and often undercovers nonzero coefficients (27.6\%--93.5\%). Posterior correction brings these ranges to 94.07\%--95.80\% and 90.0\%--95.0\%, respectively. Coverage varies little with $\lambda$ at $p=50$ and $100$, but sensitivity is more visible at $p=200$: for the coefficient $0.75$, corrected coverage is 92.0\%, 92.2\%, and 90.0\% at $\lambda=1/2$, $1$, and $2$. Thus the correction improves coverage for small nonzero coefficients and reduces overcoverage at zero.

\begin{figure}[p]
\centering
\includegraphics[width=\linewidth]{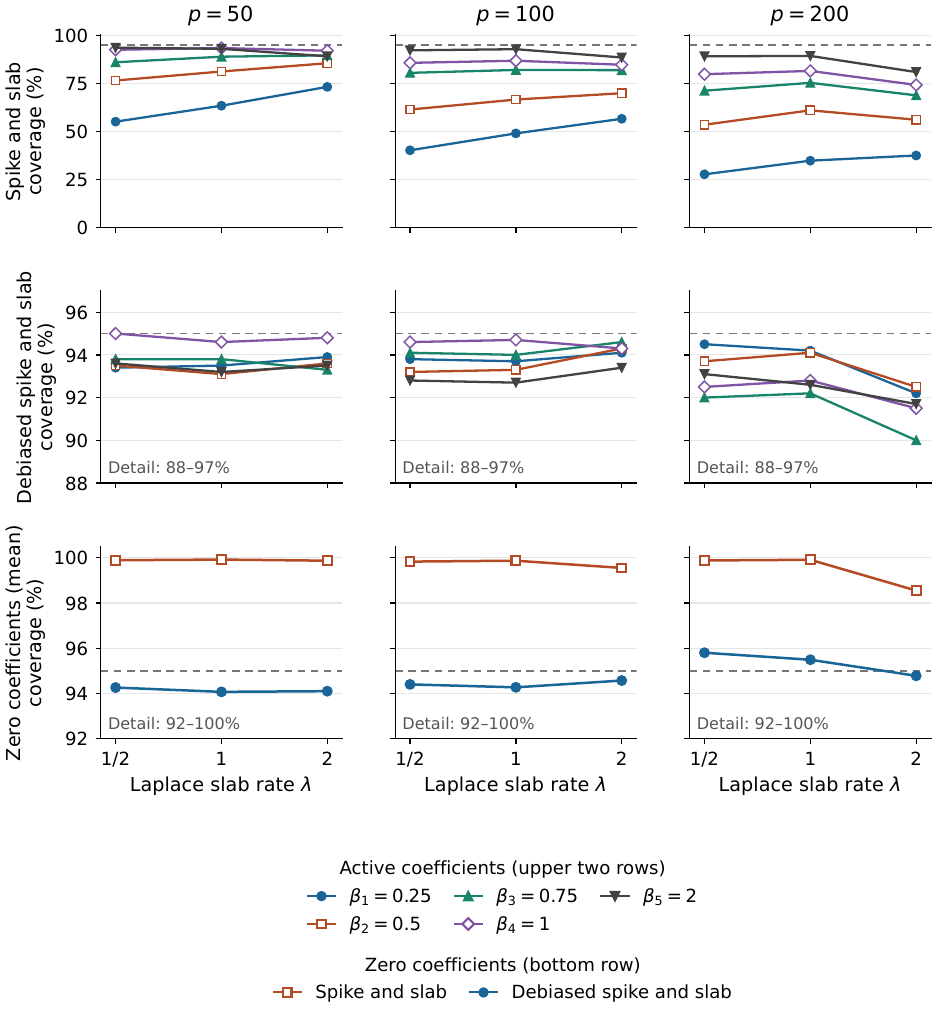}
\caption{Local slab-rate sensitivity in benchmark design S1. Columns correspond to $p=50,100,200$, with $n=100$. The upper two rows show all five nonzero coefficients separately under the uncorrected Bayesian and posterior-corrected procedures, whereas the bottom row averages coverage across $j=6,\ldots,p$. Dashed lines mark 95\%. }\label{Fig: lambda benchmark}
\end{figure}

\paragraph{Calibrated simulations.}
Figure~\ref{Fig: lambda calibrated} repeats the comparison for both calibrated designs in Section~\ref{Sec: calibrated simulations}, at nominal levels of 90\% and 95\%. In the first example on the economic growth, corrected 95\% coverage is 91.7\%, 91.4\%, and 90.7\% at $\lambda=1/2$, $1$, and $2$, compared with 1.7\%, 2.2\%, and 2.2\% before correction. Corrected 90\% coverage is 86.5\%, 85.5\%, and 84.7\%. In the second example about the crime rate, uncorrected coverage is zero throughout, while corrected coverage remains 90.6\% and 94.8\% at the two nominal levels. Here the fitted variational distribution with all $\lambda$ under consideration remains a point mass at zero in every replication.

\begin{figure}[p]
\centering
\includegraphics[width=\linewidth]{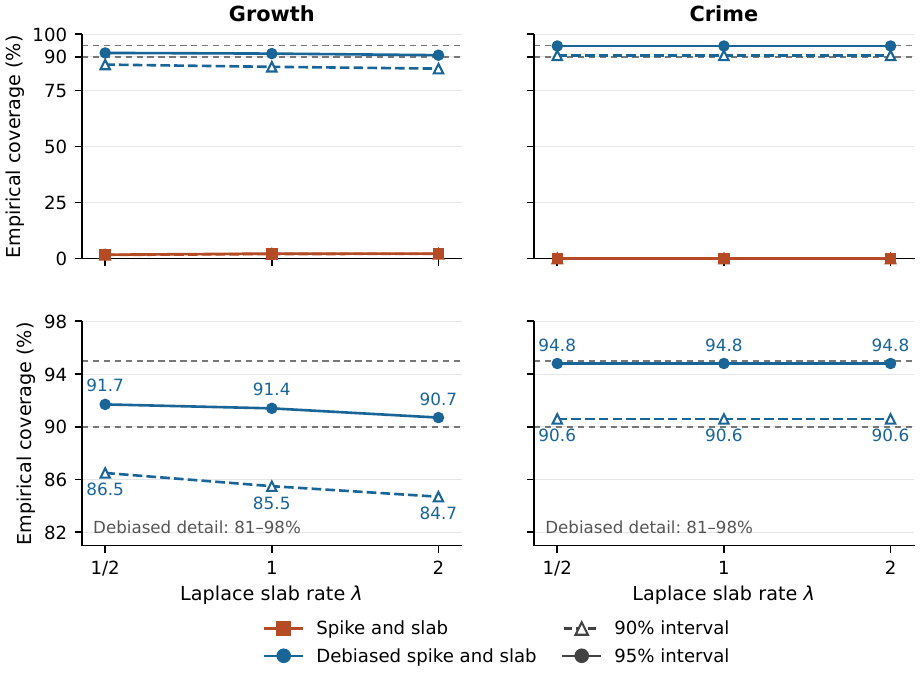}
\caption{Local slab-rate sensitivity in the calibrated growth and crime designs. Empirical coverage is computed on the same 1,000 samples at each displayed $\lambda$, with all other settings fixed within each design. Both nominal levels and both spike-and-slab procedures are shown. Dashed lines mark the nominal levels.}\label{Fig: lambda calibrated}
\end{figure}

\subsection{Simulation with Horseshoe Priors}
\renewcommand{\theequation}{D.\arabic{equation}}
\setcounter{equation}{0}
\pgfplotstableread{
	beta  alpha  Horseshoe  DebiasHorseshoe  DebiasLasso
	0     0.95   0.9987111   0.9632222        0.9522
	1     0.95   0.7760000   0.9210000        0.8590
	2     0.95   0.7910000   0.9340000        0.9380
	3     0.95   0.7760000   0.9380000        0.9500
	4     0.95   0.7690000   0.9310000        0.9410
	5     0.95   0.8910000   0.9460000        0.9680
}\HCovpfnoheterSigmaone  %Horseshoe p50 n100 heter Sigma1

\pgfplotstableread{
	beta  alpha  Horseshoe  DebiasHorseshoe  DebiasLasso
	0     0.95   0.9997778   0.9599333        0.9498444
	1     0.95   0.8720000   0.9190000        0.9370000
	2     0.95   0.7690000   0.9470000        0.9390000
	3     0.95   0.7400000   0.9430000        0.9550000
	4     0.95   0.7140000   0.9450000        0.9550000
	5     0.95   0.8120000   0.9400000        0.9560000
}\HCovpfnohomochiSigmaone  %Horseshoe p50 n100 homochi Sigma1

\pgfplotstableread{
	beta  alpha  Horseshoe  DebiasHorseshoe  DebiasLasso
	0     0.95   0.9992      0.9680889        0.9544222
	1     0.95   0.8730      0.9410000        0.9470000
	2     0.95   0.8670      0.9280000        0.9240000
	3     0.95   0.8890      0.9240000        0.9320000
	4     0.95   0.9030      0.9330000        0.9360000
	5     0.95   0.9500      0.9520000        0.9510000
}\HCovpfnohomonSigmaone  %Horseshoe p50 n100 homon Sigma1

\pgfplotstableread{
	beta  alpha  Horseshoe   DebiasHorseshoe  DebiasLasso
	0     0.95   0.9998737   0.9707263        0.9519053
	1     0.95   0.7680000   0.8970000        0.8520000
	2     0.95   0.7450000   0.9370000        0.9500000
	3     0.95   0.7480000   0.9310000        0.9500000
	4     0.95   0.7180000   0.9250000        0.9470000
	5     0.95   0.8900000   0.9460000        0.9590000
}\HCovponoheterSigmaone  %Horseshoe p100 n100 heter S1

\pgfplotstableread{
	beta  alpha  Horseshoe   DebiasHorseshoe  DebiasLasso
	0     0.95   0.9999053   0.9653579        0.9499684
	1     0.95   0.8540000   0.9430000        0.9580000
	2     0.95   0.7240000   0.9530000        0.9560000
	3     0.95   0.6640000   0.9350000        0.9370000
	4     0.95   0.6900000   0.9470000        0.9530000
	5     0.95   0.7870000   0.9460000        0.9610000
}\HCovponohomochiSigmaone  %Horseshoe p100 n100 homochi S1

\pgfplotstableread{
	beta  alpha  Horseshoe   DebiasHorseshoe  DebiasLasso
	0     0.95   0.9997579   0.9763789        0.9560211
	1     0.95   0.8630000   0.9220000        0.9440000
	2     0.95   0.8390000   0.9170000        0.9310000
	3     0.95   0.8750000   0.9370000        0.9390000
	4     0.95   0.8970000   0.9300000        0.9380000
	5     0.95   0.9420000   0.9490000        0.9480000
}\HCovponohomonSigmaone  %Horseshoe p100 n100 homon S1

\pgfplotstableread{
	beta  alpha  Horseshoe   DebiasHorseshoe  DebiasLasso
	0     0.95   0.9999385   0.9778           0.9519128
	1     0.95   0.7040000   0.8970           0.8770000
	2     0.95   0.7060000   0.9320           0.9400000
	3     0.95   0.6870000   0.9420           0.9410000
	4     0.95   0.7010000   0.9400           0.9520000
	5     0.95   0.8630000   0.9440           0.9550000
}\HCovptnoheterSigmaone  %Horseshoe p200 n100 heter S1

\pgfplotstableread{
	beta  alpha  Horseshoe   DebiasHorseshoe  DebiasLasso
	0     0.95   0.9999692   0.9719538        0.9508462
	1     0.95   0.7770000   0.9340000        0.9450000
	2     0.95   0.6410000   0.9320000        0.9460000
	3     0.95   0.6070000   0.9490000        0.9520000
	4     0.95   0.5930000   0.9340000        0.9580000
	5     0.95   0.7600000   0.9480000        0.9650000
}\HCovptnohomochiSigmaone  %Horseshoe p200 n100 homochi S1

\pgfplotstableread{
	beta  alpha  Horseshoe   DebiasHorseshoe  DebiasLasso
	0     0.95   0.9998974   0.9854974        0.9564872
	1     0.95   0.8250000   0.9090000        0.9550000
	2     0.95   0.8420000   0.9210000        0.9380000
	3     0.95   0.8610000   0.9300000        0.9350000
	4     0.95   0.8960000   0.9330000        0.9420000
	5     0.95   0.9410000   0.9580000        0.9520000
}\HCovptnohomonSigmaone  %Horseshoe p200 n100 homon S1

\pgfplotstableread{
	beta  alpha  Horseshoe   DebiasHorseshoe  DebiasLasso
	0     0.95   0.9990667   0.9547556        0.9522667
	1     0.95   0.2620000   0.9210000        0.8870000
	2     0.95   0.1650000   0.9120000        0.9100000
	3     0.95   0.2330000   0.9030000        0.8820000
	4     0.95   0.4630000   0.9180000        0.6290000
	5     0.95   0.7100000   0.9200000        0.8400000
}\HCovpfnoheterSigmatwo  %Horseshoe p50 n100 heter Sigma2

\pgfplotstableread{
	beta  alpha  Horseshoe  DebiasHorseshoe  DebiasLasso
	0     0.95   1.0000000  0.9544           0.9522222
	1     0.95   0.2120000  0.9400           0.9480000
	2     0.95   0.1470000  0.9230           0.9270000
	3     0.95   0.1900000  0.9060           0.9110000
	4     0.95   0.3900000  0.9080           0.6450000
	5     0.95   0.6460000  0.9220           0.8710000
}\HCovpfnohomochiSigmatwo  %Horseshoe p50 n100 homochi Sigma2

\pgfplotstableread{
	beta  alpha  Horseshoe  DebiasHorseshoe  DebiasLasso
	0     0.95   1.0000000  0.965            0.9519111
	1     0.95   0.5060000  0.911            0.9250000
	2     0.95   0.6300000  0.889            0.8470000
	3     0.95   0.8010000  0.892            0.7230000
	4     0.95   0.9150000  0.921            0.6220000
	5     0.95   0.9490000  0.914            0.7560000
}\HCovpfnohomonSigmatwo  %Horseshoe p50 n100 homon Sigma2

\pgfplotstableread{
	beta  alpha  Horseshoe  DebiasHorseshoe  DebiasLasso
	0     0.95   1.0000000  0.9547895        0.9528421
	1     0.95   0.0910000  0.9430000        0.9080000
	2     0.95   0.0740000  0.9350000        0.9440000
	3     0.95   0.1190000  0.8870000        0.8970000
	4     0.95   0.3210000  0.9050000        0.6760000
	5     0.95   0.5620000  0.9240000        0.8530000
}\HCovponoheterSigmatwo  %Horseshoe p100 n100 heter S2

\pgfplotstableread{
	beta  alpha  Horseshoe  DebiasHorseshoe  DebiasLasso
	0     0.95   1.0000000  0.9527368        0.9528105
	1     0.95   0.0570000  0.9480000        0.9590000
	2     0.95   0.0540000  0.9260000        0.9290000
	3     0.95   0.0890000  0.9100000        0.9310000
	4     0.95   0.2870000  0.8960000        0.7100000
	5     0.95   0.5270000  0.9170000        0.8700000
}\HCovponohomochiSigmatwo  %Horseshoe p100 n100 homochi S2

\pgfplotstableread{
	beta  alpha  Horseshoe  DebiasHorseshoe  DebiasLasso
	0     0.95   0.9977474  0.9692632        0.9601368
	1     0.95   0.2880000  0.9350000        0.9330000
	2     0.95   0.5190000  0.8790000        0.8700000
	3     0.95   0.7510000  0.8970000        0.7330000
	4     0.95   0.8900000  0.9090000        0.5900000
	5     0.95   0.9530000  0.9310000        0.7280000
}\HCovponohomonSigmatwo  %Horseshoe p100 n100 homon S2

\pgfplotstableread{
	beta  alpha  Horseshoe  DebiasHorseshoe  DebiasLasso
	0     0.95   1.0000000  0.9535077        0.9531590
	1     0.95   0.0370000  0.9420000        0.9210000
	2     0.95   0.0380000  0.9220000        0.9440000
	3     0.95   0.0520000  0.8800000        0.9180000
	4     0.95   0.2260000  0.9020000        0.7280000
	5     0.95   0.4420000  0.9260000        0.8850000
}\HCovptnoheterSigmatwo  %Horseshoe p200 n100 heter S2

\pgfplotstableread{
	beta  alpha  Horseshoe  DebiasHorseshoe  DebiasLasso
	0     0.95   1.0000000  0.9525128        0.9530615
	1     0.95   0.0180000  0.9390000        0.9540000
	2     0.95   0.0160000  0.9090000        0.9410000
	3     0.95   0.0420000  0.8820000        0.9340000
	4     0.95   0.1840000  0.9090000        0.7540000
	5     0.95   0.3850000  0.9230000        0.8780000
}\HCovptnohomochiSigmatwo  %Horseshoe p200 n100 homochi S2

\pgfplotstableread{
	beta  alpha  Horseshoe  DebiasHorseshoe  DebiasLasso
	0     0.95   0.9990462  0.9700564        0.9647282
	1     0.95   0.1720000  0.9430000        0.9440000
	2     0.95   0.4290000  0.8690000        0.8640000
	3     0.95   0.6520000  0.8700000        0.7250000
	4     0.95   0.8700000  0.9170000        0.5640000
	5     0.95   0.9490000  0.9560000        0.7400000
}\HCovptnohomonSigmatwo  %Horseshoe p200 n100 homon S2

\pgfplotstableread{
	beta  Horseshoe    DebiasHorseshoe  DebiasLasso
	0     0.0015659    0.0029568         0.0015855
	1     0.1056429    0.0200812         0.0061546
	2     0.2293909    0.0175956         0.0061479
	3     0.3375094    0.0320988         0.0057910
	4     0.4393100    0.0634626         0.0294288
	5     0.3223037    0.0313444         0.0083655
	
}\HBiaspfnoheterSigmaone  %Horseshoe p50 n100 heter Sigma1

\pgfplotstableread{
	beta  Horseshoe    DebiasHorseshoe  DebiasLasso
	0     0.0035822    0.0205809         0.0228993
	1     0.1392791    0.0245394         0.0051747
	2     0.2842965    0.0318579         0.0025413
	3     0.4289701    0.0435347         0.0052156
	4     0.5794847    0.0808911         0.0242093
	5     0.6045194    0.0471718         0.0012220
	
}\HBiaspfnohomochiSigmaone  %Horseshoe p50 n100 homochi Sigma1

\pgfplotstableread{
	beta  Horseshoe    DebiasHorseshoe  DebiasLasso
	0     0.0028244    0.0046566         0.0031779
	1     0.0752019    0.0155349         0.0117140
	2     0.1020635    0.0115795         0.0089705
	3     0.1195556    0.0187743         0.0198663
	4     0.1274907    0.0259005         0.0304509
	5     0.0604226    0.0073524         0.0222719
}\HBiaspfnohomonSigmaone  %Horseshoe p50 n100 homon Sigma1

\pgfplotstableread{
	beta  Horseshoe    DebiasHorseshoe  DebiasLasso
	0     0.0027054    0.0016010         0.0004310
	1     0.1194426    0.0303918         0.0058490
	2     0.2817094    0.0610015         0.0268647
	3     0.3788836    0.0506596         0.0079332
	4     0.4782947    0.0769025         0.0239652
	5     0.3300119    0.0275547         0.0058530
}\HBiasponoheterSigmaone  %Horseshoe p100 n100 heter S1

\pgfplotstableread{
	beta  Horseshoe    DebiasHorseshoe  DebiasLasso
	0     0.0061760    0.0143315         0.0140827
	1     0.1481874    0.0190226         0.0116534
	2     0.3126230    0.0437019         0.0043870
	3     0.4703053    0.0772264         0.0261680
	4     0.5738914    0.0452879         0.0243693
	5     0.6707859    0.0719500         0.0026534
	
}\HBiasponohomochiSigmaone  %Horseshoe p100 n100 homochi S1

\pgfplotstableread{
	beta  Horseshoe    DebiasHorseshoe  DebiasLasso
	0     0.0010356    0.0029000         0.0023041
	1     0.0872685    0.0258190         0.0136999
	2     0.1324778    0.0349940         0.0285193
	3     0.1314595    0.0253750         0.0241935
	4     0.1383359    0.0311334         0.0348457
	5     0.0506238    0.0006050         0.0187188
}\HBiasponohomonSigmaone  %Horseshoe p100 n100 homon S1

\pgfplotstableread{
	beta  Horseshoe    DebiasHorseshoe  DebiasLasso
	0     0.0027712    0.0026460         0.0066543
	1     0.1326482    0.0441098         0.0116349
	2     0.2961015    0.0657982         0.0108203
	3     0.4235386    0.0880147         0.0186378
	4     0.5155163    0.0909964         0.0085993
	5     0.4049892    0.0790871         0.0293302
}\HBiasptnoheterSigmaone  %Horseshoe p200 n100 heter S1

\pgfplotstableread{
	beta  Horseshoe    DebiasHorseshoe  DebiasLasso
	0     0.0035517    0.0271306         0.0323457
	1     0.1641325    0.0363870         0.0055557
	2     0.3414131    0.0730233         0.0123862
	3     0.5198052    0.0887905         0.0224435
	4     0.6743130    0.1301401         0.0368437
	5     0.7963287    0.1362562         0.0389249
	
}\HBiasptnohomochiSigmaone  %Horseshoe p200 n100 homochi S1

\pgfplotstableread{
	beta  Horseshoe    DebiasHorseshoe  DebiasLasso
	0     0.0023330    0.0026537         0.0037572
	1     0.0980262    0.0364055         0.0164813
	2     0.1407563    0.0395384         0.0229221
	3     0.1540943    0.0423698         0.0297046
	4     0.1432345    0.0369463         0.0317285
	5     0.0751611    0.0253543         0.0385599
}\HBiasptnohomonSigmaone  %Horseshoe p200 n100 homon S1

\pgfplotstableread{
	beta  Horseshoe    DebiasHorseshoe  DebiasLasso
	0     0.0033110    0.0102027         0.0924768
	1     0.2293780    0.0597873         0.0402821
	2     0.4579267    0.1141242         0.1241847
	3     0.6436791    0.1516728         0.1652059
	4     0.6294855    0.1274823         0.3259147
	5     0.2739574    0.0658269         0.1817431

}\HBiaspfnoheterSigmatwo  %Horseshoe p50 n100 heter Sigma2

\pgfplotstableread{
	beta  Horseshoe    DebiasHorseshoe  DebiasLasso
	0     0.0061714    0.0167749         0.1099846
	1     0.2394778    0.0592841         0.0383719
	2     0.4651057    0.1028757         0.1197552
	3     0.6569376    0.1433931         0.1465508
	4     0.6775266    0.1284734         0.3340718
	5     0.3058536    0.0728868         0.1811170
	
}\HBiaspfnohomochiSigmatwo  %Horseshoe p50 n100 homochi Sigma2

\pgfplotstableread{
	beta  Horseshoe    DebiasHorseshoe  DebiasLasso
	0     0.0010984    0.0022196         0.0183525
	1     0.1789254    0.0414132         0.0302075
	2     0.2094115    0.0552688         0.0690753
	3     0.1344985    0.0392792         0.1099556
	4     0.0675630    0.0211029         0.1362857
	5     0.0251219    0.0079500         0.0898970
	
}\HBiaspfnohomonSigmatwo  %Horseshoe p50 n100 homon Sigma2

\pgfplotstableread{
	beta  Horseshoe    DebiasHorseshoe  DebiasLasso
	0     0.0019544    0.0093817         0.0933573
	1     0.2416303    0.0675885         0.0393635
	2     0.4817852    0.1145709         0.1099146
	3     0.6986325    0.1768218         0.1569183
	4     0.7427058    0.1552949         0.3197262
	5     0.3374740    0.0868445         0.1801497
}\HBiasponoheterSigmatwo  %Horseshoe p100 n100 heter S2

\pgfplotstableread{
	beta  Horseshoe    DebiasHorseshoe  DebiasLasso
	0     0.0027888    0.0167194         0.1125766
	1     0.2464392    0.0602397         0.0336093
	2     0.4868385    0.1043093         0.1056026
	3     0.7109288    0.1693280         0.1415056
	4     0.7819951    0.1488197         0.3232465
	5     0.3659671    0.0932085         0.1794251

}\HBiasponohomochiSigmatwo  %Horseshoe p100 n100 homochi S2

\pgfplotstableread{
	beta  Horseshoe    DebiasHorseshoe  DebiasLasso
	0     0.0003341    0.0002990         0.0181754
	1     0.2158279    0.0577137         0.0351035
	2     0.2679974    0.0733194         0.0723636
	3     0.1730905    0.0501391         0.1203933
	4     0.0786047    0.0242458         0.1534730
	5     0.0275258    0.0111723         0.1047097
}\HBiasponohomonSigmatwo  %Horseshoe p100 n100 homon S2

\pgfplotstableread{
	beta  Horseshoe    DebiasHorseshoe  DebiasLasso
	0     0.0021544    0.0089392         0.0932504
	1     0.2461889    0.0589050         0.0273293
	2     0.4900134    0.1115116         0.0945867
	3     0.7320815    0.1896711         0.1422347
	4     0.8345876    0.1734593         0.3048099
	5     0.3835798    0.0913690         0.1588771
}\HBiasptnoheterSigmatwo  %Horseshoe p200 n100 heter S2

\pgfplotstableread{
	beta  Horseshoe    DebiasHorseshoe  DebiasLasso
	0     0.0011054    0.0169743         0.1096036
	1     0.2487929    0.0577169         0.0271311
	2     0.4955838    0.1209678         0.1101223
	3     0.7340638    0.1812434         0.1260607
	4     0.8651861    0.1628056         0.2983622
	5     0.4150527    0.1050468         0.1683140

}\HBiasptnohomochiSigmatwo  %Horseshoe p200 n100 homochi S2

\pgfplotstableread{
	beta  Horseshoe    DebiasHorseshoe  DebiasLasso
	0     0.0003050    0.0015297         0.0190292
	1     0.2318257    0.0589779         0.0331154
	2     0.3230196    0.0882988         0.0724974
	3     0.2361063    0.0671364         0.1264725
	4     0.1086897    0.0343801         0.1644692
	5     0.0281382    0.0076782         0.1059272
}\HBiasptnohomonSigmatwo  %Horseshoe p200 n100 homon S2

\pgfplotstableread{
	beta  Horseshoe    DebiasHorseshoe  DebiasLasso
	0     0.1621833    0.4606703         0.5441416
	1     0.2233599    0.2567861         0.2786553
	2     0.3368450    0.2885356         0.3153846
	3     0.4734924    0.3460208         0.3726307
	4     0.6102357    0.4086266         0.4267199
	5     0.6817552    0.4624781         0.4594298
	
}\HRMSEpfnoheterSigmaone  %Horseshoe p50 n100 heter Sigma1

\pgfplotstableread{
	beta  Horseshoe    DebiasHorseshoe  DebiasLasso
	0     0.1893556    0.5997569         0.6778838
	1     0.2111189    0.2574541         0.2839598
	2     0.3736063    0.3550463         0.3906908
	3     0.5493390    0.4356661         0.4665322
	4     0.7113304    0.4912554         0.5193839
	5     0.9902699    0.5880395         0.5860584
}\HRMSEpfnohomochiSigmaone  %Horseshoe p50 n100 homochi Sigma1

\pgfplotstableread{
	beta  Horseshoe    DebiasHorseshoe  DebiasLasso
	0     0.1078735    0.2421218         0.2804317
	1     0.1298004    0.1049769         0.1133537
	2     0.2077481    0.1606335         0.1713026
	3     0.2562824    0.1945380         0.2049842
	4     0.2841239    0.2204666         0.2322905
	5     0.2553760    0.2408694         0.2590445
	
}\HRMSEpfnohomonSigmaone  %Horseshoe p50 n100 homon Sigma1

\pgfplotstableread{
	beta  Horseshoe    DebiasHorseshoe  DebiasLasso
	0     0.1296107    0.4373507         0.5396947
	1     0.2230031    0.2568597         0.2876093
	2     0.3572556    0.2786310         0.3088945
	3     0.4948150    0.3407979         0.3620187
	4     0.6493928    0.4158020         0.4370849
	5     0.6983642    0.4597650         0.4526235
	
}\HRMSEponoheterSigmaone  %Horseshoe p100 n100 heter S1

\pgfplotstableread{
	beta  Horseshoe    DebiasHorseshoe  DebiasLasso
	0     0.1584602    0.5576725         0.6581995
	1     0.2017275    0.2375354         0.2779783
	2     0.3850636    0.3410620         0.3790894
	3     0.5773419    0.4579921         0.4917899
	4     0.7259529    0.5141449         0.5549353
	5     1.0506520    0.5900034         0.5695781
}\HRMSEponohomochiSigmaone  %Horseshoe p100 n100 homochi S1

\pgfplotstableread{
	beta  Horseshoe    DebiasHorseshoe  DebiasLasso
	0     0.08763656   0.2251160         0.2740181
	1     0.13571776   0.1083582         0.1137970
	2     0.22631979   0.1645208         0.1697419
	3     0.26849292   0.1974217         0.2035512
	4     0.30611645   0.2305285         0.2403927
	5     0.25493349   0.2404578         0.2550367
}\HRMSEponohomonSigmaone  %Horseshoe p100 n100 homon S1

\pgfplotstableread{
	beta  Horseshoe    DebiasHorseshoe  DebiasLasso
	0     0.1238089    0.4281696         0.5368550
	1     0.2178885    0.2443793         0.2783154
	2     0.3638890    0.2769387         0.3109570
	3     0.5348957    0.3526196         0.3799376
	4     0.6746917    0.4136281         0.4322975
	5     0.7648796    0.4801275         0.4638324
}\HRMSEptnoheterSigmaone  %Horseshoe p200 n100 heter S1

\pgfplotstableread{
	beta  Horseshoe    DebiasHorseshoe  DebiasLasso
	0     0.1277925    0.5782927         0.6929169
	1     0.2126636    0.2390558         0.2777084
	2     0.4031279    0.3541598         0.3996979
	3     0.5961679    0.4330627         0.4660220
	4     0.7809999    0.5069256         0.5404113
	5     1.1367422    0.6013373         0.5716976
}\HRMSEptnohomochiSigmaone  %Horseshoe p200 n100 homochi S1

\pgfplotstableread{
	beta  Horseshoe    DebiasHorseshoe  DebiasLasso
	0     0.08652936   0.2152969         0.2820404
	1     0.14347312   0.1116047         0.1163172
	2     0.23457421   0.1658087         0.1718338
	3     0.29346200   0.2074454         0.2156093
	4     0.31270024   0.2331667         0.2427989
	5     0.26800511   0.2468332         0.2719630
	
}\HRMSEptnohomonSigmaone  %Horseshoe p200 n100 homon S1

\pgfplotstableread{
	beta  Horseshoe    DebiasHorseshoe  DebiasLasso
	0     0.02756204   0.2054152         0.2046474
	1     0.24557164   0.2907642         0.3104504
	2     0.46577041   0.2424960         0.2600198
	3     0.66495286   0.2620390         0.2602138
	4     0.70808509   0.2652669         0.3858627
	5     0.36529653   0.2458757         0.2675998

}\HRMSEpfnoheterSigmatwo  %Horseshoe p50 n100 heter Sigma2

\pgfplotstableread{
	beta  Horseshoe    DebiasHorseshoe  DebiasLasso
	0     0.03602385   0.2225492        0.2245792
	1     0.24407252   0.2491741        0.2720877
	2     0.47140171   0.2506053        0.2657801
	3     0.67933346   0.2742856        0.2627743
	4     0.75475541   0.2875376        0.4032482
	5     0.39203549   0.2501474        0.2700973
}\HRMSEpfnohomochiSigmatwo  %Horseshoe p50 n100 homochi Sigma2

\pgfplotstableread{
	beta  Horseshoe    DebiasHorseshoe  DebiasLasso
	0     0.02030690   0.08760473       0.08294056
	1     0.20100591   0.11415534       0.11920002
	2     0.26335089   0.12763446       0.13439824
	3     0.20727158   0.12203519       0.15986872
	4     0.14166943   0.11234729       0.17797528
	5     0.09469945   0.10520609       0.14070302

}\HRMSEpfnohomonSigmatwo  %Horseshoe p50 n100 homon Sigma2

\pgfplotstableread{
	beta  Horseshoe    DebiasHorseshoe  DebiasLasso
	0     0.01426283   0.2148858        0.2142471
	1     0.24706069   0.2811261        0.3013392
	2     0.48544953   0.2333684        0.2482455
	3     0.71122282   0.2731912        0.2609650
	4     0.80520433   0.2794152        0.3805860
	5     0.40902182   0.2425602        0.2650259

}\HRMSEponoheterSigmatwo  %Horseshoe p100 n100 heter S2

\pgfplotstableread{
	beta  Horseshoe    DebiasHorseshoe  DebiasLasso
	0     0.01237328   0.2249134        0.2315057
	1     0.24733349   0.2457967        0.2692475
	2     0.48845642   0.2520387        0.2669972
	3     0.71958238   0.2742335        0.2557865
	4     0.83481015   0.2933071        0.3924041
	5     0.43024969   0.2549835        0.2717588
}\HRMSEponohomochiSigmatwo  %Horseshoe p100 n100 homochi S2

\pgfplotstableread{
	beta  Horseshoe    DebiasHorseshoe  DebiasLasso
	0     0.01032684   0.0899872        0.08989901
	1     0.22563225   0.1122279        0.11997834
	2     0.31473020   0.1320989        0.13516457
	3     0.25199782   0.1287008        0.16585410
	4     0.15908920   0.1160940        0.19242741
	5     0.09192330   0.1025238        0.15130191
}\HRMSEponohomonSigmatwo  %Horseshoe p100 n100 homon S2

\pgfplotstableread{
	beta  Horseshoe    DebiasHorseshoe  DebiasLasso
	0     0.03261598   0.2102031        0.2146040
	1     0.24857250   0.2748019        0.2976811
	2     0.49254533   0.2365470        0.2488787
	3     0.73566556   0.2805174        0.2581164
	4     0.87243939   0.2786525        0.3629847
	5     0.43750101   0.2435472        0.2545549
}\HRMSEptnoheterSigmatwo  %Horseshoe p200 n100 heter S2

\pgfplotstableread{
	beta  Horseshoe    DebiasHorseshoe  DebiasLasso
	0     0.00713519   0.2226069        0.2322819
	1     0.24890194   0.2529971        0.2803563
	2     0.49596073   0.2584996        0.2762522
	3     0.73744005   0.2832415        0.2575012
	4     0.89760162   0.2860883        0.3690936
	5     0.46101129   0.2655783        0.2759601
}\HRMSEptnohomochiSigmatwo  %Horseshoe p200 n100 homochi S2

\pgfplotstableread{
	beta  Horseshoe    DebiasHorseshoe  DebiasLasso
	0     0.004441697  0.0865064        0.09127577
	1     0.237044296  0.1102464        0.11992921
	2     0.364674159  0.1434153        0.13812807
	3     0.323157188  0.1436920        0.17382285
	4     0.201454367  0.1231794        0.20269544
	5     0.097450030  0.1025398        0.15347705

}\HRMSEptnohomonSigmatwo  %Horseshoe p200 n100 homon S2
In this appendix, we conduct additional Monte Carlo simulations for regression models in Section \ref{Sec: Monte Carlo} where the coefficient prior is specified as a horseshoe prior. The horseshoe prior is given by the hierarchical structure: independently for all $j=1,\ldots, p$,
\begin{align*}
	\beta_j \mid \lambda_j &\sim \text{Normal}\left(0, \lambda_j^2\right), \\
	\lambda_j \mid \tau &\sim \text{Cauchy}^+(0, \tau), \\
	\tau \mid \sigma &\sim \text{Cauchy}^+(0, \sigma), \\
	\sigma &\sim \text{Cauchy}^+(0, 10).
\end{align*}
Here, the global shrinkage parameter $\tau$ governs the overall degree of regularization by pulling all coefficients toward zero, whereas the local parameters $\lambda_j$ allow individual coefficients to escape this global shrinkage, thereby inducing heavy tails and adaptive sparsity. Posterior draws of $\beta$ are obtained using the MCMC algorithm of \citet{KimLeeGupta2020BayesianSC}. Specifically, we generate 8,000 posterior samples after discarding the initial 8,000 burn-in iterations.

For each replication, we collect posterior samples and implement our proposed posterior-correction procedure as described in Algorithm \ref{Algorithm}. As in the spike-and-slab simulation study, we compare our approach against two benchmarks: (i) standard Bayesian inference under the horseshoe prior, and (ii) debiased inference based on the LASSO estimator. We use the same abbreviations for the methods under comparison as those in Table \ref{tab:methods}. The simulation design and data-generating processes are identical to those used in the spike-and-slab experiments. 

Figures \ref{Fig:hone} and \ref{Fig:htwo} present the empirical coverage rates for different values of $\beta_0$, where the columns correspond to model specifications and the rows to varying dimensionalities $p$. The black, red, and blue curves represent, respectively, the standard Bayesian inference, the Debiased-LASSO, and our proposed horseshoe-based posterior correction. Estimation accuracy---measured by bias and MSE---is summarized in Figures \ref{Fig:hthree}–\ref{Fig:hsix}.

The main findings are as follows. Overall, the results closely mirror those reported in Section \ref{Sec: Monte Carlo}. The standard Bayesian approach with the horseshoe prior improves estimation for nonzero coefficients relative to the spike-and-slab prior; however, estimation for zero coefficients still exhibits non-negligible RMSE (see Figure \ref{Fig:hfive}). Incorporating our debiasing adjustment markedly enhances performance, yielding improved coverage and substantial reductions in both bias and RMSE across all coefficients and simulation scenarios. Compared with the frequentist benchmark, our posterior-correction procedure under the horseshoe prior achieves competitive---and, in several correlated-design settings (S4–S6), superior---performance relative to the Debiased-LASSO.

\begin{figure}[!ht]
	\centering\scriptsize
	
	\begin{tikzpicture} % columns*rows
		\begin{groupplot}[group style={group name=myplots,group size=3 by 3,horizontal sep= 0.8cm,vertical sep=1.1cm},
			grid = minor,
			width = 0.375\textwidth,
			xmax=5,xmin=0,
			ymax=1,ymin=0,
			every axis title/.style={below,at={(0.2,0.8)}},
			xlabel=$\beta_0$,
			x label style={at={(axis description cs:0.95,0.04)},anchor=south},
			xtick={0,1,2,3,4,5},
			ytick={0,0.5,0.95,1},
			tick label style={/pgf/number format/fixed},
			legend style={text=black,cells={align=center},row sep = 3pt,legend columns = -1, draw=none,fill=none},
			cycle list={%
				%{smooth,tension=0.5,color=NavyBlue, no markers,line width=0.25pt, densely dotted}, % alpha
				{smooth,tension=0,color=black, mark=halfsquare*,every mark/.append style={rotate=270},mark size=1.5pt,line width=0.5pt},% LSW-Op
				{smooth,tension=0,color=blue, mark=halfsquare*,every mark/.append style={rotate=90},mark size=1.5pt,line width=0.5pt}, % LSW-Un
				{smooth,tension=0,color=red, mark=10-pointed star,mark size=1.5pt,line width=0.5pt},% Chetverikov
				{smooth,tension=0,color=RoyalBlue1, mark=halfcircle*,every mark/.append style={rotate=90},mark size=1.5pt,line width=0.5pt}, % Quad,Knots 0/// Cub,Knots 3
				{smooth,tension=0,color=RoyalBlue2, mark=halfcircle*,every mark/.append style={rotate=180},mark size=1.5pt,line width=0.5pt}, % Quad,Knots/// Cub,Knots 5
				{smooth,tension=0,color=RoyalBlue3, mark=halfcircle*,every mark/.append style={rotate=270},mark size=1.5pt,line width=0.5pt},% Cub,Knots 0/// Cub,Knots 7
				{smooth,tension=0,color=RoyalBlue4, mark=halfcircle*,every mark/.append style={rotate=360},mark size=1.5pt,line width=0.5pt},% Cub,Knots 1
			}
			]
			% Monotonicity for the univariate case
			\nextgroupplot
			\node[anchor=north] at (axis description cs: 0.25,  0.95) {\fontsize{5}{4}\selectfont \shortstack{
					\\  \\
					Homo-Normal\\
					$p=50$
			}};
			\addplot[smooth,tension=0.5,color=NavyBlue, no markers,line width=0.25pt, densely dotted,forget plot] table[x = beta,y=alpha] from \HCovpfnohomonSigmaone;
			\addplot table[x = beta,y=Horseshoe] from \HCovpfnohomonSigmaone;
			\addplot table[x = beta,y=DebiasHorseshoe] from \HCovpfnohomonSigmaone;
			\addplot table[x = beta,y=DebiasLasso] from \HCovpfnohomonSigmaone;
			
			\nextgroupplot

			\node[anchor=north] at (axis description cs: 0.25,  0.95) {\fontsize{5}{4}\selectfont \shortstack{
					\\  \\
					Homo-Chi\\
					$p=50$
			}};
			\addplot[smooth,tension=0.5,color=NavyBlue, no markers,line width=0.25pt, densely dotted,forget plot] table[x = beta,y=alpha] from \HCovpfnohomochiSigmaone;
			\addplot table[x = beta,y=Horseshoe] from \HCovpfnohomochiSigmaone;
			\addplot table[x = beta,y=DebiasHorseshoe] from \HCovpfnohomochiSigmaone;
			\addplot table[x = beta,y=DebiasLasso] from \HCovpfnohomochiSigmaone;

			\nextgroupplot[legend style = {column sep = 7pt, legend to name = LegendMon17}]
			\addplot[smooth,tension=0.5,color=NavyBlue, no markers,line width=0.25pt, densely dotted,forget plot] table[x = beta,y=alpha] from \HCovpfnoheterSigmaone;
			\addplot table[x = beta,y=Horseshoe] from \HCovpfnoheterSigmaone;
			\addplot table[x = beta,y=DebiasHorseshoe] from \HCovpfnoheterSigmaone;
			\addplot table[x = beta,y=DebiasLasso] from \HCovpfnoheterSigmaone;
			\node[anchor=north] at (axis description cs: 0.25,  0.95) {\fontsize{5}{4}\selectfont \shortstack{
					\\  \\
					Hetero\\
					$p=50$
			}};
			
			% Monotonicity for the bivariate case
			
			\nextgroupplot
			\node[anchor=north] at (axis description cs: 0.25,  0.95) {\fontsize{5}{4}\selectfont \shortstack{
					\\  \\
					Homo-Normal\\
					$p=100$
			}};
			\addplot[smooth,tension=0.5,color=NavyBlue, no markers,line width=0.25pt, densely dotted,forget plot] table[x = beta,y=alpha] from \HCovponohomonSigmaone;
			\addplot table[x = beta,y=Horseshoe] from \HCovponohomonSigmaone;
			\addplot table[x = beta,y=DebiasHorseshoe] from \HCovponohomonSigmaone;
			\addplot table[x = beta,y=DebiasLasso] from \HCovponohomonSigmaone;

			\nextgroupplot
			\node[anchor=north] at (axis description cs: 0.25,  0.95) {\fontsize{5}{4}\selectfont \shortstack{
					\\  \\
					Homo-Chi\\
					$p=100$
			}};
			\addplot[smooth,tension=0.5,color=NavyBlue, no markers,line width=0.25pt, densely dotted,forget plot] table[x = beta,y=alpha] from \HCovponohomochiSigmaone;
			\addplot table[x = beta,y=Horseshoe] from \HCovponohomochiSigmaone;
			\addplot table[x = beta,y=DebiasHorseshoe] from \HCovponohomochiSigmaone;
			\addplot table[x = beta,y=DebiasLasso] from \HCovponohomochiSigmaone;
			
			\nextgroupplot[legend style = {column sep = 3.5pt, legend to name = LegendMon27}]
			
			\node[anchor=north] at (axis description cs: 0.25,  0.95) {\fontsize{5}{4}\selectfont \shortstack{
					\\  \\
					Hetero\\
					$p=100$
			}};
			\addplot[smooth,tension=0.5,color=NavyBlue, no markers,line width=0.25pt, densely dotted,forget plot] table[x = beta,y=alpha] from \HCovponoheterSigmaone;
			\addplot table[x = beta,y=Horseshoe] from \HCovponoheterSigmaone;
			\addplot table[x = beta,y=DebiasHorseshoe] from \HCovponoheterSigmaone;
			\addplot table[x = beta,y=DebiasLasso] from \HCovponoheterSigmaone;

			\nextgroupplot
			\node[anchor=north] at (axis description cs: 0.25,  0.95) {\fontsize{5}{4}\selectfont \shortstack{
					\\  \\
					Homo-Normal\\
					$p=200$
			}};
			\addplot[smooth,tension=0.5,color=NavyBlue, no markers,line width=0.25pt, densely dotted,forget plot] table[x = beta,y=alpha] from \ptonohomon;
			\addplot table[x = beta,y=Horseshoe] from \HCovptnohomonSigmaone;
			\addplot table[x = beta,y=DebiasHorseshoe] from \HCovptnohomonSigmaone;
			\addplot table[x = beta,y=DebiasLasso] from \HCovptnohomonSigmaone;

			\nextgroupplot
			\node[anchor=north] at (axis description cs: 0.25,  0.95) {\fontsize{5}{4}\selectfont \shortstack{
					\\  \\
					Homo-Chi\\
					$p=200$
			}};
			\addplot[smooth,tension=0.5,color=NavyBlue, no markers,line width=0.25pt, densely dotted,forget plot] table[x = beta,y=alpha] from \HCovptnohomochiSigmaone;
			\addplot table[x = beta,y=Horseshoe] from \HCovptnohomochiSigmaone;
			\addplot table[x = beta,y=DebiasHorseshoe] from \HCovptnohomochiSigmaone;
			\addplot table[x = beta,y=DebiasLasso] from \HCovptnohomochiSigmaone;
			
			\nextgroupplot[legend style = {column sep = 3.5pt, legend to name = LegendMon37}]
			
			\node[anchor=north] at (axis description cs: 0.25,  0.95) {\fontsize{5}{4}\selectfont \shortstack{
					\\  \\
					Hetero\\
					$p=200$
			}};
			\addplot[smooth,tension=0.5,color=NavyBlue, no markers,line width=0.25pt, densely dotted,forget plot] table[x = beta,y=alpha] from \HCovptnoheterSigmaone;
			\addplot table[x = beta,y=Horseshoe] from \HCovptnoheterSigmaone;
			\addplot table[x = beta,y=DebiasHorseshoe] from \HCovptnoheterSigmaone;
			\addplot table[x = beta,y=DebiasLasso] from \HCovptnoheterSigmaone;
			\addlegendentry{Bayes}
			\addlegendentry{Posterior-Corrected}
			\addlegendentry{Debiased-LASSO}

		\end{groupplot}
		\node at ($(myplots c2r1) + (0,-2.25cm)$) {\ref{LegendMon17}};
		\node at ($(myplots c2r2) + (0,-2.25cm)$) {\ref{LegendMon27}};
		\node at ($(myplots c2r3) + (0,-2.25cm)$) {\ref{LegendMon37}};
	\end{tikzpicture}
	\caption{Coverage rates corresponding to different values of $\beta_0$ under settings S1 (first column), S2 (second column), and S3 (third column). The dashed line represents the 95\% benchmark.} \label{Fig:hone}
\end{figure}

\begin{figure}[!ht]
	\centering\scriptsize
	
	\begin{tikzpicture} % columns*rows
		\begin{groupplot}[group style={group name=myplots,group size=3 by 3,horizontal sep= 0.8cm,vertical sep=1.1cm},
			grid = minor,
			width = 0.375\textwidth,
			xmax=5,xmin=0,
			ymax=1,ymin=0,
			every axis title/.style={below,at={(0.2,0.8)}},
			xlabel=$\beta_0$,
			x label style={at={(axis description cs:0.95,0.04)},anchor=south},
			xtick={0,1,2,3,4,5},
			ytick={0,0.5,0.95,1},
			tick label style={/pgf/number format/fixed},
			legend style={text=black,cells={align=center},row sep = 3pt,legend columns = -1, draw=none,fill=none},
			cycle list={%
				%{smooth,tension=0.5,color=NavyBlue, no markers,line width=0.25pt, densely dotted}, % alpha
				{smooth,tension=0,color=black, mark=halfsquare*,every mark/.append style={rotate=270},mark size=1.5pt,line width=0.5pt},% LSW-Op
				{smooth,tension=0,color=blue, mark=halfsquare*,every mark/.append style={rotate=90},mark size=1.5pt,line width=0.5pt}, % LSW-Un
				{smooth,tension=0,color=red, mark=10-pointed star,mark size=1.5pt,line width=0.5pt},% Chetverikov
				{smooth,tension=0,color=RoyalBlue1, mark=halfcircle*,every mark/.append style={rotate=90},mark size=1.5pt,line width=0.5pt}, % Quad,Knots 0/// Cub,Knots 3
				{smooth,tension=0,color=RoyalBlue2, mark=halfcircle*,every mark/.append style={rotate=180},mark size=1.5pt,line width=0.5pt}, % Quad,Knots/// Cub,Knots 5
				{smooth,tension=0,color=RoyalBlue3, mark=halfcircle*,every mark/.append style={rotate=270},mark size=1.5pt,line width=0.5pt},% Cub,Knots 0/// Cub,Knots 7
				{smooth,tension=0,color=RoyalBlue4, mark=halfcircle*,every mark/.append style={rotate=360},mark size=1.5pt,line width=0.5pt},% Cub,Knots 1
			}
			]
			% Monotonicity for the univariate case
			\nextgroupplot
			\node[anchor=north] at (axis description cs: 0.25,  0.95) {\fontsize{5}{4}\selectfont \shortstack{
					\\  \\
					Homo-Normal\\
					$p=50$
			}};
			\addplot[smooth,tension=0.5,color=NavyBlue, no markers,line width=0.25pt, densely dotted,forget plot] table[x = beta,y=alpha] from \HCovpfnohomonSigmatwo;
			\addplot table[x = beta,y=Horseshoe] from \HCovpfnohomonSigmatwo;
			\addplot table[x = beta,y=DebiasHorseshoe] from \HCovpfnohomonSigmatwo;
			\addplot table[x = beta,y=DebiasLasso] from \HCovpfnohomonSigmatwo;
			
			\nextgroupplot

			\node[anchor=north] at (axis description cs: 0.25,  0.95) {\fontsize{5}{4}\selectfont \shortstack{
					\\  \\
					Homo-Chi\\
					$p=50$
			}};
			\addplot[smooth,tension=0.5,color=NavyBlue, no markers,line width=0.25pt, densely dotted,forget plot] table[x = beta,y=alpha] from \HCovpfnohomochiSigmatwo;
			\addplot table[x = beta,y=Horseshoe] from \HCovpfnohomochiSigmatwo;
			\addplot table[x = beta,y=DebiasHorseshoe] from \HCovpfnohomochiSigmatwo;
			\addplot table[x = beta,y=DebiasLasso] from \HCovpfnohomochiSigmatwo;

			\nextgroupplot[legend style = {column sep = 7pt, legend to name = LegendMon18}]
			\addplot[smooth,tension=0.5,color=NavyBlue, no markers,line width=0.25pt, densely dotted,forget plot] table[x = beta,y=alpha] from \HCovpfnoheterSigmatwo;
			\addplot table[x = beta,y=Horseshoe] from \HCovpfnoheterSigmatwo;
			\addplot table[x = beta,y=DebiasHorseshoe] from \HCovpfnoheterSigmatwo;
			\addplot table[x = beta,y=DebiasLasso] from \HCovpfnoheterSigmatwo;
			\node[anchor=north] at (axis description cs: 0.25,  0.95) {\fontsize{5}{4}\selectfont \shortstack{
					\\  \\
					Hetero\\
					$p=50$
			}};
			
			% Monotonicity for the bivariate case
			
			\nextgroupplot
			\node[anchor=north] at (axis description cs: 0.25,  0.95) {\fontsize{5}{4}\selectfont \shortstack{
					\\  \\
					Homo-Normal\\
					$p=100$
			}};
			\addplot[smooth,tension=0.5,color=NavyBlue, no markers,line width=0.25pt, densely dotted,forget plot] table[x = beta,y=alpha] from \HCovponohomonSigmatwo;
			\addplot table[x = beta,y=Horseshoe] from \HCovponohomonSigmatwo;
			\addplot table[x = beta,y=DebiasHorseshoe] from \HCovponohomonSigmatwo;
			\addplot table[x = beta,y=DebiasLasso] from \HCovponohomonSigmatwo;

			\nextgroupplot
			\node[anchor=north] at (axis description cs: 0.25,  0.95) {\fontsize{5}{4}\selectfont \shortstack{
					\\  \\
					Homo-Chi\\
					$p=100$
			}};
			\addplot[smooth,tension=0.5,color=NavyBlue, no markers,line width=0.25pt, densely dotted,forget plot] table[x = beta,y=alpha] from \HCovponohomochiSigmatwo;
			\addplot table[x = beta,y=Horseshoe] from \HCovponohomochiSigmatwo;
			\addplot table[x = beta,y=DebiasHorseshoe] from \HCovponohomochiSigmatwo;
			\addplot table[x = beta,y=DebiasLasso] from \HCovponohomochiSigmatwo;
			
			\nextgroupplot[legend style = {column sep = 3.5pt, legend to name = LegendMon28}]
			
			\node[anchor=north] at (axis description cs: 0.25,  0.95) {\fontsize{5}{4}\selectfont \shortstack{
					\\  \\
					Hetero\\
					$p=100$
			}};
			\addplot[smooth,tension=0.5,color=NavyBlue, no markers,line width=0.25pt, densely dotted,forget plot] table[x = beta,y=alpha] from \HCovponoheterSigmatwo;
			\addplot table[x = beta,y=Horseshoe] from \HCovponoheterSigmatwo;
			\addplot table[x = beta,y=DebiasHorseshoe] from \HCovponoheterSigmatwo;
			\addplot table[x = beta,y=DebiasLasso] from \HCovponoheterSigmatwo;

			\nextgroupplot
			\node[anchor=north] at (axis description cs: 0.25,  0.95) {\fontsize{5}{4}\selectfont \shortstack{
					\\  \\
					Homo-Normal\\
					$p=200$
			}};
			\addplot[smooth,tension=0.5,color=NavyBlue, no markers,line width=0.25pt, densely dotted,forget plot] table[x = beta,y=alpha] from \ptonohomon;
			\addplot table[x = beta,y=Horseshoe] from \HCovptnohomonSigmatwo;
			\addplot table[x = beta,y=DebiasHorseshoe] from \HCovptnohomonSigmatwo;
			\addplot table[x = beta,y=DebiasLasso] from \HCovptnohomonSigmatwo;

			\nextgroupplot
			\node[anchor=north] at (axis description cs: 0.25,  0.95) {\fontsize{5}{4}\selectfont \shortstack{
					\\  \\
					Homo-Chi\\
					$p=200$
			}};
			\addplot[smooth,tension=0.5,color=NavyBlue, no markers,line width=0.25pt, densely dotted,forget plot] table[x = beta,y=alpha] from \HCovptnohomochiSigmatwo;
			\addplot table[x = beta,y=Horseshoe] from \HCovptnohomochiSigmatwo;
			\addplot table[x = beta,y=DebiasHorseshoe] from \HCovptnohomochiSigmatwo;
			\addplot table[x = beta,y=DebiasLasso] from \HCovptnohomochiSigmatwo;
			
			\nextgroupplot[legend style = {column sep = 3.5pt, legend to name = LegendMon38}]
			
			\node[anchor=north] at (axis description cs: 0.25,  0.95) {\fontsize{5}{4}\selectfont \shortstack{
					\\  \\
					Hetero\\
					$p=200$
			}};
			\addplot[smooth,tension=0.5,color=NavyBlue, no markers,line width=0.25pt, densely dotted,forget plot] table[x = beta,y=alpha] from \HCovptnoheterSigmatwo;
			\addplot table[x = beta,y=Horseshoe] from \HCovptnoheterSigmatwo;
			\addplot table[x = beta,y=DebiasHorseshoe] from \HCovptnoheterSigmatwo;
			\addplot table[x = beta,y=DebiasLasso] from \HCovptnoheterSigmatwo;
			\addlegendentry{Bayes}
			\addlegendentry{Posterior-Corrected}
			\addlegendentry{Debiased-LASSO}

		\end{groupplot}
		\node at ($(myplots c2r1) + (0,-2.25cm)$) {\ref{LegendMon18}};
		\node at ($(myplots c2r2) + (0,-2.25cm)$) {\ref{LegendMon28}};
		\node at ($(myplots c2r3) + (0,-2.25cm)$) {\ref{LegendMon38}};
	\end{tikzpicture}
	\caption{Coverage rates corresponding to different values of $\beta_0$ under settings S1 (first column), S2 (second column), and S3 (third column). The dashed line represents the 95\% benchmark.}\label{Fig:htwo}
\end{figure}

\begin{figure}[!ht]
	\centering\scriptsize
	\begin{tikzpicture} % columns*rows
		\begin{groupplot}[group style={group name=myplots,group size=3 by 3,horizontal sep= 0.8cm,vertical sep=1.1cm},
			grid = minor,
			width = 0.375\textwidth,
			xmax=5,xmin=0,
			ymax=1,ymin=0,
			every axis title/.style={below,at={(0.2,0.8)}},
			xlabel=$\beta_0$,
			x label style={at={(axis description cs:0.95,0.04)},anchor=south},
			xtick={0,1,2,3,4,5},
			ytick={0,0.5,1,1.5},
			tick label style={/pgf/number format/fixed},
			legend style={text=black,cells={align=center},row sep = 3pt,legend columns = -1, draw=none,fill=none},
			cycle list={%
				%{smooth,tension=0.5,color=NavyBlue, no markers,line width=0.25pt, densely dotted}, % alpha
				{smooth,tension=0,color=black, mark=halfsquare*,every mark/.append style={rotate=270},mark size=1.5pt,line width=0.5pt},% LSW-Op
				{smooth,tension=0,color=blue, mark=halfsquare*,every mark/.append style={rotate=90},mark size=1.5pt,line width=0.5pt}, % LSW-Un
				{smooth,tension=0,color=red, mark=10-pointed star,mark size=1.5pt,line width=0.5pt},% Chetverikov
				{smooth,tension=0,color=RoyalBlue1, mark=halfcircle*,every mark/.append style={rotate=90},mark size=1.5pt,line width=0.5pt}, % Quad,Knots 0/// Cub,Knots 3
				{smooth,tension=0,color=RoyalBlue2, mark=halfcircle*,every mark/.append style={rotate=180},mark size=1.5pt,line width=0.5pt}, % Quad,Knots/// Cub,Knots 5
				{smooth,tension=0,color=RoyalBlue3, mark=halfcircle*,every mark/.append style={rotate=270},mark size=1.5pt,line width=0.5pt},% Cub,Knots 0/// Cub,Knots 7
				{smooth,tension=0,color=RoyalBlue4, mark=halfcircle*,every mark/.append style={rotate=360},mark size=1.5pt,line width=0.5pt},% Cub,Knots 1
			}
			]
			% Monotonicity for the univariate case
			\nextgroupplot
			\node[anchor=north] at (axis description cs: 0.25,  0.95) {\fontsize{5}{4}\selectfont \shortstack{
					\\  \\
					Homo-Normal\\
					$p=50$
			}};
			\addplot table[x = beta,y=Horseshoe] from \HBiaspfnohomonSigmaone;
			\addplot table[x = beta,y=DebiasHorseshoe] from \HBiaspfnohomonSigmaone;
			\addplot table[x = beta,y=DebiasLasso] from \HBiaspfnohomonSigmaone;
			
			\nextgroupplot

			\node[anchor=north] at (axis description cs: 0.25,  0.95) {\fontsize{5}{4}\selectfont \shortstack{
					\\  \\
					Homo-Chi\\
					$p=50$
			}};
			\addplot table[x = beta,y=Horseshoe] from \HBiaspfnohomochiSigmaone;
			\addplot table[x = beta,y=DebiasHorseshoe] from \HBiaspfnohomochiSigmaone;
			\addplot table[x = beta,y=DebiasLasso] from \HBiaspfnohomochiSigmaone;

			\nextgroupplot[legend style = {column sep = 7pt, legend to name = LegendMon19}]
			\addplot table[x = beta,y=Horseshoe] from \HBiaspfnoheterSigmaone;
			\addplot table[x = beta,y=DebiasHorseshoe] from \HBiaspfnoheterSigmaone;
			\addplot table[x = beta,y=DebiasLasso] from \HBiaspfnoheterSigmaone;
			\node[anchor=north] at (axis description cs: 0.25,  0.95) {\fontsize{5}{4}\selectfont \shortstack{
					\\  \\
					Hetero\\
					$p=50$
			}};
			
			% Monotonicity for the bivariate case
			\nextgroupplot
			\node[anchor=north] at (axis description cs: 0.25,  0.95) {\fontsize{5}{4}\selectfont \shortstack{
					\\  \\
					Homo-Normal\\
					$p=100$
			}};
			\addplot table[x = beta,y=Horseshoe] from \HBiasponohomonSigmaone;
			\addplot table[x = beta,y=DebiasHorseshoe] from \HBiasponohomonSigmaone;
			\addplot table[x = beta,y=DebiasLasso] from \HBiasponohomonSigmaone;

			\nextgroupplot
			\node[anchor=north] at (axis description cs: 0.25,  0.95) {\fontsize{5}{4}\selectfont \shortstack{
					\\  \\
					Homo-Chi\\
					$p=100$
			}};
			\addplot table[x = beta,y=Horseshoe] from \HBiasponohomochiSigmaone;
			\addplot table[x = beta,y=DebiasHorseshoe] from \HBiasponohomochiSigmaone;
			\addplot table[x = beta,y=DebiasLasso] from \HBiasponohomochiSigmaone;
			
			\nextgroupplot[legend style = {column sep = 3.5pt, legend to name = LegendMon29}]
			
			\node[anchor=north] at (axis description cs: 0.25,  0.95) {\fontsize{5}{4}\selectfont \shortstack{
					\\  \\
					Hetero\\
					$p=100$
			}};
			\addplot table[x = beta,y=Horseshoe] from \HBiasponoheterSigmaone;
			\addplot table[x = beta,y=DebiasHorseshoe] from \HBiasponoheterSigmaone;
			\addplot table[x = beta,y=DebiasLasso] from \HBiasponoheterSigmaone;
			
			\nextgroupplot
			\node[anchor=north] at (axis description cs: 0.25,  0.95) {\fontsize{5}{4}\selectfont \shortstack{
					\\  \\
					Homo-Normal\\
					$p=200$
			}};
			
			\addplot table[x = beta,y=Horseshoe] from \HBiasptnohomonSigmaone;
			\addplot table[x = beta,y=DebiasHorseshoe] from \HBiasptnohomonSigmaone;
			\addplot table[x = beta,y=DebiasLasso] from \HBiasptnohomonSigmaone;
			
			\nextgroupplot
			\node[anchor=north] at (axis description cs: 0.25,  0.95) {\fontsize{5}{4}\selectfont \shortstack{
					\\  \\
					Homo-Chi\\
					$p=200$
			}};
			\addplot table[x = beta,y=Horseshoe] from \HBiasptnohomochiSigmaone;
			\addplot table[x = beta,y=DebiasHorseshoe] from \HBiasptnohomochiSigmaone;
			\addplot table[x = beta,y=DebiasLasso] from \HBiasptnohomochiSigmaone;
			
			\nextgroupplot[legend style = {column sep = 3.5pt, legend to name = LegendMon39}]
			
			\node[anchor=north] at (axis description cs: 0.25,  0.95) {\fontsize{5}{4}\selectfont \shortstack{
					\\  \\
					Heter\\
					$p=200$
			}};
			
			\addplot table[x = beta,y=Horseshoe] from \HBiasptnoheterSigmaone;
			\addplot table[x = beta,y=DebiasHorseshoe] from \HBiasptnoheterSigmaone;
			\addplot table[x = beta,y=DebiasLasso] from \HBiasptnoheterSigmaone;
			\addlegendentry{Bayes}
			\addlegendentry{Posterior-Corrected}
			\addlegendentry{Debiased-LASSO}

		\end{groupplot}
		\node at ($(myplots c2r1) + (0,-2.25cm)$) {\ref{LegendMon19}};
		\node at ($(myplots c2r2) + (0,-2.25cm)$) {\ref{LegendMon29}};
		\node at ($(myplots c2r3) + (0,-2.25cm)$) {\ref{LegendMon39}};
	\end{tikzpicture}
	\caption{Bias corresponding to different values of $\beta_0$ under settings S1 (first column), S2 (second column), and S3 (third column).}\label{Fig:hthree}
\end{figure}

\begin{figure}[!ht]
	\centering\scriptsize
	\begin{tikzpicture} % columns*rows
		\begin{groupplot}[group style={group name=myplots,group size=3 by 3,horizontal sep= 0.8cm,vertical sep=1.1cm},
			grid = minor,
			width = 0.375\textwidth,
			xmax=5,xmin=0,
			ymax=1,ymin=0,
			every axis title/.style={below,at={(0.2,0.8)}},
			xlabel=$\beta_0$,
			x label style={at={(axis description cs:0.95,0.04)},anchor=south},
			xtick={0,1,2,3,4,5},
			ytick={0,0.5,1,1.5},
			tick label style={/pgf/number format/fixed},
			legend style={text=black,cells={align=center},row sep = 3pt,legend columns = -1, draw=none,fill=none},
			cycle list={%
				%{smooth,tension=0.5,color=NavyBlue, no markers,line width=0.25pt, densely dotted}, % alpha
				{smooth,tension=0,color=black, mark=halfsquare*,every mark/.append style={rotate=270},mark size=1.5pt,line width=0.5pt},% LSW-Op
				{smooth,tension=0,color=blue, mark=halfsquare*,every mark/.append style={rotate=90},mark size=1.5pt,line width=0.5pt}, % LSW-Un
				{smooth,tension=0,color=red, mark=10-pointed star,mark size=1.5pt,line width=0.5pt},% Chetverikov
				{smooth,tension=0,color=RoyalBlue1, mark=halfcircle*,every mark/.append style={rotate=90},mark size=1.5pt,line width=0.5pt}, % Quad,Knots 0/// Cub,Knots 3
				{smooth,tension=0,color=RoyalBlue2, mark=halfcircle*,every mark/.append style={rotate=180},mark size=1.5pt,line width=0.5pt}, % Quad,Knots/// Cub,Knots 5
				{smooth,tension=0,color=RoyalBlue3, mark=halfcircle*,every mark/.append style={rotate=270},mark size=1.5pt,line width=0.5pt},% Cub,Knots 0/// Cub,Knots 7
				{smooth,tension=0,color=RoyalBlue4, mark=halfcircle*,every mark/.append style={rotate=360},mark size=1.5pt,line width=0.5pt},% Cub,Knots 1
			}
			]
			% Monotonicity for the univariate case
			\nextgroupplot
			\node[anchor=north] at (axis description cs: 0.25,  0.95) {\fontsize{5}{4}\selectfont \shortstack{
					\\  \\
					Homo-Normal\\
					$p=50$
			}};
			\addplot table[x = beta,y=Horseshoe] from \HBiaspfnohomonSigmatwo;
			\addplot table[x = beta,y=DebiasHorseshoe] from \HBiaspfnohomonSigmatwo;
			\addplot table[x = beta,y=DebiasLasso] from \HBiaspfnohomonSigmatwo;
			
			\nextgroupplot

			\node[anchor=north] at (axis description cs: 0.25,  0.95) {\fontsize{5}{4}\selectfont \shortstack{
					\\  \\
					Homo-Chi\\
					$p=50$
			}};
			\addplot table[x = beta,y=Horseshoe] from \HBiaspfnohomochiSigmatwo;
			\addplot table[x = beta,y=DebiasHorseshoe] from \HBiaspfnohomochiSigmatwo;
			\addplot table[x = beta,y=DebiasLasso] from \HBiaspfnohomochiSigmatwo;

			\nextgroupplot[legend style = {column sep = 7pt, legend to name = LegendMon110}]
			\addplot table[x = beta,y=Horseshoe] from \HBiaspfnoheterSigmatwo;
			\addplot table[x = beta,y=DebiasHorseshoe] from \HBiaspfnoheterSigmatwo;
			\addplot table[x = beta,y=DebiasLasso] from \HBiaspfnoheterSigmatwo;
			\node[anchor=north] at (axis description cs: 0.25,  0.95) {\fontsize{5}{4}\selectfont \shortstack{
					\\  \\
					Hetero\\
					$p=50$
			}};
			
			% Monotonicity for the bivariate case
			\nextgroupplot
			\node[anchor=north] at (axis description cs: 0.25,  0.95) {\fontsize{5}{4}\selectfont \shortstack{
					\\  \\
					Homo-Normal\\
					$p=100$
			}};
			\addplot table[x = beta,y=Horseshoe] from \HBiasponohomonSigmatwo;
			\addplot table[x = beta,y=DebiasHorseshoe] from \HBiasponohomonSigmatwo;
			\addplot table[x = beta,y=DebiasLasso] from \HBiasponohomonSigmatwo;

			\nextgroupplot
			\node[anchor=north] at (axis description cs: 0.25,  0.95) {\fontsize{5}{4}\selectfont \shortstack{
					\\  \\
					Homo-Chi\\
					$p=100$
			}};
			\addplot table[x = beta,y=Horseshoe] from \HBiasponohomochiSigmatwo;
			\addplot table[x = beta,y=DebiasHorseshoe] from \HBiasponohomochiSigmatwo;
			\addplot table[x = beta,y=DebiasLasso] from \HBiasponohomochiSigmatwo;
			
			\nextgroupplot[legend style = {column sep = 3.5pt, legend to name = LegendMon210}]
			
			\node[anchor=north] at (axis description cs: 0.25,  0.95) {\fontsize{5}{4}\selectfont \shortstack{
					\\  \\
					Hetero\\
					$p=100$
			}};
			\addplot table[x = beta,y=Horseshoe] from \HBiasponoheterSigmatwo;
			\addplot table[x = beta,y=DebiasHorseshoe] from \HBiasponoheterSigmatwo;
			\addplot table[x = beta,y=DebiasLasso] from \HBiasponoheterSigmatwo;
			
			\nextgroupplot
			\node[anchor=north] at (axis description cs: 0.25,  0.95) {\fontsize{5}{4}\selectfont \shortstack{
					\\  \\
					Homo-Normal\\
					$p=200$
			}};
			
			\addplot table[x = beta,y=Horseshoe] from \HBiasptnohomonSigmatwo;
			\addplot table[x = beta,y=DebiasHorseshoe] from \HBiasptnohomonSigmatwo;
			\addplot table[x = beta,y=DebiasLasso] from \HBiasptnohomonSigmatwo;
			
			\nextgroupplot
			\node[anchor=north] at (axis description cs: 0.25,  0.95) {\fontsize{5}{4}\selectfont \shortstack{
					\\  \\
					Homo-Chi\\
					$p=200$
			}};
			\addplot table[x = beta,y=Horseshoe] from \HBiasptnohomochiSigmatwo;
			\addplot table[x = beta,y=DebiasHorseshoe] from \HBiasptnohomochiSigmatwo;
			\addplot table[x = beta,y=DebiasLasso] from \HBiasptnohomochiSigmatwo;
			
			\nextgroupplot[legend style = {column sep = 3.5pt, legend to name = LegendMon310}]
			
			\node[anchor=north] at (axis description cs: 0.25,  0.95) {\fontsize{5}{4}\selectfont \shortstack{
					\\  \\
					Heter\\
					$p=200$
			}};
			
			\addplot table[x = beta,y=Horseshoe] from \HBiasptnoheterSigmatwo;
			\addplot table[x = beta,y=DebiasHorseshoe] from \HBiasptnoheterSigmatwo;
			\addplot table[x = beta,y=DebiasLasso] from \HBiasptnoheterSigmatwo;
			\addlegendentry{Bayes}
			\addlegendentry{Posterior-Corrected}
			\addlegendentry{Debiased-LASSO}

		\end{groupplot}
		\node at ($(myplots c2r1) + (0,-2.25cm)$) {\ref{LegendMon110}};
		\node at ($(myplots c2r2) + (0,-2.25cm)$) {\ref{LegendMon210}};
		\node at ($(myplots c2r3) + (0,-2.25cm)$) {\ref{LegendMon310}};
	\end{tikzpicture}
	\caption{Bias corresponding to different values of $\beta_0$ under settings S1 (first column), S2 (second column), and S3 (third column).}\label{Fig:hfour}
\end{figure}

\begin{figure}[!ht]
	\centering\scriptsize
	\begin{tikzpicture} % columns*rows
		\begin{groupplot}[group style={group name=myplots,group size=3 by 3,horizontal sep= 0.8cm,vertical sep=1.1cm},
			grid = minor,
			width = 0.375\textwidth,
			xmax=5,xmin=0,
			ymax=1.5,ymin=0,
			every axis title/.style={below,at={(0.2,0.8)}},
			xlabel=$\beta_0$,
			x label style={at={(axis description cs:0.95,0.04)},anchor=south},
			xtick={0,1,2,3,4,5},
			ytick={0,0.5,1,1.5},
			tick label style={/pgf/number format/fixed},
			legend style={text=black,cells={align=center},row sep = 3pt,legend columns = -1, draw=none,fill=none},
			cycle list={%
				%{smooth,tension=0.5,color=NavyBlue, no markers,line width=0.25pt, densely dotted}, % alpha
				{smooth,tension=0,color=black, mark=halfsquare*,every mark/.append style={rotate=270},mark size=1.5pt,line width=0.5pt},% LSW-Op
				{smooth,tension=0,color=blue, mark=halfsquare*,every mark/.append style={rotate=90},mark size=1.5pt,line width=0.5pt}, % LSW-Un
				{smooth,tension=0,color=red, mark=10-pointed star,mark size=1.5pt,line width=0.5pt},% Chetverikov
				{smooth,tension=0,color=RoyalBlue1, mark=halfcircle*,every mark/.append style={rotate=90},mark size=1.5pt,line width=0.5pt}, % Quad,Knots 0/// Cub,Knots 3
				{smooth,tension=0,color=RoyalBlue2, mark=halfcircle*,every mark/.append style={rotate=180},mark size=1.5pt,line width=0.5pt}, % Quad,Knots/// Cub,Knots 5
				{smooth,tension=0,color=RoyalBlue3, mark=halfcircle*,every mark/.append style={rotate=270},mark size=1.5pt,line width=0.5pt},% Cub,Knots 0/// Cub,Knots 7
				{smooth,tension=0,color=RoyalBlue4, mark=halfcircle*,every mark/.append style={rotate=360},mark size=1.5pt,line width=0.5pt},% Cub,Knots 1
			}
			]
			% Monotonicity for the univariate case
			\nextgroupplot
			\node[anchor=north] at (axis description cs: 0.25,  0.95) {\fontsize{5}{4}\selectfont \shortstack{
					\\  \\
					Homo-Normal\\
					$p=50$
			}};
			\addplot table[x = beta,y=Horseshoe] from \HRMSEpfnohomonSigmaone;
			\addplot table[x = beta,y=DebiasHorseshoe] from \HRMSEpfnohomonSigmaone;
			\addplot table[x = beta,y=DebiasLasso] from \HRMSEpfnohomonSigmaone;
			
			\nextgroupplot

			\node[anchor=north] at (axis description cs: 0.25,  0.95) {\fontsize{5}{4}\selectfont \shortstack{
					\\  \\
					Homo-Chi\\
					$p=50$
			}};
			\addplot table[x = beta,y=Horseshoe] from \HRMSEpfnohomochiSigmaone;
			\addplot table[x = beta,y=DebiasHorseshoe] from \HRMSEpfnohomochiSigmaone;
			\addplot table[x = beta,y=DebiasLasso] from \HRMSEpfnohomochiSigmaone;
			
			\nextgroupplot[legend style = {column sep = 7pt, legend to name = LegendMon111}]
			\addplot table[x = beta,y=Horseshoe] from \HRMSEpfnoheterSigmaone;
			\addplot table[x = beta,y=DebiasHorseshoe] from \HRMSEpfnoheterSigmaone;
			\addplot table[x = beta,y=DebiasLasso] from \HRMSEpfnoheterSigmaone;
			\node[anchor=north] at (axis description cs: 0.25,  0.95) {\fontsize{5}{4}\selectfont \shortstack{
					\\  \\
					Hetero\\
					$p=50$
			}};

			% Monotonicity for the bivariate case
			\nextgroupplot
			\node[anchor=north] at (axis description cs: 0.25,  0.95) {\fontsize{5}{4}\selectfont \shortstack{
					\\  \\
					Homo-Normal\\
					$p=100$
			}};
			\addplot table[x = beta,y=Horseshoe] from \HRMSEponohomonSigmaone;
			\addplot table[x = beta,y=DebiasHorseshoe] from \HRMSEponohomonSigmaone;
			\addplot table[x = beta,y=DebiasLasso] from \HRMSEponohomonSigmaone;

			\nextgroupplot
			\node[anchor=north] at (axis description cs: 0.25,  0.95) {\fontsize{5}{4}\selectfont \shortstack{
					\\  \\
					Homo-Chi\\
					$p=100$
			}};
			\addplot table[x = beta,y=Horseshoe] from \HRMSEponohomochiSigmaone;
			\addplot table[x = beta,y=DebiasHorseshoe] from \HRMSEponohomochiSigmaone;
			\addplot table[x = beta,y=DebiasLasso] from \HRMSEponohomochiSigmaone;
			
			\nextgroupplot[legend style = {column sep = 3.5pt, legend to name = LegendMon211}]
			
			\node[anchor=north] at (axis description cs: 0.25,  0.95) {\fontsize{5}{4}\selectfont \shortstack{
					\\  \\
					Hetero\\
					$p=100$
			}};
			\addplot table[x = beta,y=Horseshoe] from \HRMSEponoheterSigmaone;
			\addplot table[x = beta,y=DebiasHorseshoe] from \HRMSEponoheterSigmaone;
			\addplot table[x = beta,y=DebiasLasso] from \HRMSEponoheterSigmaone;
			
			\nextgroupplot
			\node[anchor=north] at (axis description cs: 0.25,  0.95) {\fontsize{5}{4}\selectfont \shortstack{
					\\  \\
					Homo-Normal\\
					$p=200$
			}};
			
			\addplot table[x = beta,y=Horseshoe] from \HRMSEptnohomonSigmaone;
			\addplot table[x = beta,y=DebiasHorseshoe] from \HRMSEptnohomonSigmaone;
			\addplot table[x = beta,y=DebiasLasso] from \HRMSEptnohomonSigmaone;

			\nextgroupplot
			\node[anchor=north] at (axis description cs: 0.25,  0.95) {\fontsize{5}{4}\selectfont \shortstack{
					\\  \\
					Homo-Chi\\
					$p=200$
			}};
			\addplot table[x = beta,y=Horseshoe] from \HRMSEptnohomochiSigmaone;
			\addplot table[x = beta,y=DebiasHorseshoe] from \HRMSEptnohomochiSigmaone;
			\addplot table[x = beta,y=DebiasLasso] from \HRMSEptnohomochiSigmaone;
			
			\nextgroupplot[legend style = {column sep = 3.5pt, legend to name = LegendMon311}]
			
			\node[anchor=north] at (axis description cs: 0.25,  0.95) {\fontsize{5}{4}\selectfont \shortstack{
					\\  \\
					Hetero\\
					$p=200$
			}};
			
			\addplot table[x = beta,y=Horseshoe] from \HRMSEptnoheterSigmaone;
			\addplot table[x = beta,y=DebiasHorseshoe] from \HRMSEptnoheterSigmaone;
			\addplot table[x = beta,y=DebiasLasso] from \HRMSEptnoheterSigmaone;
			\addlegendentry{Bayes}
			\addlegendentry{Posterior-Corrected}
			\addlegendentry{Debiased-LASSO}

		\end{groupplot}
		\node at ($(myplots c2r1) + (0,-2.25cm)$) {\ref{LegendMon111}};
		\node at ($(myplots c2r2) + (0,-2.25cm)$) {\ref{LegendMon211}};
		\node at ($(myplots c2r3) + (0,-2.25cm)$) {\ref{LegendMon311}};
	\end{tikzpicture}
	\caption{RMSE corresponding to different values of $\beta_0$ under settings S1 (first column), S2 (second column), and S3 (third column).}\label{Fig:hfive}
\end{figure}

\begin{figure}[!ht]
	\centering\scriptsize
	\begin{tikzpicture} % columns*rows
		\begin{groupplot}[group style={group name=myplots,group size=3 by 3,horizontal sep= 0.8cm,vertical sep=1.1cm},
			grid = minor,
			width = 0.375\textwidth,
			xmax=5,xmin=0,
			ymax=1,ymin=0,
			every axis title/.style={below,at={(0.2,0.8)}},
			xlabel=$\beta_0$,
			x label style={at={(axis description cs:0.95,0.04)},anchor=south},
			xtick={0,1,2,3,4,5},
			ytick={0,0.5,1,1.5},
			tick label style={/pgf/number format/fixed},
			legend style={text=black,cells={align=center},row sep = 3pt,legend columns = -1, draw=none,fill=none},
			cycle list={%
				%{smooth,tension=0.5,color=NavyBlue, no markers,line width=0.25pt, densely dotted}, % alpha
				{smooth,tension=0,color=black, mark=halfsquare*,every mark/.append style={rotate=270},mark size=1.5pt,line width=0.5pt},% LSW-Op
				{smooth,tension=0,color=blue, mark=halfsquare*,every mark/.append style={rotate=90},mark size=1.5pt,line width=0.5pt}, % LSW-Un
				{smooth,tension=0,color=red, mark=10-pointed star,mark size=1.5pt,line width=0.5pt},% Chetverikov
				{smooth,tension=0,color=RoyalBlue1, mark=halfcircle*,every mark/.append style={rotate=90},mark size=1.5pt,line width=0.5pt}, % Quad,Knots 0/// Cub,Knots 3
				{smooth,tension=0,color=RoyalBlue2, mark=halfcircle*,every mark/.append style={rotate=180},mark size=1.5pt,line width=0.5pt}, % Quad,Knots/// Cub,Knots 5
				{smooth,tension=0,color=RoyalBlue3, mark=halfcircle*,every mark/.append style={rotate=270},mark size=1.5pt,line width=0.5pt},% Cub,Knots 0/// Cub,Knots 7
				{smooth,tension=0,color=RoyalBlue4, mark=halfcircle*,every mark/.append style={rotate=360},mark size=1.5pt,line width=0.5pt},% Cub,Knots 1
			}
			]
			% Monotonicity for the univariate case
			\nextgroupplot
			\node[anchor=north] at (axis description cs: 0.25,  0.95) {\fontsize{5}{4}\selectfont \shortstack{
					\\  \\
					Homo-Normal\\
					$p=50$
			}};
			\addplot table[x = beta,y=Horseshoe] from \HRMSEpfnohomonSigmatwo;
			\addplot table[x = beta,y=DebiasHorseshoe] from \HRMSEpfnohomonSigmatwo;
			\addplot table[x = beta,y=DebiasLasso] from \HRMSEpfnohomonSigmatwo;
			
			\nextgroupplot

			\node[anchor=north] at (axis description cs: 0.25,  0.95) {\fontsize{5}{4}\selectfont \shortstack{
					\\  \\
					Homo-Chi\\
					$p=50$
			}};
			\addplot table[x = beta,y=Horseshoe] from \HRMSEpfnohomochiSigmatwo;
			\addplot table[x = beta,y=DebiasHorseshoe] from \HRMSEpfnohomochiSigmatwo;
			\addplot table[x = beta,y=DebiasLasso] from \HRMSEpfnohomochiSigmatwo;
			
			\nextgroupplot[legend style = {column sep = 7pt, legend to name = LegendMon112}]
			\addplot table[x = beta,y=Horseshoe] from \HRMSEpfnoheterSigmatwo;
			\addplot table[x = beta,y=DebiasHorseshoe] from \HRMSEpfnoheterSigmatwo;
			\addplot table[x = beta,y=DebiasLasso] from \HRMSEpfnoheterSigmatwo;
			\node[anchor=north] at (axis description cs: 0.25,  0.95) {\fontsize{5}{4}\selectfont \shortstack{
					\\  \\
					Hetero\\
					$p=50$
			}};

			% Monotonicity for the bivariate case
			\nextgroupplot
			\node[anchor=north] at (axis description cs: 0.25,  0.95) {\fontsize{5}{4}\selectfont \shortstack{
					\\  \\
					Homo-Normal\\
					$p=100$
			}};
			\addplot table[x = beta,y=Horseshoe] from \HRMSEponohomonSigmatwo;
			\addplot table[x = beta,y=DebiasHorseshoe] from \HRMSEponohomonSigmatwo;
			\addplot table[x = beta,y=DebiasLasso] from \HRMSEponohomonSigmatwo;

			\nextgroupplot
			\node[anchor=north] at (axis description cs: 0.25,  0.95) {\fontsize{5}{4}\selectfont \shortstack{
					\\  \\
					Homo-Chi\\
					$p=100$
			}};
			\addplot table[x = beta,y=Horseshoe] from \HRMSEponohomochiSigmatwo;
			\addplot table[x = beta,y=DebiasHorseshoe] from \HRMSEponohomochiSigmatwo;
			\addplot table[x = beta,y=DebiasLasso] from \HRMSEponohomochiSigmatwo;
			
			\nextgroupplot[legend style = {column sep = 3.5pt, legend to name = LegendMon212}]
			
			\node[anchor=north] at (axis description cs: 0.25,  0.95) {\fontsize{5}{4}\selectfont \shortstack{
					\\  \\
					Hetero\\
					$p=100$
			}};
			\addplot table[x = beta,y=Horseshoe] from \HRMSEponoheterSigmatwo;
			\addplot table[x = beta,y=DebiasHorseshoe] from \HRMSEponoheterSigmatwo;
			\addplot table[x = beta,y=DebiasLasso] from \HRMSEponoheterSigmatwo;
			
			\nextgroupplot
			\node[anchor=north] at (axis description cs: 0.25,  0.95) {\fontsize{5}{4}\selectfont \shortstack{
					\\  \\
					Homo-Normal\\
					$p=200$
			}};
			
			\addplot table[x = beta,y=Horseshoe] from \HRMSEptnohomonSigmatwo;
			\addplot table[x = beta,y=DebiasHorseshoe] from \HRMSEptnohomonSigmatwo;
			\addplot table[x = beta,y=DebiasLasso] from \HRMSEptnohomonSigmatwo;

			\nextgroupplot
			\node[anchor=north] at (axis description cs: 0.25,  0.95) {\fontsize{5}{4}\selectfont \shortstack{
					\\  \\
					Homo-Chi\\
					$p=200$
			}};
			\addplot table[x = beta,y=Horseshoe] from \HRMSEptnohomochiSigmatwo;
			\addplot table[x = beta,y=DebiasHorseshoe] from \HRMSEptnohomochiSigmatwo;
			\addplot table[x = beta,y=DebiasLasso] from \HRMSEptnohomochiSigmatwo;
			
			\nextgroupplot[legend style = {column sep = 3.5pt, legend to name = LegendMon312}]
			
			\node[anchor=north] at (axis description cs: 0.25,  0.95) {\fontsize{5}{4}\selectfont \shortstack{
					\\  \\
					Hetero\\
					$p=200$
			}};
			
			\addplot table[x = beta,y=Horseshoe] from \HRMSEptnoheterSigmatwo;
			\addplot table[x = beta,y=DebiasHorseshoe] from \HRMSEptnoheterSigmatwo;
			\addplot table[x = beta,y=DebiasLasso] from \HRMSEptnoheterSigmatwo;
			\addlegendentry{Bayes}
			\addlegendentry{Posterior-Corrected}
			\addlegendentry{Debiased-LASSO}

		\end{groupplot}
		\node at ($(myplots c2r1) + (0,-2.25cm)$) {\ref{LegendMon112}};
		\node at ($(myplots c2r2) + (0,-2.25cm)$) {\ref{LegendMon212}};
		\node at ($(myplots c2r3) + (0,-2.25cm)$) {\ref{LegendMon312}};
	\end{tikzpicture}
	\caption{RMSE corresponding to different values of $\beta_0$ under settings S1 (first column), S2 (second column), and S3 (third column).}\label{Fig:hsix}
\end{figure}

\end{appendices}

\end{document}